\documentclass[11pt]{article}
\RequirePackage[colorlinks,citecolor=blue,urlcolor=blue,linkcolor=blue]{hyperref}
\usepackage{url}

\DeclareSymbolFont{slenderlargesymbols}{OMX}{ccex}{m}{n}
\DeclareMathSymbol{\prod}{\mathop}{slenderlargesymbols}{"51}
\usepackage{latexsym,amssymb,amsmath,amsthm,graphics,graphicx,float,psfrag, epsfig, color, authblk, enumerate}

\usepackage{pgfplots}
\usepackage{url,amssymb,amsmath,amsthm,amscd,paralist,bbm,wrapfig}
\usepackage{bm}
\usepackage{algorithm}
\usepackage{algorithmic} %//format of the algorithm 

\usepackage{mathtools}

\usepackage[T1]{fontenc}
\usepackage[numbers]{natbib}
\tikzstyle{vertex}=[circle,black, fill=black, draw, inner sep=0pt, minimum size=6pt]

\usepackage{tikz}
\usetikzlibrary{decorations.pathreplacing}
\usepackage{xcolor}
\usepackage{mathtools}
\definecolor{cof}{RGB}{219,144,71}
\definecolor{pur}{RGB}{186,146,162}
\definecolor{greeo}{RGB}{91,173,69}
\definecolor{greet}{RGB}{52,111,72}
\pgfplotsset{compat=1.14}

\newtheorem{innercustomgeneric}{\customgenericname}
\providecommand{\customgenericname}{}
\newcommand{\newcustomtheorem}[2]{%
  \newenvironment{#1}[1]
  {%
   \renewcommand\customgenericname{#2}%
   \renewcommand\theinnercustomgeneric{##1}%
   \innercustomgeneric
  }
  {\endinnercustomgeneric}
}

\newcustomtheorem{customasmp}{Assumptions}
\renewcommand{\epsilon}{\varepsilon}
\newcommand{\R}{\mathbb{R}}
\newcommand{\C}{\mathbb{C}}
\newcommand{\N}{\mathbb{N}}
\newcommand{\al}{\alpha}
\newcommand{\PiWdix}{\Pi_{i=d}^1 W_{i,+,x}}

\newcommand{\relu}{\text{relu}}
\newcommand{\E}{\operatorname{\mathbb{E}}}
\newcommand{\Pro}{\operatorname{\mathbb{P}}}
\newcommand{\one}{\operatorname{\mathbbm{1}}}
\newcommand{\sgn}{\operatorname{sgn}}
\newcommand{\sign}{\operatorname{sgn}}
\newcommand{\diag}{\operatorname{diag}}
\newcommand{\dimension}{\operatorname{dim}}

\newcommand{\Seps}{\mathcal{S}}
\newcommand{\G}{G}

\newcommand{\F}{\mathcal{F}} 
\renewcommand{\top}{{\mathrm{T}}}

\newtheorem{thm}{Theorem}
\newtheorem{lem}{Lemma}
\newtheorem{cor}[thm]{Corollary}
\newtheorem{prop}{Proposition}

\newtheorem{defn}[thm]{Definition}

\theoremstyle{remark}

\title{Compressive Phase Retrieval: Optimal Sample Complexity with Deep Generative Priors}

\author{Paul Hand\thanks{Department of Mathematics and Khoury College of Computer and Information Science, Northeastern University, Boston, MA}, Oscar Leong\thanks{Department of Computational and Applied Mathematics, Rice University, Houston, TX}, and Vladislav Voroninski\thanks{Helm.ai, Menlo Park, CA}}

\begin{document}

\maketitle
%\author{Paul Hand\thanks{Department of Mathematics and College of Computer and Information Science, Northeastern University}, Oscar Leong\thanks{Department of Computational and Applied Mathematics, Rice University}, and Vladislav Voroninski\thanks{Helm.ai, Menlo Park, CA}}

%    Abstract is required.
\begin{abstract}

Advances in compressive sensing provided reconstruction algorithms of sparse signals from linear measurements with optimal sample complexity, but natural extensions of this methodology to nonlinear inverse problems have been met with potentially fundamental sample complexity bottlenecks. In particular, tractable algorithms for compressive phase retrieval with sparsity priors have not been able to achieve optimal sample complexity.  This has created an open problem in compressive phase retrieval: under generic, phaseless linear measurements, are there tractable reconstruction algorithms that succeed with optimal sample complexity? Meanwhile, progress in machine learning has led to the development of new data-driven signal priors in the form of generative models, which can outperform sparsity priors with significantly fewer measurements.   %there has been tremendous progress using deep learning techniques to learn image priors in the form of generative models given by deep neural networks which, when used in compressive sensing, can outperform sparsity priors in some cases. 
 In this work, we resolve the open problem in compressive phase retrieval and demonstrate that generative priors can lead to a fundamental advance by permitting optimal sample complexity by a tractable algorithm in this challenging nonlinear inverse problem. We additionally provide empirics showing that exploiting generative priors in phase retrieval can significantly outperform sparsity priors. 
 These results provide support for generative priors as a new paradigm for signal recovery in a variety of contexts, both empirically and theoretically.  The strengths of this paradigm are that (1) generative priors can represent some classes of natural signals more concisely than sparsity priors, (2) generative priors allow for direct optimization over the natural signal manifold, which is intractable under sparsity priors, and (3) the resulting non-convex optimization problems with generative priors can admit benign optimization landscapes at optimal sample complexity, perhaps surprisingly, even in cases of nonlinear measurements.

\end{abstract}

\maketitle

%\tableofcontents

%%%%%%%%%%%%%%%%%%%

\section{Introduction}

%\subsection{New Intro}

 The study of inverse problems pervades virtually all of the natural sciences including biological and astronomical imaging, X-ray crystallography, oil exploration, and shape optimization and reconstruction. An object of interest is observed via some forward mapping process, and the task is to recover the object, often subject to ill-posedness and noise. In order to increase fidelity of the estimate or decrease the number of required measurements, one can enforce structural assumptions or priors on the signal, a practice dating as far back as Tikhonov regularization \cite{TikReg} and the Nyquist sampling theorem \cite{NyquistShannonSampling}. A canonical example of an ill-posed inverse problem in the field of imaging is compressive sensing (CS), in which one aims to recover a signal from undersampled linear measurements. By exploiting the sparsity of natural images in the wavelet domain as a structural prior, CS has led to a number of practical developments across the imaging sciences, such as speeding up some forms of MRI imaging by an order of magnitude \cite{Dono2007}.

In terms of theory, advances in CS have provided reconstruction algorithms using sparsity priors with information theoretically optimal sample complexity \cite{Tao2006,Donoho2006}. A seminal result in the field states that if given $m < n$ undersampled linear measurements $b_* = Ay_*$ where $A \in \R^{m \times n}$ has i.i.d. Gaussian entries and $y_* \in \R^n$ is an $s$-sparse signal, recovery is guaranteed with high probability when $m = O(s \log n)$ by solving the following convex program:
\begin{align*}
    \min_{y \in \R^n} \|y\|_1\ \text{s.t.}\ Ay = b_*.
\end{align*}

\noindent 
%CS and the notion of convex relaxations led to many other developments, including matrix completion, . 
The success of compressive sensing has popularized the notion of signal sparsity throughout the imaging sciences, resulting in sparsity becoming a common choice as a structural prior.

Sparsity-based priors when applied to nonlinear inverse problems such as phase retrieval have been met with potentially fundamental sample complexity bottlenecks. In phase
retrieval, a signal $y_* \in \R^n$ or $\C^n$ is to be estimated from observations $|\langle a_i,y_*\rangle|^2$, $i = 1,2,\dots m$. Compressive phase retrieval considers the case $m<n$, which requires structural priors to enable recovery. While an $s$-sparse signal is information theoretically recoverable from $O(s \log n)$ generic phaseless measurements, compressive phase retrieval  algorithms have not achieved sample complexity below $O(s^2 \log n)$ \cite{SparsePRoverview}. In fact, convex algorithms such as PhaseLift \cite{CandStrohVoron2013,CandLi2014}, 
provably fail below $O(s^2 \log n)$ measurements under natural extensions to incorporate sparsity \cite{Voroninski2013, Oymaketal15}. This has created an open problem in compressive phase retrieval to find a computationally efficient algorithm to reconstruct signals from generic, phaseless linear measurements with optimal sample complexity with respect to the signal's intrinsic dimensionality. Furthermore, there is evidence to support that these sample complexity limitations may be fundamental for sparse phase retrieval. In the closely related sparse PCA problem, a reduction from planted clique was found, and it is widely conjectured to be NP-hard \cite{BerthetRigollet13, SOS_planted_clique}.  These observations open the question of whether other signal priors may successfully achieve sample-optimal reconstruction algorithms.

Simultaneously, there has been tremendous progress on priors in the form of generative models given by a deep neural network, which in some cases significantly outperforms sparsity priors at compressive sensing. These generative models, such as Generative Adversarial Networks \cite{Goodfellow2014} and Variational Autoencoders \cite{Kingma2014}, learn an explicit mapping from a low-dimensional latent space $\R^k$ to an approximation of the natural image manifold in $\R^n$ and can be trained on datasets of various natural signal classes to create  realistic, yet synthetic samples of human faces \cite{Karras2019}, MRIs \cite{MRI_GAN}, cells \cite{Cell_GAN}, human fingerprints \cite{FingerGAN}, and more. Enforcing a generative prior in CS tasks by directly optimizing over the latent space has been shown to outperform sparsity-based methods such as Lasso by 5-10x fewer measurements \cite{Price2017} in some cases. Moreover, while the optimization problem posed over latent space is non-convex, \cite{HV2017} showed that when the number of measurements $m$ is proportional to $k$ up to log factors, the empirical risk minimization problem under a suitable random generator model exhibits favorable global geometry in the sense that there are no spurious local minima away from small neighborhoods of the true solution and a negative multiple thereof. %Finally, because they directly model a nonlinear signal manifold, generative priors can lower sample complexity as these models can be more compressible than those given by a sparsity prior. That is, the latent dimensionality of a generator for a class of images can be much lower than the sparsity level of such images with respect to a wavelet basis.

The above empirical and theoretical evidence indicates that generative neural networks can potentially succeed as structural priors in nonlinear inverse problems where previous methods exploiting sparsity have thus far been met with likely fundamental bottlenecks.

In this work, we resolve the open problem in compressive phase retrieval by presenting a computationally efficient algorithm that achieves optimal sample complexity with generic measurements under a generative prior. In particular, we consider a deep generative prior for compressive phase retrieval by supposing that the desired signal lives in the range of a feed-forward neural network with ReLU activation functions and latent code dimensionality $k$. We establish the sufficiency of two deterministic conditions on the weights of the generative model and the measurement matrix to guarantee that the signal can be recovered by a subgradient descent algorithm. Moreover, we show that these conditions are satisfied with high probability for Gaussian weights and generic Gaussian measurements as soon as $m$ is proportional to $k$, up to log factors, which is information theoretically optimal in $k$. In addition to our theoretical results, we empirically establish that exploiting generative models in phase retrieval tasks can significantly outperform sparsity-based methods.

%Our results come with a surprise, in that  These results add to the literature of instances in which an optimization problem exhibits a type of benign non-convexity. 
%Traditionally, rigorous understandings of optimization tasks for inverse problems are typically posed in the convex setting.  
%More recently, 

Subsequent to the publication of preliminary versions of the results of this paper \cite{HLV18}, generative priors have also been shown to break through sample complexity barriers in PCA. In particular, all known algorithms to achieve optimal statistical sample complexity in sparse PCA are computationally intractable and all known polynomial time algorithms exhibit a sub-optimal quadratic sample complexity on the sparsity of the true signal \cite{Krauthgameretal15,DeshpandeMontanari14}. Gaps of this nature have also been observed in a number of related problems \cite{Decelleetal11,RichardMontanari14}. However, with respect to PCA, recent work in both the asymptotic \cite{SpikedGenPrior19} and non-asymptotic regimes \cite{CHV20} have shown that the low rank matrix recovery problem with generative priors does not exhibit a computational-to-statistical gap, offering further evidence of the benefit of generative priors in inverse problems.

The results in the present work provide  empirical and theoretical support to the notion that deep generative priors offer a new paradigm for signal recovery that offers fundamental advances.  In this paradigm, a model of a natural signal class is learned from data in the form of a generative model. The generative model directly parameterizes a low-dimensional signal manifold, and recovering a signal subject to noisy measurements can be posed as a direct optimization problem whose search space is restricted to the range of the generative model.  This paradigm has several strengths in comparison to sparsity priors.  First, generative models may provide better compression of natural signals than sparsity priors.  Precisely, the dimensionality of the manifold modeling the natural signal class under a generative prior may be lower than the sparsity level of the same signals. Second, generative priors allow for direct optimization over the natural signal manifold.  In contrast, sparsity priors give rise to combinatorial optimization problems which can not be directly solved.  Tractable convex relaxations have not been successful in important nonlinear settings.  Third, the non-convex optimization problems under generative priors can admit benign optimization landscapes at optimal sample complexity even in the case of nonlinear measurements.  This fundamental advance has so far not been realized by sparsity priors.

\subsection{Related Work}

%\subsubsection{Phase Retrieval}

\paragraph{Phase Retrieval:} Some of the earliest methods to solve phase retrieval tasks are the non-convex alternating minimization Gerchberg-Saxton \cite{GerchSaxAlg} and Fienup \cite{FienupAlg} algorithms.  Recently, a variety of methods have been introduced that enjoy theoretical guarantees.  Convex methods, such as the seminal lifting-based approach PhaseLift \cite{CandStrohVoron2013}, can achieve optimal sample complexity for unstructured signals \cite{CandLi2014}.  Further recovery guarantees have been extended to non-convex formulations such as Wirtinger Flow \cite{Candes2015, Wright2016, Mahdi2019} and its non-smooth variant Amplitude Flow \cite{Wang2017,Mahdi2019,AmpFlow}.  Other approaches include Phasemax \cite{Phasemax, BahmaniRomberg_LinProg}, Phasecut \cite{Phasecut}, AltMinPhase \cite{Jain2013}, and alternating projection methods \cite{Waldspurger_Alt_Proj}. %In recent years, more direct non-convex optimization approaches such as the seminal Wirtinger Flow approach \cite{Candes2015} directly analyze a gradient descent scheme with spectral initialization. Further global analyses in \cite{Wright2016} ruled out adversarial geometries of the non-convex Wirtinger Flow objective. Other works \cite{Wang2017,Mahdi2019,AmpFlow} have studied Amplitude Flow, a non-smooth variation of Wirtinger Flow, which has been shown to empirically outperform Wirtinger Flow \cite{Yeh2015}. Lifting-based techniques, made popular by the seminal approach in \cite{CandStrohVoron2013}, view the quadratic measurements as linear measurements of a rank-1 matrix and can achieve $m=O(n)$ sample complexity \cite{CandLi2014}. Other convex methodologies include linear programming-based approaches, such as PhaseMax \cite{Phasemax} and other semidefinite programming approaches such as Phasecut \cite{Phasecut}.

Since the success of exploiting sparsity in linear compressed sensing, many works have attempted to leverage similar techniques to solve the phase retrieval problem in the compressive setting $m < n$. When the $n$-dimensional signal is $s$-sparse, the information theoretic lower bound of $m = O(s \log n)$ measurements was shown to be required for the injectivity of phaseless Gaussian measurements \cite{VladSRIP2014}. However, attempts at achieving this optimal sample complexity via a polynomial time algorithm have proven quite difficult and, in some cases, impossible. For example, the natural $\ell_1$-penalized variant of Phaselift was shown to be able to recover an $s$-sparse signal with $O(s^2 \log n)$ generic measurements, but this bound was also proven to be tight \cite{Voroninski2013, Oymaketal15}. Moreover, there are a number of results that show, if one were able to construct a sufficiently accurate initializer of the true solution, then recovery from $O(s \log n)$ Gaussian measurements is possible by a variety of methods \cite{HV2016,SPARTA,Mahdi2019}. Known initialization schemes to accomplish this, however, require $O(s^2 \log n)$ measurements \cite{Cai2016}. For a more complete discussion of prior methodologies for phase retrieval, we refer the reader to \cite{Numerics_of_PR}.

Some existing works in compressive phase retrieval establish optimal sample complexity recovery guarantees under non-generic measurements \cite{SparsePRoverview}.  For example, \cite{Romberg2015} showed that assuming the measurement vectors were chosen from an incoherent subspace, then recovery is possible with $O(s \log \frac{n}{s})$ measurements. Also, using the notion of polarization, \cite{Bandeira_Mixon} showed that $O(s \log n)$ measurements also suffices for recovery when the measurement vectors have an associated graph with sufficient connectivity properties. However, these results would be difficult to generalize to the experimental setting as their measurement design architectures are often unrealistic, with generic measurements offering a closer model to the goal of Fourier diffraction  measurements.

%\paragraph{Non-convex optimization:} With the variety of success of non-convex approaches across applied mathematics and computer science, there has been a recent push towards a deeper theoretical understanding of non-convex optimization problems. In this present work, we show that our non-convex formulation in fact exhibits global benign geometry that allows it to be sufficiently optimized with first order methods. A number of different problems have exhibited a similar type of global benign geometry such as Burer-Monteiro factorization approaches for semidefinite programming \cite{Boumaletal2018}, phase retrieval \cite{Wright2016} and synchronization and community detection \cite{BBV16}. Rigorous guarantees in non-convex settings have also been explored in subspace recovery \cite{Maunuetal17}, blind deconvolution \cite{Lietal2016}, convergence guarantees in deep learning \cite{AllenZhuetal22018, Duetal2018, OymakSoltan2018, OymakSoltan2019} and in many other areas.

%\subsubsection{Signal recovery with generative priors}

\paragraph{Signal recovery with generative priors:} In \cite{Price2017}, the authors studied  enforcing a generative prior in the linear compressive sensing regime. In particular, given $m$ linear measurements $Ay_*$ where $y_* \in \R^n$, the authors modelled natural signals as being in the range of a trained  generative model $G : \R^k \rightarrow \R^n$ where $k \ll n$. To solve the inverse problem, they proposed to find a latent code $x_* \in \R^k$ such that $G(x_*) \approx y_*$ by solving the following least squares objective \begin{align}
\min_{x \in \R^k} \frac{1}{2} \Big\|AG(x) - Ay_*\Big\|^2. \label{bora_obj}
\end{align} They provided empirical evidence showing that 5-10x fewer measurements were needed to achieve comparable reconstruction errors, compared to standard sparsity-based approaches such as Lasso in some parameter regimes. Based on the success of generative priors in compressive sensing, a number of followup works have considered a similar setup in a variety of inverse problems, ranging from compressed sensing \cite{HV2016, Huangetal2018, ShahDCS2018, ADMMGenPrior, Songetal19}, denoising \cite{Heckeletal2018}, phase retrieval \cite{HLV18, Shamshad2018}, low-rank matrix recovery \cite{CHV20, SpikedGenPrior19}, one-bit compressive sensing \cite{Qiuetal19, Liuetal20}, blind deconvolution \cite{Shamshad2018, JoshiHand19}, and more.
This framework, in the case of compressive sensing, enjoys multiple theoretical analyses. A subset of the authors in \cite{HV2016} presented the first global landscape analysis of the empirical risk minimization problem and showed that, in fact, when the network is sufficiently expansive with Gaussian weights and the number of measurements is proportional to $k$, there exists a descent direction everywhere outside of potentially two small neighborhoods of the minimizer and true solution. Followup work later established convergence guarantees of first order methods in compressed sensing \cite{Huangetal2018} and denoising \cite{Heckeletal2018} under similar statistical assumptions on the generator. In the present work, we consider precisely the same random model in the context of compressive phase retrieval. %For example, in \cite{Price2017} the authors establish that if \eqref{bora_obj} can be solved, then the signal is recovered up to the noise level in the measurements, representation error, and optimization error. Under an appropriate random model for trained generative priors, subgradient methods can provably converge to the signal of interest provided the network is sufficiently expansive at each layer and that the number of measurements is on the order of the latent dimensionality of the generator \cite{HV2016, Huangetal2018, Heckeletal2018}. In the present work, we consider precisely the same random model in the context of compressive phase retrieval.

\subsection{Compressive Phase Retrieval with Generative Priors} \label{intro_of_prob_section_w_detcond}

The compressive phase retrieval problem is as follows.  We consider the real-valued version out of simplicity.  Consider a signal $y_* \in \R^n$.  Given $m$ phaseless linear measurements of the form 
\begin{align}
b_* = |Ay_*| + \eta \label{eqn:measurements-with-noise}
\end{align} where $m<n$, $A \in \R^{m \times n}$ is a known linear operator, and $\eta \in \R^m$ denotes measurement noise, the goal is to recover $y_*$ from knowledge of $b_*$ and $A$.  As $m<n$, additional structure must be exploited to accurately estimate $y_*$.  In this work, we assume that $y_*$ belongs in or near the range of a trained generative model $\G : \R^k \rightarrow \R^{n}$.  That is, $y_* \approx \G(x_*)$ for some latent code $x_*$.  In order to recover an estimate of a signal $y_*$, it suffices to recover $x_*$ and then compute $\G(x_*)$. We propose to solve the following nonlinear least squares problem:\begin{align}
\min_{x \in \R^k} f(x) := \frac{1}{2} \Big\| |A\G(x)| - b_* \Big\|^2 \label{dpr_objective}.
\end{align} 

This formulation attempts to find the signal in the range of the generative model $\G$ that is most consistent with provided measurements in a particular sense.  It is motivated both by the non-convex generative modeling formulation for compressed sensing in \cite{Price2017} with an  Amplitude Flow \cite{AmpFlow} perspective from phase retrieval and was originally introduced by the present authors in \cite{HLV18}. As the underlying optimization problem is posed over an explicitly parameterized $k$-dimensional manifold where $k \ll n$, compressive phase retrieval may be possible from $m = \Omega(k) \ll n$ measurements. %It allows the possibility for compressive phase retrieval because the underlying optimization is over an explicitly parameterized $k$-dimensional manifold where $k \ll n$, potentially rendering signal recovery possible from $m = \Omega(k) \ll n$ measurements.  %When $k \ll n$, one would hope that signal recovery is possible with $m = \Omega(k) \ll n$ measurements.  

In this paper, we prove that \eqref{dpr_objective} can be solved with sample complexity proportional to $k$, under an appropriate model for $\G$ and a generic measurement model.  This theoretical result is in stark contrast to algorithms for compressive phase retrieval based on sparsity priors, where no known tractable algorithm achieves information theoretically optimal sample complexity under a generic measurement model. This result extends the work of \cite{Heckeletal2018, Huangetal2018} which established similar algorithmic guarantees for recovery with generative priors in the linear measurement regimes of denoising and compressed sensing.
 Additionally, we provide experimental results that \eqref{dpr_objective} can outperform sparsity-based compressive phase retrieval algorithms in the presence of a trained $\G$ from standardly available datasets.

%If $k \ll s$, where $s$ is the approximate sparsity of the signal $y_*$ with respect to a particular basis, this formulation may even outperform sparse phase retrieval methods.
%is possible under this formulation as the optimization \eqref{dpr_objective}  is over an explicit  $k$-dimensional space, where $k < n$.  
%
%Estimating $y_*$ from $m \ll n$ measurements provided $k \ll n$
%
% This formulation allows compressive phase retrieval
% 
% This formulation attempts to find the signal in the range of the generative model $\G$ that is most consistent with provided measurements in the provided sense.  As the range of $\G$ is a $k$-dimensional manifold for $k<n$, 
% 
%      Solving the compressive ph
% 
%     Formulation \eqref{dpr_objective} exploits the property that $\G$ provides an explicit parameterization of a $k$-dimensional  manifold in which the signal is presumed to belong.  It attempts to find the best signal in the range of the generative model in the sense of a $\ell_2$ loss on unsquared measurements.  In the case that $k \ll n$, there is hope that compressive phase retrieval is possible for $m$ sxx.  Furthermore, if $k \ll s$, where $s$ is the approximate sparsity of the signal $y_*$ with respect to a particular basis, there is hope that this formulation may outperform sparse phase retrieval methods.

The formulation assumes that the generative model $\G$ is already known.  In practice, it typically is a neural network whose parameters (weights) are learned from a large collection of training images belonging to a particular natural signal class.  The field of generative modeling has demonstrated multiple types of neural networks which can be effectively trained (e.g. Variational Autoencoders \cite{Kingma2014} and Generative Adversarial Networks \cite{Goodfellow2014}).  The dimensionality $k$ of the latent codes is fixed at training time and its value is selected in order to balance multiple effects; for example, the range of $\G$ should be large enough to approximately include all of the desired signal class, and the image representations should be as concise as possible. A particular image of interest may not be \textit{exactly} in the range of a trained model $\G$ because the model has \textit{representation error}, but this error is expected to become smaller as techniques for training generative models improve.

\subsection{Deterministic and Probabilistic Models for Generative Priors}

In order to establish recovery guarantees for phase retrieval with generative priors, we will assume a neural network architecture and a model for the weights of the network once trained.  Our intention is to analyze a model which is realistic enough to describe trained nets, yet tractable enough to permit rigorous analysis of sample complexity for a convergent optimization algorithm.  To achieve both of these objectives we consider the following models.  We assume that the generative model $\G : \R^k \rightarrow \R^n$ is given by a $d$-layer feedforward neural network with ReLU activation functions and no bias terms. Specifically, we assume that \begin{align}
\G(x) = \relu(W_d\dots\relu(W_2\relu(W_1x))\dots) \label{Generator_formula}
\end{align}
where $\relu(x) := \max(x,0)$ acts entrywise and each $W_i \in \R^{n_i \times n_{i-1}}$ for $i \in [d]$ with $k = n_0 < n_1 < n_2 < \dots < n_d = n.$  Each matrix $W_i$ corresponds to the neural network weights of the $i$-th layer, and the $j$-th row of $W_i$ are the weights of the $j$-th neuron in the $i$-th layer.

We will assume an \textit{expansive-Gaussian} probabilistic model for the weights of $\G$.  That is, $n_i$ increases sufficiently with $i$, and the weights within each layer are i.i.d. Gaussians.  This model was introduced by a subset of the authors in \cite{HV2017}.  We additionally assume a Gaussian model of the measurement matrix $A$.  The justification for these assumptions is as follows. The Gaussianicity of $A$ ensures that measurements are suitably generic, and, indeed, achieving optimal sample complexity in sparse phase retrieval has not been attained for this measurement model. Regarding the expansivity assumption, we note that generative models with low dimensional latent spaces are inherently expansive when considered as a whole.  In a sense, the network and each layer therein could be viewed as adding redundancy to a more compact representation, though in practice some successful network architectures do not have strict layerwise expansivity.  Regarding the Gaussianicity model of neural network weights, it has been shown  that neural networks, such as AlexNet, trained on real data have resulting weights with statistics similar to Gaussians \cite{GaussNets}. 
% The expansive-Gaussian model is justified as fol multifold.  As the generative models we consider exploit learned low-dimensional representations of data, their overall architecture has larger output dimension and input dimension and are inherently expansive as a whole.  While some networks can succeed without having strict expansivity in each layer, we assume layerwise expansivity.  Regarding Gaussianicity, 
%Moreover, there has been recent work in studying networks with Gaussian weights in classification tasks \cite{Giryesetal2016}. Random networks such as in \cite{Ulyanov2017,HH2018} have also been shown to be effective in a number of inverse problems. 
We emphasize that the use of generative models as priors in regularizing inverse problems is nascent, and we use this model  because it balances mathematical tractability with authenticity toward applications.  
%This line of investigation parallels that of the beginning of compressed sensing where initial results were established in the Gaussian measurement model, subsequently leading to improvements in more realistic models.

In order to establish a recovery guarantee for this random model, we establish it for models satisfying deterministic conditions on $\G$ and $A$. Then we show that an appropriate expansive-Gaussian model satisfies these deterministic conditions with high probability.  The first deterministic condition we consider roughly states that the neural network weights are approximately distributed uniformly on a sphere of a particular radius.  For $W \in \R^{n \times k}$ and $x \in \R^k$, define $W_{+,x} := \diag(Wx > 0) W$ where the $i$-th diagonal entry of $\diag(v > 0)$ is $1$ if $v_i > 0$ and $0$ otherwise. Note that $W_{+,x}x = \relu(Wx)$. The condition is stated as follows and was introduced in \cite{HV2017}:

\begin{defn}[Weight Distribution Condition]  We say that $W \in \R^{n \times k}$ satisfies the \textbf{Weight Distribution Condition (WDC)} with constant $\epsilon > 0$ if for all nonzero $x,y \in \R^k$, \begin{align}
    \left\| W_{+,x}^\top W_{+,y} - Q_{x,y} \right\| \leqslant \epsilon \nonumber
\end{align} where \begin{align}
    Q_{x,y} := \frac{\pi - \theta_{x,y}}{2\pi} I_k + \frac{\sin \theta_{x,y}}{2\pi} M_{\hat{x}\leftrightarrow \Hat{y}}. \label{Qdef}
\end{align} Here $\theta_{x,y} = \angle(x,y)$, $\Hat{x} = x/\|x\|$, $\Hat{y} = y/\|y\|$, $I_k$ is the $k \times k$ identity matrix, and $M_{\hat{x}\leftrightarrow \Hat{y}}$\footnote{A formula for this matrix is as follows: consider a rotation matrix $R$ that sends $\hat{x} \mapsto e_1$ and $\hat{y} \mapsto \cos\theta_{x,y} e_1 + \sin\theta_{x,y} e_2$ where $\theta_{x,y} = \angle(x,y)$. Then $M_{\hat{x} \leftrightarrow \hat{y}} = R^\top \left[\begin{array}{ccc}
    \cos\theta_{x,y} & \sin\theta_{x,y} & 0  \\
     \sin \theta_{x,y} & -\cos \theta_{x,y} & 0 \\
     0 & 0 & 0_{k-2}
\end{array}\right]R$ where $0_{k-2}$ is the $k-2 \times k-2$ matrix of zeros. Note that if $\theta_{x,y} = 0$ or $\pi$, $M_{\hat{x} \leftrightarrow \hat{y}} = \hat{x}\hat{x}^\top$ or $-\hat{x}\hat{x}^\top$, respectively.} is the matrix that sends $\Hat{x} \mapsto \Hat{y}$, $\Hat{y} \mapsto \Hat{x}$, and $z \mapsto 0$ for any $z \in \text{span}(\{x,y\})^{\perp}.$
\end{defn}  %Note that if $W_{ij} \sim \mathcal{N}(0,1/n)$, then $\E[W_{+,x}^\top W_{+,y}] = Q_{x,y}$ for any $x,y \in \R^k$. Hence this condition can be interpreted as determining whether the weights in the network are distributed approximately like a Gaussian or a uniform random variable on a sphere of a particular radius. In \cite{HV2017}, it was shown that Gaussian $W$ satisfies the WDC with high probability: \begin{lem}[WDC] \label{WDC_lemma}
%Fix $0 < \epsilon < 1$ and suppose $W \in \R^{n \times k}$ has i.i.d. $\mathcal{N}(0,1/n)$ entries. Then if $n \geqslant C_{\epsilon} k \log k$, then with probability at least $1 - 8n \exp(-\gamma_{\epsilon} k)$, $W$ satisfies the WDC with constant $\epsilon$. Here $C_{\epsilon}$ and $\gamma_{\epsilon}^{-1}$ depend polynomially on $\epsilon^{-1}$.
%\end{lem}

The second deterministic condition provides an RIP-like property for the measurement matrix $A$ when acting on pairs of secant directions within the range of $\G$. For $A \in \R^{m \times n}$ and $z \in \R^n$, define $A_z := \diag(\sgn(Az))A$ where $\sgn$ acts entrywise, and $\sgn(0) = 0$. Note that $A_zz = |Az|$. The condition is stated as follows and was introduced in a conference version of this work \cite{HLV18}: 

\begin{defn}[Range Restricted Concentration Property] \label{rrcp_def} We say that $A$ satisfies the \textbf{Range Restricted Concentration Property (RRCP)} with respect to $\G$ with constant $\epsilon > 0$ if for all $x,y,x_1,x_2,x_3,x_4 \in \R^k:$ 
%\red{[Make sure you deal with case where $\G(x)$ or $\G(y) = 0$?]}  \blue{[This case follows from the following: if either $\G(x) = 0$ or $\G(y)=0$, then both $A_{\G(x)}^{\top}A_{\G(y)}$ and $\Phi_{\G(x),\G(y)}$ are the $0$ matrix (by our def of $\sgn$) meaning that the inequality trivially holds. We could change the def of $\Phi$ to accomodate this as shown below.]}  
\begin{align}
    |\langle (A_{\G(x)}^\top A_{\G(y)} & - \Phi_{\G(x),\G(y)})(\G(x_1) - \G(x_2)), \G(x_3)-\G(x_4)\rangle | \nonumber \\
    & \leqslant L\epsilon\|\G(x_1) - \G(x_2)\|\|\G(x_3)-\G(x_4)\| \nonumber
\end{align} where \begin{align}
    \Phi_{z,w}:=  \begin{cases} \frac{\pi - 2\theta_{z,w}}{\pi}I_n + \frac{2\sin \theta_{z,w}}{\pi} M_{\hat{z} \leftrightarrow \hat{w}}  & \text{ if } z\neq 0, w \neq 0, \\
    0 & \text{ otherwise.}
\end{cases} \label{Phi_def}
\end{align} Here, $L$ is a universal constant and can be taken to be $33$.
\end{defn}

\subsection{Algorithm}
We provide a subgradient algorithm for optimizing \eqref{dpr_objective} under noisy measurements.  We show that this algorithm converges up to the noise level when the WDC and RRCP  properties are met.  
%These two deterministic conditions on the generative model and measurement matrix can be used to show that a variation of subgradient descent can provably recover the latent code $x_*$ up to a small amount of noise. 
In order to state the algorithm, we need some notation.  For a locally Lipschitz function $f : \mathcal{X} \rightarrow \R$ from a Hilbert space $\mathcal{X}$ to $\R$, the Clarke generalized directional derivative of $f$ at $x \in \mathcal{X}$ in the direction $u$, is defined by \begin{align*}
    f^o(x;u) := \limsup_{y \rightarrow x, t \downarrow 0} \frac{f(y+tu) - f(y)}{t}.
\end{align*} Then the generalized subdifferential of $f$ at $x$ is defined by \begin{align*}
    \partial f(x) = \{v \in \R^k : \langle v, u \rangle \leqslant f^o(x;u),\ \forall u \in \mathcal{X}\}.
\end{align*} Any element $v_x \in \partial f(x)$ is called a subgradient of $f$ at $x$. When $f$ is differentiable at $x$, $\partial f(x) = \{\nabla f(x)\}$.

\begin{algorithm}
\caption{Deep Phase Retrieval (DPR) Subgradient method}
\label{alg}
\begin{algorithmic}[1]
\REQUIRE Weights $W_i$, measurement matrix $A$, measurements $b_* = |Ay_*| + \eta$, \& step size $\alpha > 0$
\STATE Choose an arbitrary initial point $x_0 \in \mathbb{R}^k\setminus \{0\}$
%\ENSURE $x$ satisfying $\|x - x_*\| < \epsilon$
\FOR {$t = 0,1, 2, \ldots$}  \label{alg:st3}
\IF {$f(-x_{t}) < f(x_{t})$} \label{alg:cond1}
\STATE $\bar{x}_{t} \gets - x_{t}$; \label{alg:cond2}
\ELSE
\STATE $\bar{x}_t \gets x_t$; \label{alg:cond3}
\ENDIF \label{alg:endif}
\STATE Compute $v_{\bar{x}_{t}} \in \partial f(\bar{x}_t)$; %where for differentiable points $\bar{x}_t$, $v_{\bar{x}_t} = \Lambda_{\bar{x}_{t}}^\top A_{\G(x_{t})}^\top \left(A_{\G(\bar{x}_{t})}\Lambda_{\bar{x}_{t}} \bar{x}_{t} - b_*\right)$;
\STATE $x_{t+1} = \bar{x}_{t} - \alpha  v_{\bar{x}_{t}}$;
\ENDFOR
\end{algorithmic}
\end{algorithm}

We now introduce a subgradient descent scheme, given by Algorithm \ref{alg}, whose intuition is as follows. In expectation, the optimization landscape is characterized by Figure \ref{fig:expected_loss}. There exists two critical points away from the origin: the true minimizer and a negative multiple thereof. Moreover, the value of the objective function is higher near the negative multiple than near the global minimizer. At each iterate, we check the objective function value at the current latent code and its negative, choosing the point with smaller objective function value as our new iterate; see Steps \ref{alg:cond1}--\ref{alg:endif}. We then perform subgradient descent.

\begin{figure}[!htbp]
    \centering
    \includegraphics[scale=0.4]{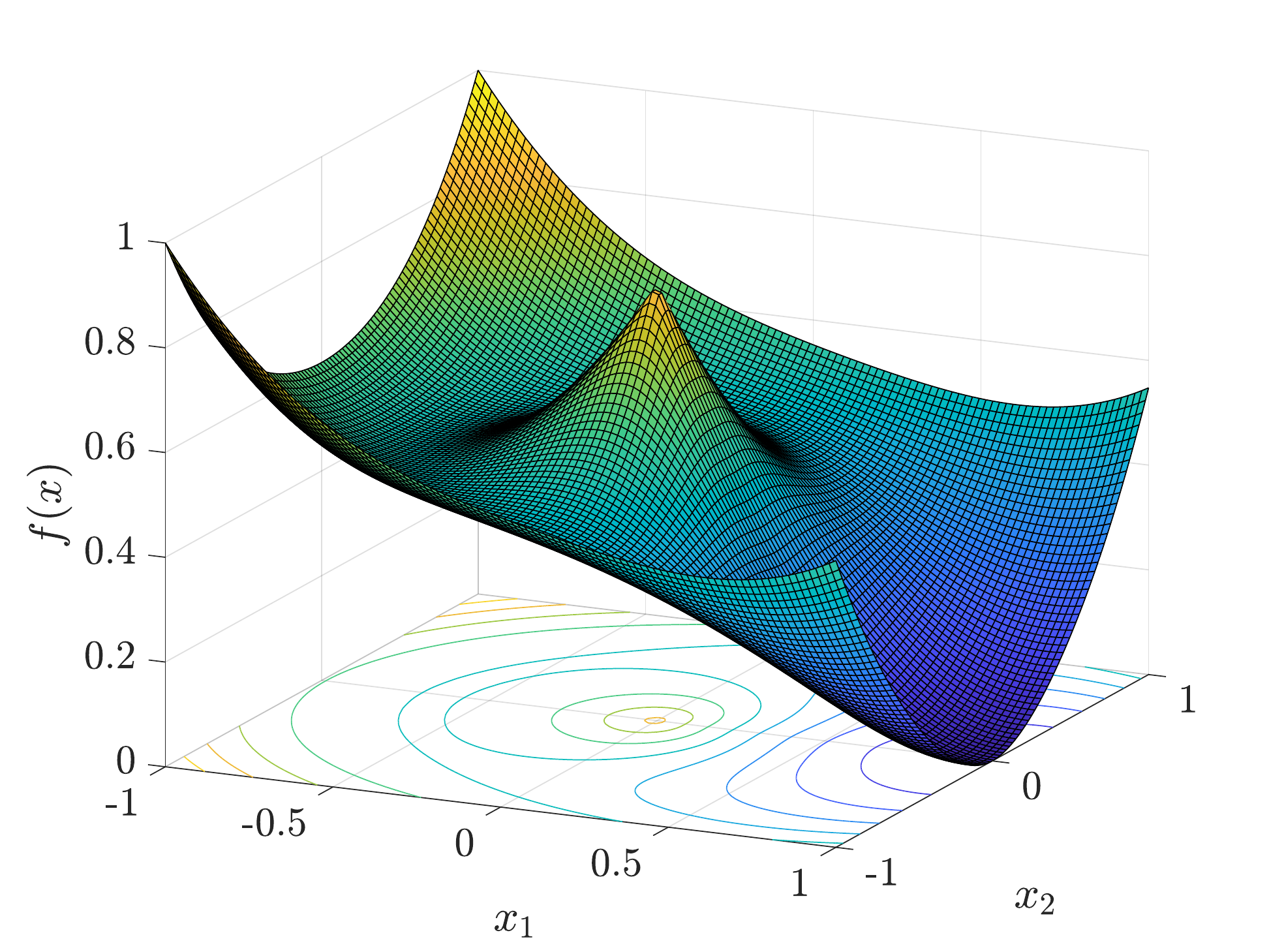}
    \includegraphics[scale=0.4]{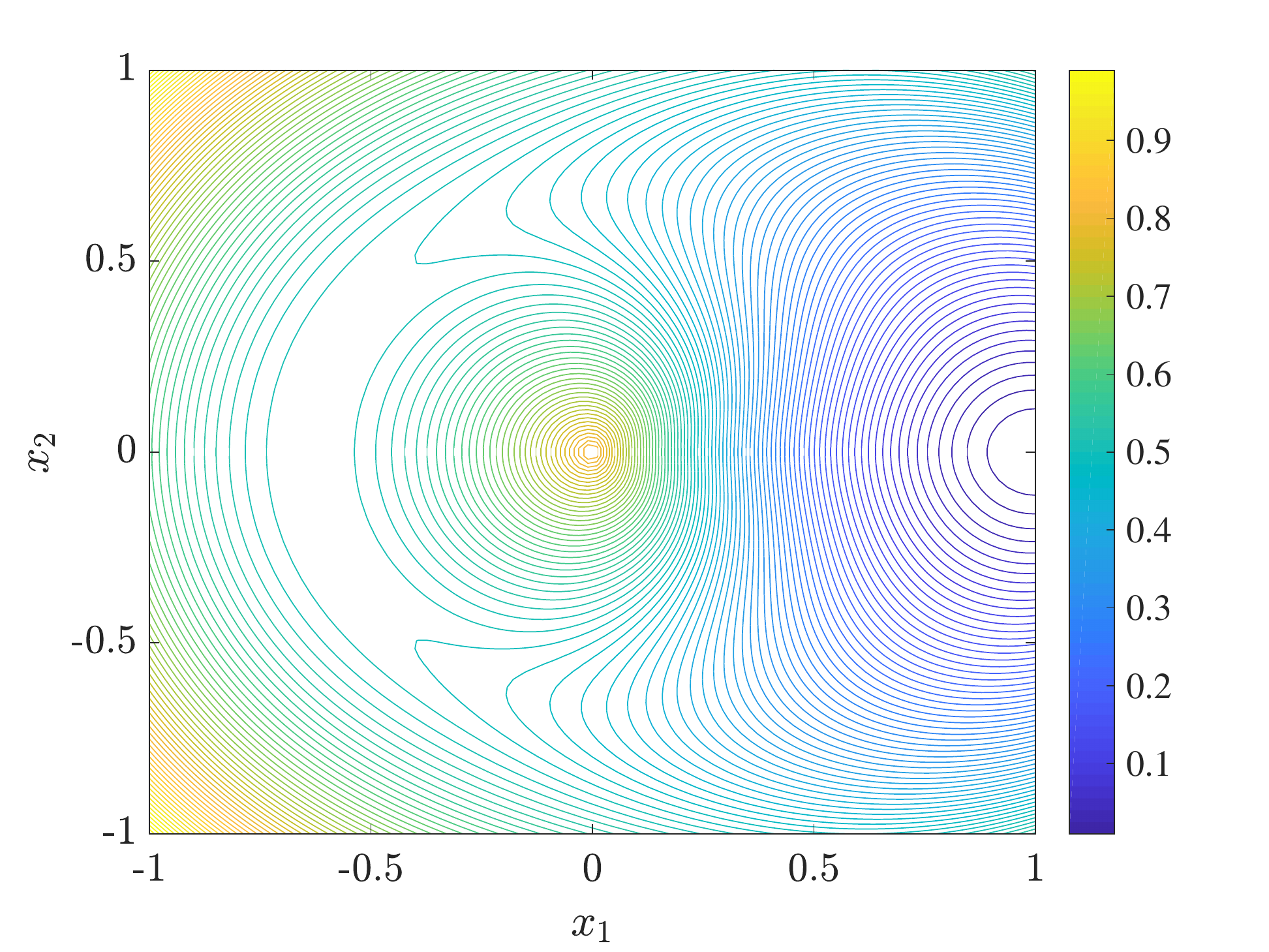}
    \caption{The landscape of \eqref{dpr_objective} where $y_* = \G(x_*)$ with $x_* = [1,\ 0]^\top \in \R^2$, $\G$ is a $1$-layer network, and the two determinstic conditions are satisfied with $\epsilon = 0$. The objective function's explicit form is given by $\mathcal{F}(x):= \frac{1}{4}(\|x\|^2 + \|x_*\|^2) - \left(\frac{\pi - 2g(\theta_{x,x_*})}{\pi}\langle x,Q_{x,x_*}x_*\rangle + \frac{2\sin g(\theta_{x,x_*})}{\pi}\|x_*\|\|x\|\right)$ where $\theta_{x,x_*} = \angle(x,x_*)$, $g $ is defined in \eqref{gdef}, $Q_{x,x_*}$ is defined in \eqref{Qdef} and whose $d$-layer form is given by equation \eqref{expected_loss_full}. We note that the idealized loss has only three critical points: the global minimizer $x_*$, a negative multiple thereof $-\rho_d x_*$ for some $\rho_d \in (0,1)$, and the origin.}
    \label{fig:expected_loss}
\end{figure}

\subsection{Main Results}

In this section, we outline our main results in both the probabilistic and deterministic settings. In particular, in Theorem \ref{main_deterministic_conv_result} we show that if the weights of $G$ satisfy the WDC and the measurement matrix $A$ satisfies the RRCP, then the iterates of Algorithm \ref{alg} converge to the true solution up to the noise level in the measurements. Then, in Corollary \ref{main_stochastic_conv_result} we show that the same conclusion holds when $W_i$ and $A$ are Gaussian with high probability as soon as $m = \Omega(dk\log(n_1n_2\dots n_d))$.

We consider the possibly noisy measurements \eqref{eqn:measurements-with-noise} and assume that the signal $y_*$ is in the range of $\G$ with latent code $x_*$; that is, $y_* = \G(x_*)$.
%Our formal convergence result will be stated with the two deterministic conditions outlined in Section \ref{intro_of_prob_section_w_detcond}. 
The following Theorem states that if the two deterministic conditions are satisfied with a sufficiently small parameter $\epsilon$ and the noise is sufficiently small, then the iterates of Algorithm \ref{alg} will converge to $x_*$ up to the noise level. 
\begin{thm}[Deterministic Convergence Guarantee] \label{main_deterministic_conv_result}
Let $d \geqslant 2$ and fix $0 < \epsilon < c_1\frac{1}{d^{102}}$. Suppose the noise  satisfies $\|\eta\| \leqslant c_{2}\frac{\|x_*\|}{2^{d/2}d^{48}}$. Suppose each $W_i$ of $\G$ satisfies the WDC with constant $\epsilon$, and suppose $A$ satisfies the RRCP with respect to $\G$ with constant $\epsilon$. Then the iterates $\{x_t\}_{t \geqslant 0}$ generated by Algorithm \ref{alg} with step size $\al \leqslant c_3\frac{2^d}{d^2}$ obey the following: \begin{enumerate}
    \item there exists an $N \in \N$ satisfying $N \leqslant C_4\frac{f(x_*)(2^{2d})}{d^6 \al \epsilon \|x_*\|^2}$ such that \begin{align}
        \|x_N - x_*\| \leqslant C_5 d^{12}\sqrt{\epsilon}\|x_*\| + C_6 d^92^{d/2}\|\eta\|; \label{bound_on_Nth_iterate}
        \end{align}
    \item for all $t \geqslant N$, we have \begin{align}
        \|x_{t+1} - x_*\| & \leqslant \tau^{t+1-N}\|x_N-x_*\| + \vartheta \frac{2^{d/2}}{d^2}\|\eta\|,\ \text{and} \label{final_iterate_bound} \\
        \|\G(x_{t+1}) - \G(x_*) \| & \leqslant \frac{1.2}{2^{d/2}}\tau^{t+1-N}\|x_{N}-x_*\| + \frac{1.2}{d^2}\vartheta \|\eta\| \label{final_output_through_genmodel_bound}
    \end{align} where $\tau := 1 - \frac{7}{8}\frac{\al}{2^d} \in (0,1)$ and $\vartheta := \frac{2c_3}{1-\tau}$.
\end{enumerate} Here $c_1, c_2, c_3, C_4, C_5,$ and $C_6$ are positive universal constants.
\end{thm} 

This result asserts that the iterates of Algorithm \ref{alg} will eventually be in a small neighborhood of the true solution whose size depends on $\epsilon$ and $\|\eta\|$ after $N = O(\epsilon^{-1})$ iterations. Furthermore, once in this neighborhood, the iterates will continue to converge linearly to the true solution up to the noise level. If no noise is present, then the true signal will be recovered. Note that the $2^d$ factors in the theorem are an artifact of the problem scaling. Roughly, the weights $W_i$ have spectral norm approximately $1$, and subsequent application of a ReLU will effectively zero out roughly half of the rows of $W_i$.  The resulting rows of $W_i$ will have spectral norm of roughly $1/2$.   Hence $\G(x)$ scales like $2^{-d/2}\|x\|$, $f(x)$ scales like $2^{-d}\|x\|^2$, and any subgradient $v_x$ scales like $2^d$. We also assume the noise scales like $2^{-d/2}$ to ensure it is on the order of the measurements. Doubling the variance of each entry of $W_i$ would eliminate these factors, but we consider the unscaled version because it is more convenient in the analysis.

We now address the expansive-Gaussian model.  We appeal to a result that shows that expansive (tall) Gaussian matrices satisfy the WDC with high probability.  
 \begin{lem}[Lemma 11 in  \cite{HV2017}] \label{WDC_lemma}
Fix $0 < \epsilon < 1$ and suppose $W \in \R^{n \times k}$ has i.i.d. $\mathcal{N}(0,1/n)$ entries. Then if $n \geqslant C_{\epsilon} k \log k$, then with probability at least $1 - 8n \exp(-\gamma_{\epsilon} k)$, $W$ satisfies the WDC with constant $\epsilon$. Here $C_{\epsilon}$ and $\gamma_{\epsilon}^{-1}$ depend polynomially on $\epsilon^{-1}$.
\end{lem}

 In this work, we establish that Gaussian matrices $A$ satisfies the RRCP with respect to an expansive-Gaussian $\G$ with high probability if they are sufficiently tall. This result is proven in Section \ref{RRCP_section}: \begin{lem}[RRCP] \label{RRCP_lemma}
Fix $0 < \epsilon < 1$ and suppose $A \in \R^{m \times n}$ has i.i.d. $\mathcal{N}(0,1/m)$ entries. Let $\G$ be a generative model of the form \eqref{Generator_formula} where each $W_i \in \R^{n_i \times n_{i-1}}$ has i.i.d. $\mathcal{N}(0,1/n_i)$ entries. If $m \geqslant C_{\epsilon} kd \log(n_1n_2\dots n_d)$, then with probability $1-\gamma m^{4k}\exp(-c_{\epsilon}m)$, $A$ satisfies the RRCP with respect to $\G$ with constant $\epsilon$. Here $\gamma$ is a universal constant and $C_{\epsilon}$ and $c_{\epsilon}^{-1}$ depend polynomially on $\epsilon^{-1}$. 
\end{lem}

Hence for Gaussian measurements and weight ensembles, we can combine Lemma \ref{WDC_lemma} and Lemma \ref{RRCP_lemma} with Theorem \ref{main_deterministic_conv_result} to obtain the following Corollary: 

\begin{cor}[Probabilistic Convergence Guarantee] \label{main_stochastic_conv_result}
Fix $0 < \epsilon < c_1\frac{1}{d^{102}}$ and suppose the noise satisfies $\|\eta\| \leqslant c_{2}\frac{\|x_*\|}{2^{d/2}d^{48}}$ for some universal constants $c_1$ and $c_2$. Suppose $\G$ is such that $W_i \in \R^{n_i \times n_{i-1}}$ has i.i.d. $\mathcal{N}(0,1/n_i)$ entries for $i = 1,\dots,d$. Suppose that $A \in \R^{m \times n_d}$ has i.i.d. $\mathcal{N}(0,1/m)$ entries independent from $\{W_i\}$. Then if $m \geqslant C_{\epsilon} dk \log (n_1n_2\dots n_d)$ and $n_i \geqslant C_{\epsilon} n_{i-1}\log n_{i-1}$ for $i =1,\dots,d$, then with probability at least $1 - \sum_{i=1}^d \gamma n_i\exp(-c_{\epsilon}n_{i-1}) - \gamma m^{4k} \exp(-c_{\epsilon} m)$, the same conclusion as Theorem \ref{main_deterministic_conv_result} holds. Here $C_{\epsilon}$ depends polynomially on $\epsilon^{-1}$, $c_{\epsilon}$ depends on $\epsilon$, and $\gamma$ is a universal constant.
\end{cor}

To the author's knowledge, this is the first result establishing provable signal recovery with a computationally efficient algorithm for undersampled generic phaseless linear measurements with optimal sample complexity. This sample complexity in our result scales with $k$, which can not be improved.  We made no attempt to obtain tight bounds on $d$, except to ensure that all dependences on $d$ are polynomial.  We remind the reader that any $2^d$ terms that appear are due to the problem scaling. We further note that subsequent developments since the original release of \cite{HLV18} relaxed the logarithmic growth factor on the sizes of each layer of the generative model \cite{Daskalaskisetal2020}.

%{ is optimal as it scales linearly with $k$ for fixed $d$. The fixed $d$ regime is realistic in contexts with real deep generative models as current state-of-the-art models have $d$ at most on the order of $10$. While the sample complexity seems to worsen with larger values of $d$, note that larger depth $d$ allows for smaller values of $k$. This is because the range of $\G$ lies in the union of a finite number of $k$-dimensional subspaces and the number of subspaces grows exponentially in $d$. Furthermore, while the dependencies on $\epsilon$, $d$, and $n_i$ could be improved, they are all polynomial and no attempt was made to optimize them. Similarly, the $2^d$ factors are simply an artifact of the problems scaling and normalizing the weights would eliminate such factors. Lastly, we note that this result in fact implies recovery with the same optimal sample complexity for compressive sensing under a generative prior. This is notoriously not the case for sparse phase retrieval methods as the best known computationally efficient algorithms all require sample complexity that is quadratic in the sparsity level.}

Lastly, we note that this result for compressive phase retrieval under optimal sample complexity implies recovery for linear compressive sensing under optimal sample complexity.  As such, this work subsumes the work of a subset of the authors in \cite{HV2017}.  This generalization of compressed sensing to compressive phase retrieval is conspicuously absent for structural priors based on sparsity, as the best known computationally efficient algorithms for sparsity priors require sample complexity that is quadratic in the sparsity level.

%%%%%%%%%%%%%%%%%%%

\subsection{Experiments on MNIST} \label{experiments section}

In this section, we compare the generative modeling approach for compressive phase retrieval with three sparse phase retrieval algorithms: the sparse truncated amplitude flow algorithm (SPARTA) \cite{SPARTA}, Thresholded Wirtinger Flow (TWF) \cite{Cai2016}, and the alternating minimization algorithm CoPRAM \cite{Hedge2017}. For the generative modeling approach, we used a modified version of Algorithm \ref{alg} as we empirically found the negation step (Steps \ref{alg:cond1}-\ref{alg:cond2}) only occurred at the first iterate. Hence we ran two gradient descents, one starting from a random initial iterate $x_0$ and another starting from its negation $-x_0$. We report results for the most successful reconstruction. Gradient descent was performed using the Adam optimizer \cite{ADAM}. For the remainder of this section, we will refer to the generative modeling approach as DPR. %\red{This can be more clearly presented.  Are we using ADAM? Or are we running subgradient descent?  This paragraph also mixes together things that are common to all of the problems, with things that only apply to the DPR.  Have a cleaner separation}  

In each task, the goal is to recover an image $y_*$ given $|Ay_*|$ where $A \in \R^{m \times n}$ has i.i.d. $\mathcal{N}(0,1/m)$ entries. The images were from the MNIST dataset \cite{MNIST}. This dataset consists of $60,000$ $28 \times 28$ images of handwritten digits. The generative model was a pretrained Variational Autoencoder (VAE) from \cite{Price2017}. The encoder network is of size $784 - 500 - 500 - 20$ while the generator network $\G$ is of size $20 - 500 - 500 - 784$. The latent code space dimension is $k = 20$.

For the sparse phase retrieval methods, we performed sparse recovery in the Daubechies-4 Wavelet domain.  We zero-padded the images to be of size $32 \times 32$. The resulting images generated by our algorithm were also uniformly padded with zeros around the border to obtain $32 \times 32$ images. For SPARTA and CoPRAM, we ran each algorithm with sparsity parameters ranging from $2$ to $212$ in increments of $15$, choosing the best reconstruction in terms of lowest reconstruction error. 

\begin{figure}[h]
  \centering
  \includegraphics[scale = 0.3]{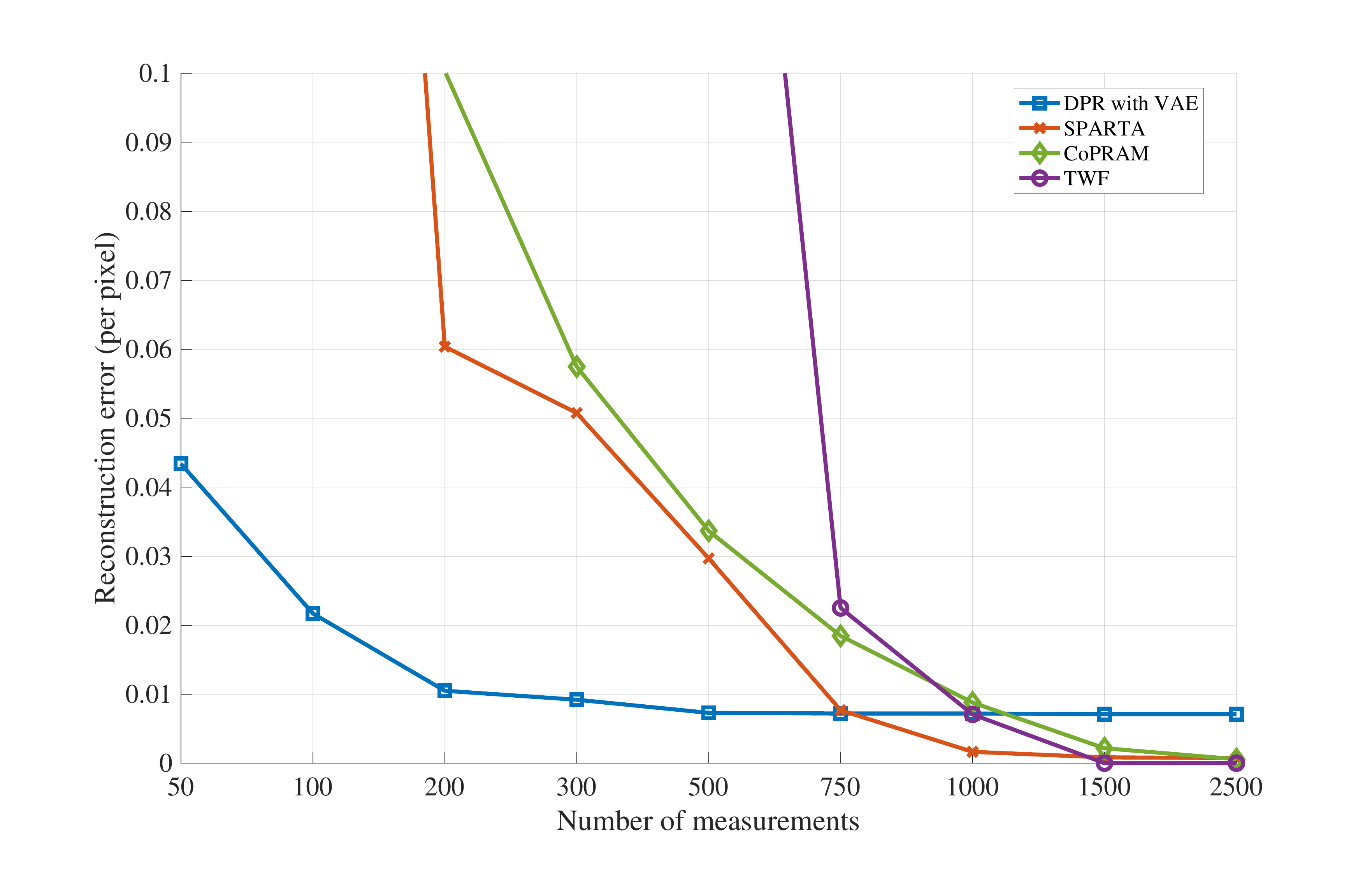}
  \includegraphics[scale = 0.3]{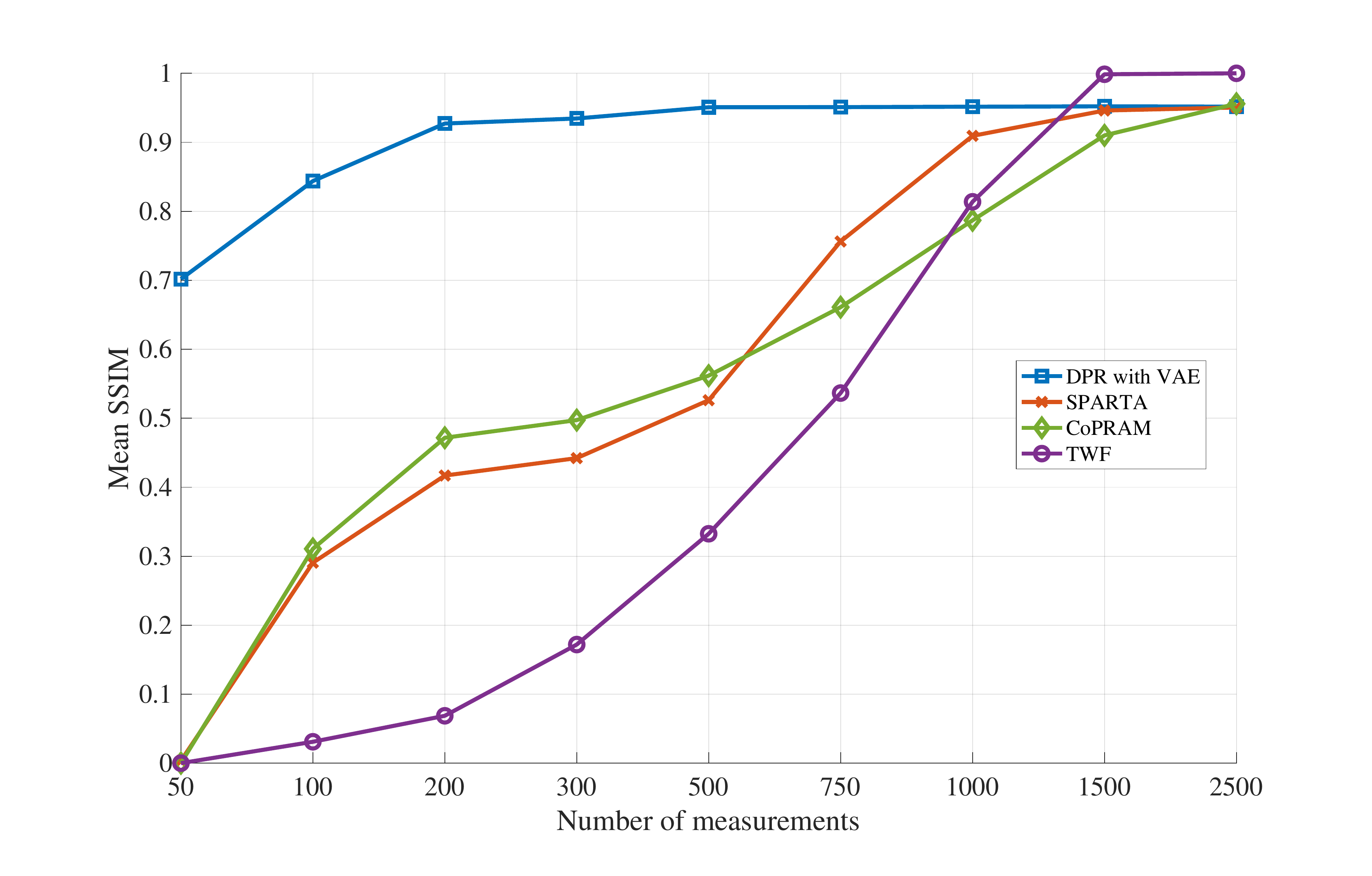}
  \caption{Each algorithm's average reconstruction error (top) and mean SSIM (bottom) over $10$ images from the MNIST test set for different numbers of measurements.}
  \label{fig:mnist_quant_results}
\end{figure}

We aimed to reconstruct $10$ images from the MNIST test set.  We allowed $5$ random restarts for each algorithm and chose the result with the least $\ell_2$ reconstruction error per pixel. We also report the Structural Similarity Index Measure (SSIM) \cite{Wang2004} for each reconstruction. The results in Figure \ref{fig:mnist_quant_results} demonstrate the success of our algorithm with very few measurements. For $200$ measurements, we can achieve accurate recovery with a mean SSIM value of over $0.9$ while other algorithms require $1000$ measurements or more. In terms of reconstruction error, our algorithm exhibits recovery with $200$ measurements comparable to the alternatives requiring $750$ measurements or more, which is where they begin to succeed.

We note that while our algorithm succeeds with fewer measurements than the other methods, our performance, as measured by per-pixel reconstruction error, saturates as the number of measurements increases since our reconstruction accuracy is ultimately bounded by the generative model's representational error. As generative models improve, their representational errors will decrease.  Nonetheless, as can be seen in the reconstructed digits in Figure \ref{fig:mnist_qualitative_results}, the recoveries are semantically correct (the correct digit is legibly recovered) even though the reconstruction error does not decay to zero.  In applications, such as MRI and molecular structure estimation via X-ray crystallography, semantic error measures would be a more informative estimates of recovery performance than per-pixel error measures.

\begin{figure}[h]
  \centering
  \includegraphics[scale = 0.275]{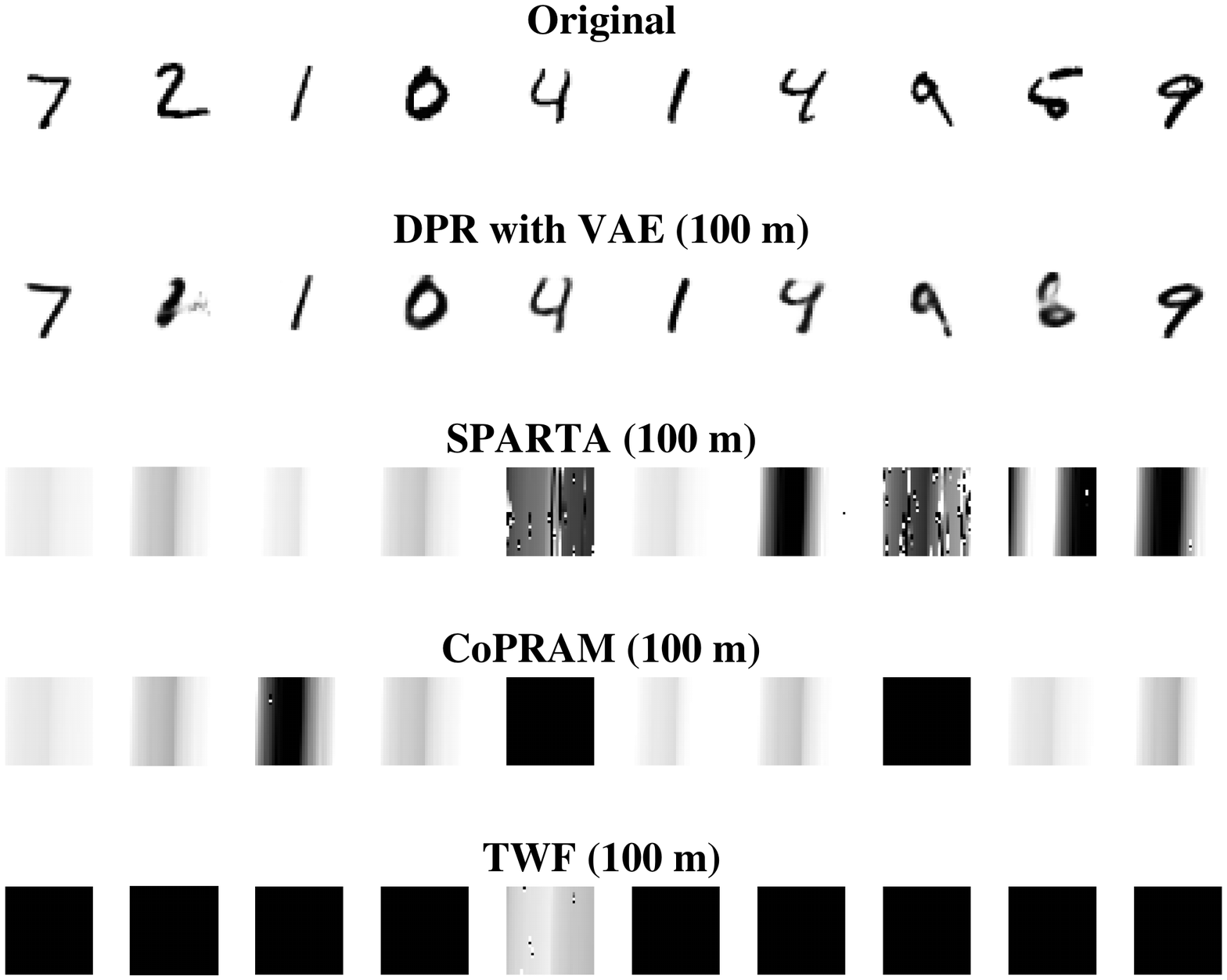}
  \includegraphics[scale = 0.275]{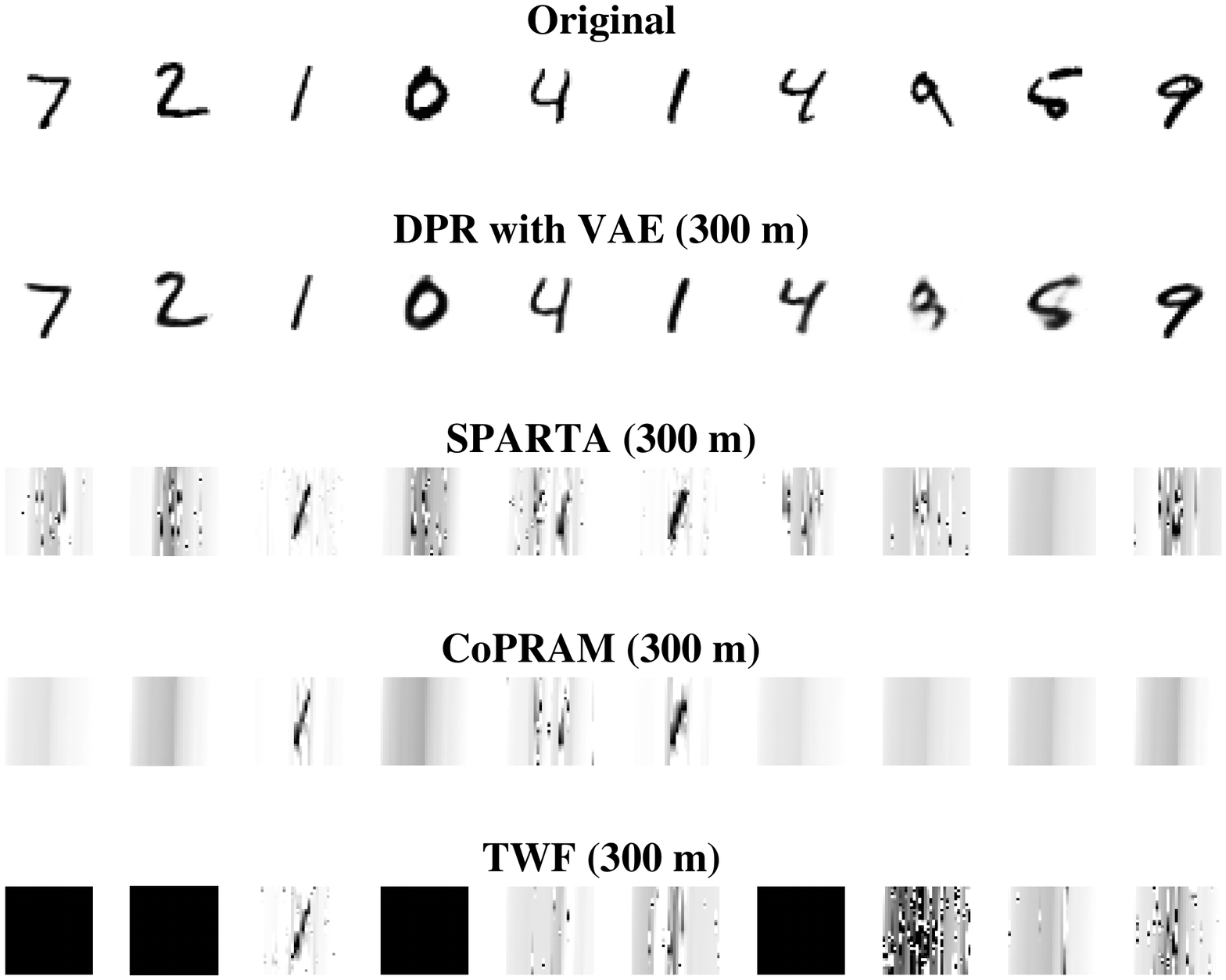}
  \caption{Each algorithm's reconstructed images with $100$ measurements (left) and $300$ measurements (right). If an image is blank, then the reconstruction error between the blank image and the original image was lower than that of the algorithm's reconstructed image and the original image. We note that even for as few as $100$ measurements, nearly all of DPR's reconstructions are semantically correct.}
  \label{fig:mnist_qualitative_results}
\end{figure}

%%%%%%%%%%%%%%%%%%%

\section{Proof of Convergence Result in Deterministic Setting} \label{proof_of_convergence_section}

In this section, we will formally prove Theorem \ref{main_deterministic_conv_result}. Section \ref{notation_section} outlines the notation we will use throughout the proofs. Section \ref{det_proof_sketch} provides a high-level sketch of our proof and outlines its central arguments while Section \ref{preliminaries} discusses preliminary results that are used throughout the proofs. Then Section \ref{formal_proof_section} presents the proof of Theorem \ref{main_deterministic_conv_result} which is broken down into four central results. Finally, Section \ref{additional_results_sec} presents supplementary results and their proofs that aid in establishing Theorem \ref{main_deterministic_conv_result}.

\subsection{Notation} \label{notation_section}

Let $ (\cdot)^\top$ denote the real transpose. Let $[n] = \{1,\dots,n\}$. Let $\mathcal{B}(x,r)$ denote the closed Euclidean ball centered at $x$ with radius $r$. Let $\|\cdot\|$ denote the $\ell_2$ norm for vectors and spectral norm for matrices. For any non-zero $x \in \R^n$, let $\hat{x} = x/\|x\|$. For non-zero $x,y \in \R^n$, let $\theta_{x,y} = \angle(x,y)$ Let $\relu(x) := \max(x,0)$. Define $\sgn : \R \rightarrow \R$ to be $\sgn(x) = x/|x|$ for non-zero $x \in \R$ and $\sgn(0) = 0$ otherwise. Let $\one(E)$ be the indicator function on the event $E$. For a vector $v \in \R^n$, $\text{diag}(\sgn(v))$ is $\sgn(v_i)$ in the $i$-th diagonal entry and $\text{diag}(v > 0)$ is $1$ in the $i$-th diagonal entry if $v_i > 0$ and $0$ otherwise. Let $\Pi_{i=d}^1 W_i = W_d W_{d-1} \dots W_1$. For any $x \in \R^k$ and $i \in [d]$, define $W_{i,+,x} : = \text{diag}(W_{i-1,+,x}\dots W_{2,+,x}W_{1,+,x}x > 0)W_i.$ Set $\Lambda_x : = \PiWdix$ and $x_d := \Lambda_x x$. Note that we have the following string of equalities: $\G(x) = \PiWdix x = \Lambda_x x = x_d$. Let $I_n$ be the $n \times n$ identity matrix. Let $\mathcal{S}^{k-1}$ denote the unit sphere in $\R^k$. We write $\gamma = \Omega(\delta)$ when $\gamma \geqslant C\delta$ for some positive constant $C$. Similarly, we write $\gamma = O(\delta)$ when $\gamma \leqslant C \delta$ for some positive constant $C$. When we say that a constant depends polynomially on $\epsilon^{-1}$, this means that it is at least $C \epsilon^{-k}$ for some positive $C$ and positive integer $k$. Positive numerical constants will be denoted using $C$ or $K$ with various subscripts. In general, numerical constants larger than $1$ will be denoted by capital letters and constants smaller than $1$ with lower case letters. For notational convenience, we write $a = b + O_1(\epsilon)$ if $\|a - b\| \leqslant \epsilon$ where $\|\cdot\|$ denotes $|\cdot|$ for scalars, $\ell_2$ norm for vectors, and spectral norm for matrices.

\subsection{Sketch of Proof for Theorem \ref{main_deterministic_conv_result}} \label{det_proof_sketch}

Theorem \ref{main_deterministic_conv_result} is proven by showing that for all $x \in \R^k$, any subgradient $v_x \in \partial f(x)$ is approximated by $h_x \in \R^k$ which has an analytical expression and that does not vanish outside of neighborhoods of the true solution $x_*$ and a negative multiple $-\rho_d x_*$ for some $\rho_d  \in (0,1)$. Thus any $v_x \in \partial f(x)$ is bounded away from zero for $x$ outside of these two neighborhoods, leading to convergence towards one of these regions. Then we ensure that the negation step of our algorithm (Steps \ref{alg:cond1}-\ref{alg:endif}) will update any iterate near $-\rho_d x_*$ to be in a neighborhood of $x_*$. Finally, we ensure convergence to $x_*$ up to the noise level by showing that the objective function exhibits a convexity-like property in a neighborhood of $x_*$.

To provide our sketch, we define some quantities. Define the function $g : [0,2\pi] \rightarrow \R$ by \begin{align}
    g(\theta) := \cos^{-1}\left(\frac{(\pi-\theta)\cos\theta + \sin\theta}{\pi}\right) .\label{gdef}
\end{align} For any $x \in \R^k \setminus \{0\}$, let $h_x \in \R^k$ be defined as \begin{align*}
h_{x} & : =  \frac{\|x_*\|}{2^d}\left(\frac{\pi - 2\overline{\theta}_{d,x}}{\pi}\right)\left(\prod_{i=0}^{d-1}\frac{\pi - \overline{\theta}_{i,x}}{\pi}\right)\hat{x}_* \\
& + \frac{1}{2^d}\left[\|x\| - \|x_*\|\left( \frac{2\sin \overline{\theta}_{d,x}}{\pi} + \left(\frac{\pi - 2\overline{\theta}_{d,x}}{\pi}\right) \sum_{i=0}^{d-1} \frac{\sin \overline{\theta}_{i,x}}{\pi}\left(\prod_{j=i+1}^{d-1} \frac{\pi - \overline{\theta}_{i,x}}{\pi}\right)\right)\right]\hat{x}
\end{align*} where $\overline{\theta}_{0,x} = \angle(x,x_*)$ and $\overline{\theta}_{i,x} = g(\overline{\theta}_{i-1,x})$ for $i \in [d]$. We further define \begin{align*}
\rho_d : = \frac{2 \sin \breve{\theta}_d}{\pi} + \left(\frac{\pi - 2 \breve{\theta}_d}{\pi}\right) \sum_{i=0}^{d-1}\frac{\sin \breve{\theta}_i}{\pi} \left(\prod_{j=i+1}^{d-1} \frac{\pi - \breve{\theta}_j}{\pi}\right)
\end{align*} where $\breve{\theta}_0 = \pi$ and $\breve{\theta}_i = g(\breve{\theta}_{i-1})$ for $i \in [d]$. For a parameter $\beta > 0$, define \begin{align}
    \Seps_{\beta} := \left\{x \in \R^k \setminus \{0\} : \|h_x\| \leqslant \frac{1}{2^d}\beta\max(\|x\|,\|x_*\|)\right\}. \label{Sbeta_def}
\end{align} A direct analysis in Lemma \ref{Seps_lemma} shows that for appropriate values of $\beta$, $\mathcal{S}_{\beta}$ is contained in the union of neighborhoods of $x_*$ and $-\rho_d x_*$: \begin{align*}
    \Seps_{\beta} & \subset  \mathcal{B}(x_*,70000\pi^2 d^9\beta\|x_*\|) \cup \mathcal{B}(-\rho_d x_*, 77422\pi^2d^{12}\sqrt{\beta}\|x_*\|). %\\
    %& =: \Seps_{\beta}^+ \cup \Seps_{\beta}^{-}.
\end{align*} Set $\Seps_{\beta}^+:= \Seps_{\beta} \cap \mathcal{B}(x_*,70000\pi^2 d^9\beta\|x_*\|)$ and $\Seps_{\beta}^{-}:= \Seps_{\beta} \cap \mathcal{B}(-\rho_d x_*, 77422\pi^2d^{12}\sqrt{\beta}\|x_*\|).$

A sketch of our proof is outlined as follows:

\begin{itemize}
    \item First, we establish that all subgradients are bounded away from zero for iterates outside of $\Seps_{\beta}$. Specifically, we show that when the WDC and RRCP are satisfied with constant $\epsilon$, any $v_{x_t} \in \partial f(x_t)$ satisfies $v_{x_t} \approx h_{x_t}$ and $h_{x_t}$ is bounded away from $0$ by the definition of $\Seps_{\beta}$. Thus $\|v_{x_t}\|$ must be bounded away from zero for points $x_t \notin \mathcal{S}_{\beta}$. See Section \ref{conv_proof_zero_subsection}.
    \item Next, we establish convergence to $\Seps_{\beta}$. In particular, we show that the previous result implies that subgradient descent at each iteration makes progress in the sense that for each non-zero $x_t \notin \mathcal{S}_{\beta}$ $$f(x_{t+1}) - f(x_t) \leqslant -C\epsilon$$ for some $C > 0$. Thus after $\Omega(\epsilon^{-1})$ iterations, the iterates will eventually belong to $\mathcal{S}_{\beta}.$ See Section \ref{conv_proof_first_subsection}.
    \item Third, we show that the negation step of our algorithm ensures that the iterates converge to $\Seps_{\beta}^+$. Specifically, we prove that for points $x \approx x_*$ and $y \approx -\rho_d x_*$, $f(x) < f(y)$. Thus if an iterate $x_t \in \mathcal{S}_{\beta}^-$, $f(-x_t) < f(x_t)$ so the negation step of our algorithm (Steps \ref{alg:cond1}--\ref{alg:endif}) ensures $\bar{x}_t = -x_t$ and $\bar{x}_t \in \mathcal{S}_{\beta}^+$. See Section \ref{conv_proof_second_subsection}.
    \item Finally, we establish convergence to $x_*$ up to the noise level. Specifically, we prove that once in $\mathcal{S}_{\beta}^+$, a convexity-like property near $x_*$ implies that the iterates converge to $x_*$ up to the noise level in the measurements. See Section \ref{conv_proof_third_subsection}.
\end{itemize} 

\subsection{Preliminaries for Proofs} \label{preliminaries}

We will make use of the following fact concerning the Clarke subdifferential of the objective function $f$. Since $f$ is piecewise quadratic, Theorem 9.6 from \cite{Clason2017} asserts that for any $x \in \R^k$, the Clarke subdifferential $\partial f(x)$ can be written equivalently as \begin{align}
    \partial f(x) = \text{conv}(v_1,v_2,\dots,v_{s}) = \left\{\sum_{\ell=1}^s c_{\ell}v_{\ell} :  \sum_{\ell=1}^s c_{\ell} = 1\ \text{and}\ c_{\ell} \geqslant 0\ \text{for}\ \ell \in [s]\right\} \label{subdifferential_definition}
\end{align} where $\text{conv}(\cdot) $ denotes the convex hull, $s$ is the number of quadratic functions adjoint to $x$, and $v_{\ell}$ is the gradient of the $\ell$-th quadratic function of $f$ at $x$. Moreover, for each $v_{\ell}$, there exists a direction $w_{\ell}$ and a sufficiently small $\delta_{\ell} > 0$ such that $f$ is differentiable at $x + \delta_{\ell} w_{\ell}$ and $v_{\ell} = \lim_{\delta_{\ell} \rightarrow 0^+} \nabla f(x + \delta_{\ell}w_{\ell})$.

\subsection{Proof of Theorem \ref{main_deterministic_conv_result}} \label{formal_proof_section}

We now set out to prove Theorem \ref{main_deterministic_conv_result}. In Sections \ref{conv_proof_zero_subsection} - \ref{conv_proof_third_subsection}, we establish four main lemmas, each of which pertain to one of the items in the sketch of our proof from Section \ref{det_proof_sketch}. Theorem \ref{main_deterministic_conv_result} is then proven in Section \ref{final_proof_subsection}. Prior to beginning the proof, we state the necessary assumptions we will make: \begin{customasmp}{A} \label{assumptions} We assume the following hold for some numerical constants $c_1$, $c_2$, and $c_3$:
\begin{enumerate}
    \item[\textbf{A1.}] $0 < \epsilon < c_1 d^{-102}$,
    \item[\textbf{A2.}] the noise $\eta$ satsifes $\|\eta\| \leqslant \frac{c_2\|x_*\|}{2^{d/2}d^{48}}$, and
    \item[\textbf{A3.}] the step size $\al > 0$ satisfies $\al \leqslant c_3\frac{2^d}{d^2}$.
\end{enumerate}
\end{customasmp} \noindent We note that Proposition \ref{bounded_iterates} shows that after a polynomial number of steps, the iterates of our algorithm stay outside of a ball of the origin. Hence we assume throughout that the norm of our iterates are bounded away from zero. This result is proven in Section \ref{bounded_iterates_subsection}.

\subsubsection{Uniform control over subgradients} \label{conv_proof_zero_subsection}

We first show that the descent direction does not vanish for points outside of $\Seps_{\beta}$. The main idea of this result is that for points $x$ such that $\|h_x\|$ is sufficiently bounded away from zero, any $v_x \in \partial f(x)$ is also bounded away from zero. 

To prove this, we require the following three lemmas. The first gives a simple upper bound on the norm of our descent direction.

\begin{lem}\label{bound_on_norm_of_v}
Fix $\epsilon > 0$ such that $K d^3 \sqrt{\epsilon} \leqslant 1$ where $K$ is a universal constant. Suppose $A \in \R^{m \times n_d}$ satisfies the RRCP with respect to $\G$ with constant $\epsilon$ and $\G$ is such that each $W_i \in \R^{n_i \times n_{i-1}}$ satisfies the WDC with constant $\epsilon$ for $i \in [d]$. Then for any $x \in \R^k \setminus \{0\}$ and $v_x \in \partial f(x)$, \begin{align*}
    \|v_x\| \leqslant \frac{Cd}{2^d}\max(\|x\|,\|x_*\|) + \frac{2}{2^{d/2}}\|\eta\|
\end{align*} where $C$ is a numerical constant.
\end{lem}

\noindent The second shows that $h_x$ is Lipschitz with respect to $x$ for points away from the origin.

\begin{lem} \label{h_lipschitz_lemma}
For all $x,y\neq 0$, we have that \begin{align*}
\|h_x - h_y\| \leqslant \left(\frac{(2d^2 + (10\pi+8)d+20\pi)\|x_*\|}{\pi^2 2^d}\max\left(\frac{1}{\|x\|},\frac{1}{\|y\|}\right) + \frac{1}{2^d}\right) \|x-y\|.
\end{align*} In particular, if $x,y\notin \mathcal{B}(0,r\|x_*\|)$ for some $r > 0$, then \begin{align*}
\|h_x-h_y\| \leqslant \left(\frac{2d^2 + (10\pi+8)d+20\pi}{r\pi^2 2^d} + \frac{1}{2^d}\right)\|x-y\|.
\end{align*} 
\end{lem}

\noindent The third states that for any non-zero $x \in \R^k$ and any $v_x \in \partial f(x)$, $h_x$ approximates $v_x$ well.

\begin{lem} \label{conc_v_to_h}
Fix $\epsilon > 0$ such that $\epsilon < d^{-4}(1/16\pi)^2$ If $A \in \R^{m \times n_d}$ satisfies the RRCP with respect to $\G$ with constant $\epsilon$ and $\G$ is such that each $W_i \in \R^{n_i \times n_{i-1}}$ satisfies the WDC with constant $\epsilon$ for $i \in [d]$, then for any $x \neq 0$ and $v_x \in \partial f(x)$ \begin{align*}
\|v_x - h_x\| \leqslant  \frac{Kd^3\sqrt{\epsilon}}{2^d}\max(\|x\|,\|x_*\|) + \frac{2}{2^{d/2}}\|\eta\|
\end{align*} where $K$ is a universal constant. 
\end{lem}
\noindent Each of these results are proven in Section \ref{proofs_for_controlling_v_subsection}. We are now ready to state and prove the main result of this section.

\begin{lem} \label{closeness_of_subgrads} Suppose Assumptions A1-A3 are satisfied and set $\beta:= 4Kd^3\sqrt{\epsilon} + 11\|\eta\|2^{d/2}/\|x_*\|$ where $K$ is a numerical constant. Let $A \in \R^{m \times n}$ satisfy the RRCP with respect to $\G$ with constant $\epsilon$. Let $\G$ be such that $W_{i} \in \R^{n_i \times n_{i-1}}$ satisfy the WDC with constant $\epsilon$ for all $i \in [d]$. Suppose that $x \notin \Seps_{\beta}$ and $x \notin \mathcal{B}(0,c_0\|x_*\|)$ for some numerical constant $c_0$. Then for any $v_x \in \partial f(x)$, we have \begin{align}
    \frac{1}{3}\|v_{x}\| \geqslant \frac{Kd^3\sqrt{\epsilon}}{2^d}\|x_*\|. \label{norm_v_lower_bound}
\end{align} Moreover, we have that for any $\lambda \in [0,1]$, \begin{align}
    \|v_{\tilde{x}} - v_{x}\| \leqslant \frac{5}{6}\|v_{x}\| \label{v_tilde_lipschitz}
\end{align} where $\tilde{x} = x - \lambda\al v_{x}$, $v_{x} \in \partial f(x)$ and $v_{\tilde{x}} \in \partial f(\tilde{x})$.
\end{lem}

\begin{proof}[Proof of Lemma \ref{closeness_of_subgrads}]

 By Lemma \ref{h_lipschitz_lemma}, we have that $h_x$ is Lipschitz for all $x \notin \mathcal{B}(0,c_0\|x_*\|)$, i.e. there exists a numerical constant $L_{c_0}$ such that for any $x,y \notin \mathcal{B}(0,c_0\|x_*\|)$ $$\|h_x - h_y\| \leqslant \frac{L_{c_0} d^2}{2^d}\|x-y\|.$$ Moreover, Lemma \ref{conc_v_to_h} implies for any $x \neq 0$, we have \begin{align*}
    \|v_x - h_x\| \leqslant  \frac{Kd^3\sqrt{\epsilon}}{2^d}\max(\|x\|,\|x_*\|) + \frac{2}{2^{d/2}}\|\eta\|
\end{align*} for some numerical constant $K$. Hence  \begin{align}
    \|v_{\tilde{x}} - v_{x}\| & \leqslant \|v_{\tilde{x}} - h_{\tilde{x}}\| + \|h_{\tilde{x}} - h_{x}\| + \|h_{x} - v_{x}\| \nonumber\\
    & \leqslant \frac{Kd^3\sqrt{\epsilon}}{2^d}\max(\|\tilde{x}\|,\|x_*\|) + \frac{L_{c_0} d^2}{2^d}\|\tilde{x} - x\| + \frac{Kd^3\sqrt{\epsilon}}{2^d}\max(\|x\|,\|x_*\|) + \frac{4}{2^{d/2}}\|\eta\|\nonumber\\
    & \leqslant \frac{Kd^3\sqrt{\epsilon}}{2^d}\max(\|x - \lambda\al v_{x}\|,\|x_*\|) + \al\frac{L_{c_0} d^2}{2^d}\|v_{x}\| + \frac{Kd^3\sqrt{\epsilon}}{2^d}\max(\|x\|,\|x_*\|) + \frac{4}{2^{d/2}}\|\eta\| \nonumber\\
    & \leqslant \frac{Kd^3\sqrt{\epsilon}}{2^d}\max(\|x\| + \al \| v_{x}\|,\|x_*\|) + \al\frac{L_{c_0} d^2}{2^d}\|v_{x}\| + \frac{Kd^3\sqrt{\epsilon}}{2^d}\max(\|x\|,\|x_*\|) + \frac{4}{2^{d/2}}\|\eta\|\nonumber \\
    & \leqslant \frac{Kd^3\sqrt{\epsilon}}{2^d}\left(2 + \al\frac{Cd}{2^d}\right)\max(\|x\|,\|x_*\|) + \al \frac{L_{c_0} d^2}{2^d}\|v_{x}\| + \frac{6}{2^{d/2}}\|\eta\| \label{first_bound_onsubgrad_diff} %\\
    %& \leqslant \frac{K\sqrt{\epsilon}d^3}{2^d}\left(2 + \al\frac{Cd}{2^d}\right)\max(\|x_t\|,\|x_*\|) + \al \frac{L_{c_0} d^2}{2^d}\|v_{x_t}\| + \frac{4C_2\|x_*\|}{2^{d}d^{48}} 
\end{align} where we used Lemma \ref{h_lipschitz_lemma} and Lemma \ref{conc_v_to_h} in the second inequality, the definition of $\tilde{x}$ in the third inequality and Lemma \ref{bound_on_norm_of_v} in the last inequality for some numerical constant $C$. Now, we lower bound $\|v_{x}\|$. Since $x \notin \Seps_{\beta}$ we have that \begin{align}
    \|v_{x}\| & \geqslant \|h_{x}\| - \|h_{x} - v_{x}\| \nonumber \\
    & \geqslant \frac{1}{2^d}\max(\|x\|,\|x_*\|)\left(\beta - Kd^3\sqrt{\epsilon} - 2\|\eta\|\frac{2^{d/2}}{\|x_*\|}\right) \nonumber\\
    & = \frac{1}{2^d}\max(\|x\|,\|x_*\|)\left(3Kd^3\sqrt{\epsilon} + 9\|\eta\|\frac{2^{d/2}}{\|x_*\|}\right) \label{bd_w_noise_on_vtilde_diff}\\
    & \geqslant \frac{3Kd^3\sqrt{\epsilon}}{2^d}\max(\|x\|,\|x_*\|) \label{one_third_lower_bound_on_subgrad}
\end{align} where we used the definition of $\beta$ and Lemma \ref{conc_v_to_h} in the second inequality. Note that this proves \eqref{norm_v_lower_bound}. Applying \eqref{bd_w_noise_on_vtilde_diff} to equation \eqref{first_bound_onsubgrad_diff}, we attain \begin{align*}
    \|v_{\tilde{x}} - v_{x}\| & \leqslant \frac{2}{3}\|v_x\| +  \al\frac{Cd}{2^d}\cdot\frac{Kd^3\sqrt{\epsilon}}{2^d}\max(\|x\|,\|x_*\|) + \al \frac{L_{c_0} d^2}{2^d}\|v_{x}\| \\
    & \leqslant \frac{1}{3}\left(2 + \al \frac{Cd}{2^d}\right)\|v_{x}\| + \al \frac{L_{c_0} d^2}{2^d}\|v_{x}\| \\
    & \leqslant \left(\frac{2}{3} +  \frac{\al}{3}\cdot\frac{\tilde{C}d^2}{2^d}\right)\|v_{x}\| \\
    & \leqslant \frac{5}{6}\|v_{x}\|
\end{align*} where $\tilde{C} = C + L_{c_0}$. In the first inequality, we used \eqref{bd_w_noise_on_vtilde_diff}. In the second inequality, we used \eqref{one_third_lower_bound_on_subgrad}. The last inequality follows by choosing $c_3$ in the upper bound $\al \leqslant c_3\frac{2^d}{d^2}$ small enough so that $\frac{\al}{3}\frac{\tilde{C}d^2}{2^d} \leqslant \frac{1}{6}$. 
    
\end{proof}

\subsubsection{Convergence to neighborhoods of $x_*$ and $-\rho_d x_*$} \label{conv_proof_first_subsection}

Using Lemma \ref{closeness_of_subgrads}, we can now show that the iterates of our algorithm make sufficient progress at each step so that they eventually are in $\Seps_{\beta}$ after a polynomial number of iterations.

\begin{lem} \label{conv_to_neighborhoods_lemma} Suppose Assumptions A1-A3 are satisfied and set $\beta:= 4Kd^3\sqrt{\epsilon} + 11\|\eta\|2^{d/2}/\|x_*\|$ where $K$ is a numerical constant . Let $A \in \R^{m \times n}$ satisfy the RRCP with respect to $\G$ with constant $\epsilon$. Let $\G$ be such that $W_{i} \in \R^{n_i \times n_{i-1}}$ satisfy the WDC with constant $\epsilon$ for all $i \in [d]$. For $x_t \notin \mathcal{S}_{\beta}$, we have \begin{align*}
    f(x_{t+1}) - f(x_t)  \leqslant -\al \frac{9K^2d^6\epsilon}{6 (2^{2d})}\|x_*\|^2.
\end{align*} Moreover, there exists an $N \leqslant \frac{6f(x_0)(2^{2d})}{9K^2d^6 \al \epsilon \|x_*\|^2}$ such that $x_N \in \mathcal{S}_{\beta}$ where $x_0$ is the initial iterate of our algorithm and $\al > 0$ is the step size.
\end{lem}
\begin{proof}[Proof of Lemma \ref{conv_to_neighborhoods_lemma}]
Recall that by Proposition \ref{bounded_iterates}, we may assume $x_t \notin \mathcal{B}(0,c_0\|x_*\|)$ where $c_0$ is a constant. We first consider the case when $\bar{x}_t = -x_t$. Then we must have that $f(\bar{x}_t) < f(x_t)$. Hence for any $v_{\bar{x}_t} \in \partial f(\bar{x}_t)$, we have \begin{align*}
    f(x_{t+1}) - f(x_t) & = f(x_{t+1}) - f(\bar{x}_t) + f(\bar{x}_t) - f(x_t) < f(\bar{x}_t - \al v_{\bar{x}_t}) - f(\bar{x}_t)
\end{align*} where we used $f(\bar{x}_t) < f(x_t)$ and the definition of $x_{t+1}$ in the first inequality. Thus observe that it suffices to establish the inequality for $f(\bar{x}_t - \al v_{\bar{x}_t}) - f(\bar{x}_t)$ since this will also establish the case when $\bar{x}_t = x_t$.

Now, choose $v_{\bar{x}_t} \in \partial f(\bar{x}_t)$. By the generalized mean value theorem for the Clarke subdifferential (Theorem 8.13 in \cite{Clason2017}), there exists a $\lambda \in [0,1]$ and $\tilde{v}_{\tilde{x}_t} \in \partial f(\tilde{x}_t)$ where $\tilde{x}_t = \bar{x}_t - \lambda \al v_{\bar{x}_t}$ such that we have \begin{align*}
     f(\bar{x}_{t} - \alpha v_{\bar{x}_t}) - f(\bar{x}_t) 
    & = \langle \tilde{v}_{\tilde{x}_t}, -\al v_{\bar{x}_t}\rangle \nonumber \nonumber\\
    & = \langle v_{\bar{x}_t},-\al v_{\bar{x}_t} \rangle + \langle \tilde{v}_{\tilde{x}_t}-v_{\bar{x}_t},-\al v_{\bar{x}_t} \rangle \nonumber\\
    & \leqslant -\al\|v_{\bar{x}_t}\|^2 + \al\|\tilde{v}_{\tilde{x}_t} - v_{\bar{x}_t}\|\|v_{\bar{x}_t}\| \\
    & = - \al \|v_{\bar{x}_t}\|\left(\|v_{\bar{x}_t}\| - \|\tilde{v}_{\tilde{x}_t} - v_{\bar{x}_t}\|\right)
\end{align*} where we used the mean value theorem in the first equality.

We can now use our result from Section \ref{conv_proof_zero_subsection} to bound $\|\tilde{v}_{\tilde{x}_t} - v_{\bar{x}_t}\|$ and $\|v_{\bar{x}_t}\|$ from above and below, respectively. Observe that by Lemma \ref{closeness_of_subgrads}, we have \begin{align*}
    f(\bar{x}_t - \al v_{\bar{x}_t}) - f(\bar{x}_t) &  \leqslant -\al \|v_{\bar{x}_t}\|(\|v_{\bar{x}_t}\|-\|\tilde{v}_{\tilde{x}_t} - v_{\bar{x}_t}\|) \leqslant -\al\left(1 - \frac{5}{6}\right)\|v_{\bar{x}_t}\|^2 = -\frac{1}{6}\al\|v_{\bar{x}_t}\|^2
\end{align*} where we used \eqref{v_tilde_lipschitz} in the second inequality. But by our lower bound on $\|v_{\bar{x}_t}\|$, we have \begin{align*}
    f(\bar{x}_t - \al v_{\bar{x}_t}) - f(\bar{x}_t) & \leqslant -\frac{1}{6}\al\|v_{\bar{x}_t}\|^2 \leqslant -\al \frac{9K^2d^6\epsilon}{6 (2^{2d})}\|x_*\|^2
\end{align*} where we used \eqref{norm_v_lower_bound} in the second inequality. Hence there are at most $\frac{6f(x_0)(2^{2d})}{9K^2d^6 \al \epsilon \|x_*\|^2}$ iterations for which $x_t \notin \Seps_{\beta}$ where $x_0$ is the initial iterate of our algorithm. Thus there exists a natural number $N$ such that $N \leqslant \frac{6f(x_0)(2^{2d})}{9K^2d^6 \al \epsilon \|x_*\|^2}$ and $x_N \in \Seps_{\beta}$.

\end{proof}

\subsubsection{Convergence to neighborhood of $x_*$} \label{conv_proof_second_subsection}

We now show that if any iterate is in $\Seps_{\beta}$, then the negation step of the algorithm (Steps \ref{alg:cond1}--\ref{alg:endif}) ensures that our iterates will now be in a neighborhood of $x_*$ as opposed to a neighborhood of $-\rho_d x_*$. We will use the following result that $\Seps_{\beta}$ is contained in the union of neighborhoods of the true solution and a negative multiple thereof if $\beta$ is sufficiently small:

\begin{lem} \label{Seps_lemma}
If  $0 < 24\pi^2d^6\sqrt{\beta} \leqslant 1$, then \begin{align*}
    S_{\beta} \subset \mathcal{B}(x_*,70000\pi^2 d^9\beta\|x_*\|) \cup \mathcal{B}(-\rho_d x_*, 77422\pi^2d^{12}\sqrt{\beta}\|x_*\|).
 \end{align*}
\end{lem}

\noindent We also need the following lemma which shows that the objective function is smaller near $x_*$ than near $-\rho_d x_*$:

\begin{lem} \label{obj_function_bigger_near_negmult} Fix $0 < \epsilon < 1/(16\pi d^2)^2$ and suppose Assumption A3 is satisfied. Suppose that $A \in \R^{m \times n_d}$ satisfies the RRCP with respect to $\G$ with constant $\epsilon$ and $\G$ is such that each $W_i \in \R^{n_i \times n_{i-1}}$ satisfies the WDC with constant $\epsilon$. Then for any $\phi_d \in [\rho_d,1]$, we have that \begin{align}
    f(x) < f(y) \label{tweakbound}
\end{align} for all $x \in \mathcal{B}(\phi_d x_*, r_1d^{-12}\|x_*\|)$ and $y \in \mathcal{B}(-\phi_d x_*, r_1d^{-12}\|x_*\|)$ where $r_1$ is a universal constant.
\end{lem}

\noindent These results are proven in Section \ref{proofs_for_second_subsection}. The main result of this section is as follows.

\begin{lem} \label{conv_to_Sbetaplus_lemma} Suppose Assumptions A1-A3 are satisfied and set $\beta:= 4Kd^3\sqrt{\epsilon} + 11\|\eta\|2^{d/2}/\|x_*\|$ where $K$ is a numerical constant. Let $A \in \R^{m \times n}$ satisfy the RRCP with respect to $\G$ with constant $\epsilon$. Let $\G$ be such that $W_{i} \in \R^{n_i \times n_{i-1}}$ satisfy the WDC with constant $\epsilon$ for all $i \in [d]$. If $x_t \in \mathcal{S}_{\beta}$, then $\bar{x}_t \in \mathcal{S}_{\beta}^+$, i.e., \begin{align*}
    \|\bar{x}_t - x_*\| \leqslant C_5 d^{12}\sqrt{\epsilon}\|x_*\| + C_6 d^92^{d/2}\|\eta\|
\end{align*} where $C_5$ and $C_6$ are numerical constants.
\end{lem}
\begin{proof}[Proof of Lemma \ref{conv_to_Sbetaplus_lemma}]
Suppose $x_t \in \Seps_{\beta}$. We require $\beta$ to satisfy the assumption of Lemma \ref{Seps_lemma} and for $\Seps_{\beta}$ to be contained in the balls of radius $r_1d^{-12}\|x_*\|$ from Lemma \ref{obj_function_bigger_near_negmult}. Recall that by assumption $0 < \epsilon < c_1 d^{-102}$ and $\|\eta\| \leqslant \frac{c_2\|x_*\|}{2^{d/2}d^{48}}$ for some constants $c_1$ and $c_2$. Choosing $c_1$ and $c_2$ sufficiently small enough, we can have that \begin{align}
    \beta  = 4K d^3 \sqrt{\epsilon} + \frac{11\|\eta\|2^{d/2}}{\|x_*\|} & \leqslant \frac{4K\sqrt{c_1}}{d^{48}}+ \frac{11c_2}{d^{48}}  \leqslant\frac{r_1^2}{(77422\pi^2)^2d^{48}}. \nonumber
\end{align} Hence $\beta$ satisfies the assumptions of Lemma \ref{Seps_lemma} and $77422\pi^2d^{12}\sqrt{\beta}\|x_*\| \leqslant r_1d^{-12}\|x_*\|.$ Note that this implies $\Seps_{\beta}^+ \subset \mathcal{B}(x_*,r_1d^{-12}\|x_*\|)$ while $\Seps_{\beta}^- \subset \mathcal{B}(-\rho_d x_*,r_1d^{-12}\|x_*\|)$. Therefore, we can apply equation \eqref{tweakbound} in Lemma \ref{obj_function_bigger_near_negmult} so that for any $y \in \Seps_{\beta}^-$ and $x \in \Seps_{\beta}^+$, $f(x) < f(y)$. Since $x_t \in \Seps_{\beta}$, either $x_t \in \Seps_{\beta}^+$ or $x_t \in \Seps_{\beta}^-$. If $x_t \in \Seps_{\beta}^+$, then $f(x_t) < f(-x_t)$ so $\bar{x}_t = x_t \in \Seps_{\beta}^+$. Otherwise, $x_t \in \Seps_{\beta}^-$ and $-x_t \in \Seps_{\beta}^+$ so that $f(-x_t) < f(x_t)$ meaning $\bar{x}_t = -x_t \in \Seps_{\beta}^+$. In either case, we must have that $\bar{x}_t \in \Seps_{\beta}^+$. By the definition of $\Seps_{\beta}^+$, this establishes the inequality \begin{align*}
    \|\bar{x}_t - x_*\| \leqslant C_5 d^{12}\sqrt{\epsilon}\|x_*\| + C_6 d^92^{d/2}\|\eta\|
\end{align*} for some numerical constants $C_5$ and $C_6$.

%We further remark that if $\bar{x}_t \in \mathcal{B}(x_*,r)$, then the iterates for $t' \geqslant t$ will continue to stay in a ball of radius $2r$ around $x_*$. To see this, note that by Lemma \ref{neg_inner_prod_lower_bd_grad} and our assumptions on our step size and noise, our descent direction is sufficiently bounded by \begin{align*}
    %\al \|v_{x_{t'}}\| \leqslant \frac{C}{d}\max(\|x_{t'}\|,\|x_*\|)
%\end{align*} where $C$ is a numerical constant.
\end{proof}
\subsubsection{Convergence to $x_*$ up to noise} \label{conv_proof_third_subsection}

Finally, we show that once in a neighborhood of $x_*$, the iterates of our algorithm will converge to $x_*$ up to the noise level in the measurements. We will use the following convexity-like property around the minimizer: 
\begin{lem} \label{convexity_lemma} Fix $0 < \epsilon < 1/(200^4d^6 )$. Suppose that $A \in \R^{m \times n_d}$ satisfies the RRCP with respect to $\G$ with constant $\epsilon$ and $\G$ is such that each $W_i \in \R^{n_i \times n_{i-1}}$ satisfies the WDC with constant $\epsilon$ for $i \in [d]$. Then for all $x \in \mathcal{B}(x_*,d\sqrt{\epsilon}\|x_*\|)$ and any $v_x \in \partial f(x)$, we have \begin{align*}
\left\| v_x - \frac{1}{2^d} (x - x_*) \right\| \leqslant \frac{1}{8}\frac{1}{2^d}\|x- x_*\| + \frac{2}{2^{d/2}}\|\eta\|.
\end{align*}
\end{lem}

\noindent We now prove the following lemma.

\begin{lem} \label{conv_to_xstar_up_to_noise_lem}
Suppose Assumptions A1-A3 are satisfied and set $\beta:= 4Kd^3\sqrt{\epsilon} + 11\|\eta\|2^{d/2}/\|x_*\|$ where $K$ is a numerical constant. Let $A \in \R^{m \times n}$ satisfy the RRCP with respect to $\G$ with constant $\epsilon$. Let $\G$ be such that $W_{i} \in \R^{n_i \times n_{i-1}}$ satisfy the WDC with constant $\epsilon$ for all $i \in [d]$. Suppose $x_N \in \Seps_{\beta}^+$ for some $N \in \N$. Then for all $t \geqslant N$, we have that $\bar{x}_t \in \mathcal{B}(x_*, r_1d^{-12}\|x_*\|)$, $\bar{x}_t = x_t$, and \begin{align*}
    \|x_{t+1} - x_*\| \leqslant \tau^{t+1-N}\|x_N - x_*\| + \vartheta\frac{2^{d/2}}{d^2}\|\eta\|
\end{align*} where $\tau := 1 - \frac{7}{8}\frac{\al}{2^d} \in (0,1)$ and $\vartheta := \frac{2c_3}{1-\tau}$.
\end{lem}

\begin{proof}[Proof of Lemma \ref{conv_to_xstar_up_to_noise_lem}]

Suppose $t = N$ so we have that $\bar{x}_t = x_t \in \Seps_{\beta}^+ \subset \mathcal{B}(x_*,r_1d^{-12}\|x_*\|)$. As shown in Lemma \ref{conv_to_Sbetaplus_lemma}, this inclusion holds by our assumptions on $\epsilon$ and $\eta$. By Assumption A1, the requirements of Lemma \ref{convexity_lemma} are met. Observe that for any $v_{\bar{x}_t} \in \partial f(\bar{x}_t)$, we have \begin{align}
    \|x_{t+1}-x_*\| & = \left\|\bar{x}_t - \al v_{\bar{x}_t} - x_* + \frac{\al}{2^{d}}(\bar{x}_t-x_*) - \frac{\al}{2^{d}}(\bar{x}_t-x_*)\right\| \nonumber \\
    & \leqslant \left(1- \frac{\al}{2^d}\right)\|\bar{x}_t-x_*\| + \al\left\|v_{\bar{x}_t} - \frac{1}{2^d}(\bar{x}_t- x_*)\right\| \nonumber \\
    & \leqslant \left(1 - \frac{\al}{2^d}\right)\|\bar{x}_t-x_*\| + \left(\frac{\al}{8}\right)\frac{1}{2^d}\|\bar{x}_t-x_*\| + \al\frac{2}{2^{d/2}}\|\eta\| \nonumber\\
    & = \left(1 - \frac{7}{8}\frac{\al}{2^d}\right)\|\bar{x}_t-x_*\|+ \al\frac{2}{2^{d/2}}\|\eta\| \label{xtp1_bound}
\end{align} where we used Lemma \ref{convexity_lemma} in the second inequality. Using $\al \leqslant c_3\frac{2^d}{d^2}$ and $\|\eta\| \leqslant \frac{c_2}{2^dd^{48}}\|x_*\|$ for sufficiently small constants $c_2$ and $c_3$, we have that if $\bar{x}_t \in \mathcal{B}(x_*,r_1d^{-12}\|x_*\|)$, then $x_{t+1} \in \mathcal{B}(x_*,r_1d^{-12}\|x_*\|)$ so the iterates stay within a small ball around the minimizer. Hence Lemma \ref{obj_function_bigger_near_negmult} yields $\bar{x}_{t+1} = x_{t+1}$. Repeatedly applying the above logic shows that for all $t \geqslant N$, $x_t \in \mathcal{B}(x_*, r_1d^{-12}\|x_*\|)$ and $\bar{x}_t = x_t$.

Finally, using $\al \leqslant c_3\frac{2^d}{d^2}$ in the second half of equation \eqref{xtp1_bound} yields\begin{align*}
    \|x_{t+1} - x_*\| \leqslant \left(1 - \frac{7}{8}\frac{\al}{2^d}\right)\|x_t-x_*\| + 2c_3\frac{2^{d/2}}{d^2}\|\eta\| =: \tau\|x_t-x_*\|+ 2c_3\frac{2^{d/2}}{d^2}\|\eta\|
\end{align*} where $\tau := 1 - \frac{7}{8}\frac{\al}{2^d}$. Choosing $c_3$ so that $c_3 < \frac{8}{7}$ implies $\tau \in (0,1)$. Starting at $t = N$ and repeatedly applying this inequality, we attain \begin{align*}
   \|x_{t+1} - x_*\| & \leqslant \tau^{t+1-N}\|x_N-x_*\| + (\tau^{t-N} + \tau^{t-N-1} + \dots + 1)2c_3 \frac{2^{d/2}}{d^2}\|\eta\|\\
    & \leqslant \tau^{t+1-N}\|x_N-x_*\| + \frac{2c_3}{1-\tau}\frac{2^{d/2}}{d^2}\|\eta\| \\
    & =: \tau^{t+1-N}\|x_N-x_*\| + \vartheta\frac{2^{d/2}}{d^2}\|\eta\|
\end{align*} where $\vartheta := \frac{2c_3}{1 - \tau}.$ This completes the proof. \end{proof}

\subsubsection{Final proof of Theorem \ref{main_deterministic_conv_result}} \label{final_proof_subsection}

With all of the necessary lemmas proven, we bring them together to prove Theorem \ref{main_deterministic_conv_result}.

\begin{proof}[Proof of Theorem \ref{main_deterministic_conv_result}]

Set $\beta:= 4Kd^3\sqrt{\epsilon} + 11\|\eta\|2^{d/2}/\|x_*\|$ where $K$ is a numerical constant. By Proposition \ref{bounded_iterates}, we may assume that our initial iterate $x_0 \notin \mathcal{B}(0,c_0\|x_*\|)$ for some numerical constant $c_0$. Then by Lemma \ref{conv_to_neighborhoods_lemma}, there exists an $N \in \N$ such that $N \leqslant \frac{6f(x_0)(2^{2d})}{9K^2d^6 \al \epsilon \|x_*\|^2}$ and $x_N \in \Seps_{\beta}$. Then Lemma \ref{conv_to_Sbetaplus_lemma} implies $\bar{x}_N \in \Seps_{\beta}^+$ which establishes inequality \eqref{bound_on_Nth_iterate}. Finally, Lemma \ref{conv_to_xstar_up_to_noise_lem} establishes inequality \eqref{final_iterate_bound} for any $t \geqslant N$. Inequality \eqref{final_output_through_genmodel_bound} follows by using \eqref{final_iterate_bound} and the following result with $j=d$ which established Lipschitz continuity of $\G$ for $x$ within a neighborhood of $x_*$: \begin{lem}[Lemma A.8 in \cite{Huangetal2018}] \label{G_lipschitz}
Suppose $0 < \epsilon < 1/(200^4 d^6)$, $x \in \mathcal{B}(x_*,d\sqrt{\epsilon}\|x_*\|)$, and $\G$ is such that each $W_i \in \R^{n_i \times n_{i-1}}$ satisfies the WDC with constant $\epsilon$ for $i \in [d]$. Then we have that for all $j \in [d]$, \begin{align*}
    \left\| \Pi_{i=j}^1 W_{i,+,x}x - \Pi_{i=j}^1 W_{i,+,x_*}x_*\right\| \leqslant \frac{1.2}{2^{j/2}}\|x-x_*\|.
\end{align*}
\end{lem} \end{proof}

\subsection{Supplementary Results} \label{additional_results_sec}

In the following sections, we provide proofs for auxillary results that were in used in the four main lemmas used to establish Theorem \ref{main_deterministic_conv_result}. Section \ref{bounded_iterates_subsection} focuses on proving that after a polynomial number of iterations, the iterates of our algorithm are all bounded away from zero. Section \ref{proofs_for_controlling_v_subsection} establishes supplementary results about controlling subgradients in Section \ref{conv_proof_zero_subsection}. Then Section \ref{proofs_for_second_subsection} establishes results concerning the zeros of $h_x$ and properties of the objective function used in Section \ref{conv_proof_second_subsection}. Lastly, Section \ref{proof_of_convexity_lemma_subsection} focuses on establishing the convexity-like property near the minimizer which is formalized in Section \ref{conv_proof_third_subsection}.

\subsubsection{Iterates are eventually bounded away from zero} \label{bounded_iterates_subsection}

 We focus on proving the following proposition:
 \begin{prop}  \label{bounded_iterates} Fix $\epsilon > 0$ such that $Kd^3 \sqrt{\epsilon} \leqslant 1$ where $K$ is a universal constant. Suppose that $A \in \R^{m \times n_d}$ satisfies the RRCP with respect to $\G$ with constant $\epsilon$ and $\G$ is such that each $W_i \in \R^{n_i \times n_{i-1}}$ satisfies the WDC with constant $\epsilon$ for $i \in [d]$. Suppose that the step size $\al$ and noise $\eta$ satisfy $0 < \al < \frac{2^d}{104\pi (Cd+2c_2)}$ and $\|\eta\| \leqslant \frac{c_2\|x_*\|}{2^{d/2}}$ where $C$ and $c_2$ are numerical constants. If $x_t \in \mathcal{B}(0,\frac{1}{52\pi}\|x_*\|)$, then after at most $N_0 = \left\lceil\left(\frac{2^d24\pi}{\al 52\pi}\right)^2\right\rceil$ iterations, we have that for all $t > N_0$ and $\lambda \in [0,1]$, $\lambda \bar{x}_t + (1-\lambda)x_{t+1} \notin \mathcal{B}(0, \frac{1}{104\pi}\|x_*\|).$
\end{prop}

 This result asserts that if an iterate of our algorithm lies within a ball of the origin, then after a polynomial number of steps, it will leave this region. To prove it, we require the following lemma that establishes certain properties of any subgradient $v_x \in \partial f(x)$ for points $x$ near the origin: \begin{lem} \label{neg_inner_prod_lower_bd_grad}
Fix $\epsilon > 0$ such that $Kd^3 \sqrt{\epsilon} \leqslant 1$ where $K$ is a universal constant. Suppose that $A \in \R^{m \times n_d}$ satisfies the RRCP with respect to $\G$ with constant $\epsilon$ and $\G$ is such that each $W_i \in \R^{n_i \times n_{i-1}}$ satisfies the WDC with constant $\epsilon$ for $i \in [d]$. Then for all $x \in \mathcal{B}(0,\frac{1}{52\pi}\|x_*\|)$ and any $v_x \in \partial f(x)$, we have that \begin{align}
\langle x,v_x \rangle < 0\ \text{and}\ \|v_x\| \geqslant \frac{1}{2^d}\frac{1}{24\pi}\|x_*\|. \nonumber
\end{align} 
\end{lem} \noindent We are now ready to proceed with a proof of Proposition \ref{bounded_iterates}.

\begin{proof}[Proof of Proposition \ref{bounded_iterates}]
Suppose that $x_t \in \mathcal{B}(0,\frac{1}{52\pi}\|x_*\|)$. By Lemma \ref{neg_inner_prod_lower_bd_grad}, we have that $\bar{x}_t$ and the next iterate $x_{t+1} = \bar{x}_t - \al v_{\bar{x}_t}$ form an obtuse triangle for any $v_{\bar{x}_t} \in \partial f(\bar{x}_t)$. Thus \begin{align*}
   \|\bar{x}_{t+1}\|^2 = \|x_{t+1}\|^2 & \geqslant \|\bar{x}_t\|^2 + \al^2 \|v_{\bar{x}_t}\|^2 \\
    & \geqslant \|\bar{x}_t\|^2 + \al^2 \frac{1}{(2^d24\pi)^2}\|x_*\|^2
\end{align*} where the last inequality follows from Lemma \ref{neg_inner_prod_lower_bd_grad}. Thus the norm of the iterates will increase until after $N_0 = \left\lceil\left(\frac{2^d24\pi}{\al 52\pi}\right)^2\right\rceil$ iterations we have $x_{t+N_0} \notin \mathcal{B}(0,\frac{1}{52\pi}\|x_*\|)$. 

Now consider $x_t \notin \mathcal{B}(0, \frac{1}{52\pi}\|x_*\|)$. We will show that for any $\lambda \in [0,1]$, $\lambda \bar{x}_t + (1-\lambda)x_{t+1} \notin \mathcal{B}(0, \frac{1}{104\pi}\|x_*\|)$. Note that for $x_t \notin \mathcal{B}(0, \frac{1}{52\pi}\|x_*\|)$, we have $\|\bar{x}_t\| = \|x_t\| \geqslant \frac{1}{52\pi}\|x_*\|$. Then observe that for any $v_{\bar{x}_t} \in \partial f(\bar{x}_t)$, we have \begin{align}
    \al \|v_{\bar{x}_t}\| & \leqslant \al \frac{1}{2^d}\max(\|\bar{x}_t\|,\|x_*\|)\left(Cd + \frac{2}{\|x_*\|}2^{d/2}\|\eta\|\right) \nonumber \\
    & \leqslant \al \frac{1}{2^d}\max(\|\bar{x}_t\|,\|x_*\|)\left(Cd + 2c_2\right) \nonumber \\
    & \leqslant \frac{\al}{2^d}52\pi\|\bar{x}_t\|\left(Cd + 2c_2\right) \nonumber \\
    & \leqslant \frac{1}{2}\|\bar{x}_t\| \nonumber
\end{align} where the first inequality follows by Lemma \ref{bound_on_norm_of_v}, the second by the assumption on the noise energy $\|\eta\| \leqslant \frac{c_2\|x_*\|}{2^{d/2}}$, the third due to $x_t \notin \mathcal{B}(0,\frac{1}{52\pi}\|x_*\|)$, and the last inequality follows by the assumption on $\al$. Thus since $x_{t+1} = \bar{x}_t - \al v_{\bar{x}_t}$, we have that $\lambda \bar{x}_t + (1-\lambda) x_{t+1} \notin \mathcal{B}(0,\frac{1}{104\pi}\|x_*\|)$ for any $\lambda \in [0,1].$
\end{proof}

We now focus on proving Lemma \ref{neg_inner_prod_lower_bd_grad}. To show this, we first require the following angle concentration property of the map $z \mapsto A_zz$ for $z$ in the range of $\G$.  \begin{lem} \label{angle_distortion_property}
Fix $0 < \epsilon < 1/(4L)$ where $L$ is the universal constant specified in the RRCP. Let $A \in \R^{m \times n}$ satisfy the RRCP with respect to $\G$ with constant $\epsilon$. Let $\G$ be such that $W_{i} \in \R^{n_i \times n_{i-1}}$ satisfy the WDC with constant $\epsilon$ for all $i \in [d]$. Then for all $x,y \in \R^k \setminus \{0\}$, the angle $\theta_1 : = \angle(A_{\G(x)}\G(x) ,A_{\G(y)}\G(y))$ is well-defined and $$|\cos \theta_1 - \cos \varphi(\theta_d)| \leqslant 4L\epsilon$$ where $\theta_d = \angle(\G(x),\G(y))$ and $\varphi : \R \rightarrow \R$ defined by \begin{align*}
    \varphi(\theta):= \cos^{-1}\left(\frac{(\pi-2\theta)\cos\theta + 2\sin\theta}{\pi}\right).
\end{align*}
\end{lem}

\begin{proof}[Proof of Lemma \ref{angle_distortion_property}] Fix $x,y \in \R^k \setminus \{0\}$. We use the shorthand notation $\Lambda_x: = \PiWdix$ and $x_d:=\Lambda_x x$. Note that the WDC implies that for sufficiently small $\epsilon$, we have that $\Lambda_x x, \Lambda_y y \neq 0$. Hence we may assume, without loss of generality, that $\|\Lambda_x x\| = \|\Lambda_y y\| = 1.$ Now define the following quantities: \begin{align*}
\delta_1 & : = \langle \Lambda_x x, (A_{x_d}^\top A_{y_d} - \Phi_{x_d,y_d})\Lambda_y y \rangle, \\
\delta_2 & := \langle \Lambda_x x, (A_{x_d}^\top A_{x_d} - I)\Lambda_x x \rangle, \\
\delta_3 & : = \langle \Lambda_y y ,(A_{y_d}^\top A_{y_d} - I)\Lambda_y y \rangle.
\end{align*} Observe that by the RRCP, we have that $\max_{i = 1,2,3} |\delta_i | \leqslant L\epsilon.$ Hence if $0 <\epsilon < 1/L$, \begin{align*}
0 < 1-L\epsilon \leqslant \|A_{x_d}\Lambda_x x\|^2
\end{align*} so $\|A_{x_d}\Lambda_x x\|\neq 0$. The same conclusion holds for $\|A_{y_d}\Lambda_y y\|$ so $\theta_1$ is well-defined.
Furthermore, note that \begin{align*}
\cos \theta_1 & = \frac{\langle \Lambda_x x, A_{x_d}^\top A_{y_d}\Lambda_y y \rangle}{\|A_{x_d}\Lambda_x x\|\|A_{y_d}\Lambda_y y\|} \\
& = \frac{\langle \Lambda_x x, A_{x_d}^\top A_{y_d}\Lambda_y y \rangle}{\sqrt{\langle A_{x_d}\Lambda_x x,A_{x_d}\Lambda_x x\rangle \langle A_{y_d}\Lambda_y y,A_{y_d}\Lambda_y y \rangle}} \\
& = \frac{\langle \Lambda_x x, \Phi_{x_d,
y_d}\Lambda_y y \rangle + \delta_1}{\sqrt{\left(\langle \Lambda_x x, \Lambda_x x\rangle + \delta_2\right)\left(\langle \Lambda_y y, \Lambda_y y\rangle + \delta_3\right)}} \\
& = \frac{\langle \Lambda_x x, \Phi_{x_d,y_d}\Lambda_y y \rangle + \delta_1}{\sqrt{\left(1 + \delta_2\right)\left( 1 + \delta_3\right)}}.
\end{align*} Thus \begin{align*}
\left|\cos \theta_1 - \langle \Lambda_x x, \Phi_{x_d,y_d}\Lambda_y y \rangle \right| & \leqslant \left| \frac{\langle \Lambda_x x, \Phi_{x_d,y_d}\Lambda_y y \rangle + \delta_1}{\sqrt{\left(1 + \delta_2\right)\left( 1 + \delta_3\right)}} - \langle \Lambda_x x, \Phi_{x_d,
y_d}\Lambda_y y \rangle \right| \\
& \leqslant \left|\langle \Lambda_x x, \Phi_{x_d,y_d}\Lambda_y y \rangle\right| \left|1 - \frac{1}{\sqrt{\left(1 + \delta_2\right)\left( 1 + \delta_3\right)}}\right| \\
& + \frac{|\delta_1|}{\sqrt{\left(1 + \delta_2\right)\left( 1 + \delta_3\right)}} \\
& \leqslant 2\left|1 - \frac{1}{1 - L\epsilon}\right| + \frac{L\epsilon}{1 - L\epsilon} \\
& \leqslant \frac{3L\epsilon}{1 - L\epsilon} \leqslant 4L\epsilon
\end{align*} where we used $\|\Phi_{x_d,y_d}\| \leqslant 2$ in the third inequality and $L\epsilon < 1/4$ in the last inequality. The proof concludes by noting that $\langle \Lambda_x x, \Phi_{x_d,y_d}\Lambda_y y \rangle = \frac{1}{\pi}[(\pi - 2\theta_d)\cos \theta_d + 2\sin\theta_d].
$
\end{proof}

We also require upper bounds on quantities that will be useful throughout the remaining proofs.

\begin{lem} \label{upper_bounds_on_lambda_Az} Fix $0 < \epsilon < 1/(48d)$. Let $A \in \R^{m \times n}$ satisfy the RRCP with respect to $\G$ with constant $\epsilon$. Let $\G$ be such that $W_{i} \in \R^{n_i \times n_{i-1}}$ satisfy the WDC with constant $\epsilon$ for all $i \in [d]$. Then for any $x \in \R^k$, we have \begin{align}
    \|\Lambda_x\|^2 & \leqslant \frac{13}{12}2^{-d}, \label{bound_on_lambda_energy} \\
    \|A_{x_d}\Lambda_x\|^2 & \leqslant (1 + L\epsilon)\|\Lambda_x\|^2 \label{bound_on_AxdLambdax}
\end{align}

\end{lem}

\begin{proof}[Proof of Lemma \ref{upper_bounds_on_lambda_Az}]

For equation \eqref{bound_on_lambda_energy}, note that the WDC implies that $\|W_{i,+,x}\|^2 \leqslant \frac{1}{2}+\epsilon$ for each $i \in [d]$ so \begin{align*}
    \|\Lambda_x\|^2 \leqslant \prod_{i=1}^d \|W_{i,+,x}\|^2 \leqslant \left(\frac{1}{2} + \epsilon\right)^d = \frac{1}{2^d}(1 + 2\epsilon)^d = \frac{1}{2^d}e^{d\log(1+2\epsilon)}\leqslant \frac{1 + 4\epsilon d}{2^d} \leqslant \frac{13}{12}2^{-d}
\end{align*} where we used the fact that $\log(1 + u) \leqslant u$ and $e^{u} \leqslant 1 + 2u$ for $u < 1$ while the last inequality follows by our assumption on $\epsilon$: $\epsilon < 1/(48d)$.

For equation \eqref{bound_on_AxdLambdax}, observe that by the RRCP and the local linearity of $\G$, we have that for sufficiently small $z \in \R^k$, \begin{align*}
    \left|\langle A_{x_d}\Lambda_x z, A_{x_d} \Lambda_x z\rangle - \langle \Lambda_x z, \Lambda_x z \rangle\right| \leqslant L\epsilon \|\Lambda_x \|^2\|z\|^2
\end{align*} which implies that $$\left|\langle A_{x_d}\Lambda_x z, A_{x_d} \Lambda_x z\rangle\right| \leqslant (1 + L\epsilon) \|\Lambda_x \|^2\|z\|^2.$$ Since this holds for any $z \in \R^k$, we conclude $\|A_{x_d}\Lambda_x\|^2 \leqslant (1+L\epsilon)\|\Lambda_x\|^2$.
\end{proof}

Now we set out to prove Lemma \ref{neg_inner_prod_lower_bd_grad}.

\begin{proof}[Proof of Lemma \ref{neg_inner_prod_lower_bd_grad}] Suppose $f$ is differentiable at $x$ so that $v_x$ is precisely the gradient of $f$. We first show that $\langle x, v_x\rangle < 0$. Note that \begin{align}
\langle x,v_x \rangle & = \underbrace{\langle \Lambda_x^\top A_{x_d}^\top A_{x_d} \Lambda_x x,x \rangle}_{(I)} - \underbrace{\langle \Lambda_x^\top A_{x_d}^\top A_{x_{*,d}}\Lambda_{x_*}x_*,x \rangle}_{(II)} - \underbrace{\langle \Lambda_x^\top A_{x_d}^\top \eta,x\rangle}_{(III)}. \nonumber
\end{align} We will bound the first and third term from above and the second from below. We first focus on the second term as its proof will give us a result for the first term. 

\paragraph{(II):}For the second term, note that we can write it as \begin{align}
\langle \Lambda_x^\top A_{x_d}^\top A_{x_{*,d}}\Lambda_{x_*}x_*,x \rangle = \cos(\angle(A_{x_d}x_d,A_{x_{*,d}}x_{*,d}))\|A_{x_d}x_d\|\|A_{x_{*,d}}x_{*,d}\|. \nonumber
\end{align} By Lemma \ref{angle_distortion_property}, we have that \begin{align}
\cos(\varphi(\theta_d)) - 4L\epsilon \leqslant \cos(\angle(A_{x_d}x_d,A_{x_{*,d}}x_{*,d})) \leqslant \cos(\varphi(\theta_d)) + 4L\epsilon \nonumber
\end{align} where $\theta_d = \angle(x_d,x_{*,d})$. Thus \begin{align}
\langle \Lambda_x^\top A_{x_d}^\top A_{x_{*,d}}\Lambda_{x_*}x_*,x \rangle & \geqslant (\cos(\varphi(\theta_d))-4L \epsilon)\|A_{x_d}x_d\|\|A_{x_{*,d}}x_{*,d}\|. \label{bound_on_inner_prod}
\end{align} However, note that \begin{align}
\cos(\varphi(\theta)) = \frac{(\pi - 2\theta)\cos \theta + 2\sin \theta}{\pi} \geqslant \frac{2}{\pi}\ \forall\ \theta \in [0,\pi]. \label{bound_on_angle}
\end{align} Hence if $\epsilon < 1/(4L\pi)$, applying \eqref{bound_on_angle} to \eqref{bound_on_inner_prod} we have that
\begin{align}
\langle \Lambda_x^\top A_{x_d}^\top A_{x_{*,d}}\Lambda_{x_*}x_*,x \rangle \geqslant \frac{1}{\pi}\|A_{x_d}x_d\|\|A_{x_{*,d}}x_{*,d}\|. \label{inner_prod_lower_bound_w_pi}
\end{align} We now bound $\|A_{x_d}x_d\|$: observe that by the RRCP, \begin{align}
|\langle (A_{x_d}^\top A_{x_d} - I)x_d,x_d\rangle| \leqslant L\epsilon\|x_d\|^2 \Longrightarrow 
(1-L\epsilon)\|x_d\|^2 & \leqslant \|A_{x_d}x_d\|^2 \leqslant (1 + L\epsilon)\|x_d\|^2 \nonumber
\end{align} which gives \begin{align}
\sqrt{1-L\epsilon}\|x_d\| & \leqslant \|A_{x_d}x_d\| \leqslant \sqrt{1 + L\epsilon}\|x_d\|. \nonumber
\end{align} By equation (11) of \cite{HV2017}, we have that \begin{align}
\left(\frac{1}{2}-\epsilon\right)^{d/2}\|x\| \leqslant \|x_d\| \leqslant \left(\frac{1}{2} + \epsilon\right)^{d/2}\|x\|. \nonumber
\end{align} Hence we attain \begin{align}
\sqrt{1 - L\epsilon}\left(\frac{1}{2} - \epsilon\right)^{d/2}\|x\| \leqslant \|A_{x_d}x_d\| \leqslant \sqrt{1 + L\epsilon}\left(\frac{1}{2} + \epsilon\right)^{d/2}\|x\|. \label{bound_on_norm_of_Axd}
\end{align} Analogous bounds hold for $\|A_{x_*,d}x_{*,d}\|$. Applying \eqref{bound_on_norm_of_Axd} to equation \eqref{inner_prod_lower_bound_w_pi}, we conclude that
\begin{align}
\langle \Lambda_x^\top A_{x_d}^\top A_{x_{*,d}}\Lambda_{x_*}x_*,x \rangle & \geqslant \frac{1}{\pi}\|A_{x_d}x_d\|\|A_{x_{*,d}}x_{*,d}\|  \nonumber \\
& \geqslant \frac{1}{\pi}(1-L\epsilon)\left(\frac{1}{2}-\epsilon\right)^d\|x\|\|x_*\|. \nonumber
\end{align} If $2d\epsilon < 2/3$, we further have $(1/2-\epsilon)^d \geqslant (1-2d\epsilon)/2^d\geqslant 1/3(1/2^d).$ Then if $\epsilon$ is chosen such that $1 - L\epsilon \geqslant 1/2$, then we get \begin{align}
\langle \Lambda_x^\top A_{x_d}^\top A_{x_{*,d}}\Lambda_{x_*}x_*,x \rangle & \geqslant \frac{1}{6\pi}\frac{1}{2^d}\|x\|\|x_*\|. \label{final_inner_prod_lower_bound}
\end{align} This concludes the bound of the second term. We then proceed to bounding (I) and (III).

\paragraph{(I):} Observe that by equation \eqref{bound_on_norm_of_Axd} and our choice of $\epsilon$, we get \begin{align}
\langle \Lambda_x^\top A_{x_d}^\top A_{x_d} \Lambda_x x,x \rangle = \|A_{x_d}x_d\|^2 \leqslant (1+L\epsilon)\left(\frac{1}{2}+\epsilon\right)^d\|x\|^2 \leqslant 2 \cdot\frac{13}{12}\frac{1}{2^d}\|x\|^2 = \frac{13}{6}\frac{1}{2^d}\|x\|^2. \nonumber
\end{align} 

\paragraph{(III):}  Observe that \begin{align}
    \|A_{x_d}\Lambda_x\| \leqslant \sqrt{1 + L\epsilon}\|\Lambda_x\| \leqslant \sqrt{\frac{13}{12}(1+ L\epsilon)}\frac{1}{2^{d/2}} \leqslant \frac{2}{2^{d/2}} \label{bound_on_AxdLambdax_with_2}
\end{align} where we used \eqref{bound_on_AxdLambdax} in the first inequality,  \eqref{bound_on_lambda_energy} in the second inequality and our assumption on $\epsilon$ in the last inequality. Thus we attain \begin{align*}
    |\langle x, \Lambda_x^\top A_{x_d}^\top \eta \rangle | \leqslant \|x\|\|A_{x_d}\Lambda_x\|\|\eta\|\leqslant \frac{2}{2^{d/2}}\|\eta\|\|x\| \leqslant \frac{2c_2}{2^d}\|x\|\|x_*\| \leqslant \frac{1}{2^d}\frac{1}{12\pi}\|x\|\|x_*\|
\end{align*} where the third inequality follows by $\|\eta\| \leqslant \frac{c_2}{2^{d/2}}\|x_*\|$ and the last inequality is due to $c_2 < \frac{1}{24\pi}$.

Using our results for (I), (II), and (III), we conclude that \begin{align}
\langle x,v_x\rangle & = \langle \Lambda_x^\top A_{x_d}^\top A_{x_d} \Lambda_x x,x \rangle - \langle \Lambda_x^\top A_{x_d}^\top A_{x_{*,d}}\Lambda_{x_*}x_*,x \rangle - \langle \Lambda_x^\top A_{x_d}^\top \eta,x\rangle \nonumber \\
& \leqslant \frac{1}{2^d}\|x\|\left(\frac{13}{6}\|x\| + \frac{1}{12\pi}\|x_*\|- \frac{1}{6\pi}\|x_*\|\right) \nonumber \\
& \leqslant \frac{1}{2^d}\|x\|\left(\frac{13}{6}\|x\| - \frac{1}{12\pi}\|x_*\|\right) \nonumber .
\end{align} Thus if $\|x\| < \frac{1}{52\pi}\|x_*\|$, i.e. $x \in \mathcal{B}(0,\frac{1}{52\pi}\|x_*\|)$, then \begin{align}
\langle x,v_x \rangle \leqslant - \frac{1}{2^d}\frac{1}{24\pi}\|x\|\|x_*\| < 0. \label{inner_prod_neg_bound_w_constant}
\end{align} Lastly, observe that this gives \begin{align*}
    \left\langle -\frac{x}{\|x\|}, v_x \right\rangle \geqslant \frac{1}{2^d}\frac{1}{24\pi}\|x_*\|.
\end{align*} But by the Cauchy-Schwarz inequality, $\left\langle -\frac{x}{\|x\|},v_x\right\rangle \leqslant \|v_x\|$ so we obtain \begin{align}
    \|v_x\| \geqslant \frac{1}{2^d}\frac{1}{24\pi}\|x_*\|. \label{lower_bound_on_v_norm}
\end{align}

When $f$ is not differentiable at $x$, we have that by equation \eqref{subdifferential_definition}, we can write $v_x = \sum_{\ell=1}^{s}c_{\ell}v_{\ell}$ where $c_{\ell} \geqslant 0$ for $\ell \in [s]$ and $\sum_{\ell=1}^s c_{\ell} = 1$.  Applying our result for differentiable points $x$, we have that  \begin{align*}
     \langle x,v_x \rangle & = \sum_{\ell=1}^{s}c_{\ell}\langle x,v_{\ell}\rangle \leqslant - \frac{1}{2^d}\frac{1}{24\pi}\|x\|\|x_*\| \sum_{\ell=1}^{s}c_{\ell}  = - \frac{1}{2^d}\frac{1}{24\pi}\|x\|\|x_*\| < 0.
 \end{align*} For the lower bound on the norm of $v_x$, note that by the Cauchy-Schwarz inequality and equation \eqref{inner_prod_neg_bound_w_constant}, we have that \begin{align*}
     \|v_x\| = \max_{\|u\| = 1}\langle v_x,u\rangle & \geqslant \left\langle v_x, -\frac{x}{\|x\|} \right\rangle = \sum_{\ell=1}^s c_{\ell}\left\langle v_{\ell}, - \frac{x}{\|x\|}\right\rangle \geqslant \frac{1}{2^d}\frac{1}{24\pi}\|x_*\|\sum_{\ell=1}^s c_{\ell}  = \frac{1}{2^d}\frac{1}{24\pi}\|x_*\|
 \end{align*} as desired.
\end{proof}

\subsubsection{Proofs for Section \ref{conv_proof_zero_subsection}} \label{proofs_for_controlling_v_subsection}

In this section, we focus on results that aided in establishing Lemma \ref{closeness_of_subgrads} in Section \ref{conv_proof_zero_subsection}.  The first result concerns a bound on the norm of our descent direction (Lemma \ref{bound_on_norm_of_v}). The second is that $h_x$ is Lipschitz with respect to $x \in \R^k$ outside of a ball of the origin (Lemma \ref{h_lipschitz_lemma}) and the third is that for all $x \in \R^k$, $h_x$ approximates any $v_x \in \partial f(x)$ (Lemma \ref{conc_v_to_h}). Prior to beginning the proof of Lemma \ref{bound_on_norm_of_v}, we outline some notation. For $x \neq 0$, set $\psi_{d,x} : = \frac{\pi - 2\overline{\theta}_{d,x}}{\pi}$, and $\zeta_{j+1,x} : = \prod_{i=j}^{d-1}\frac{\pi - \overline{\theta}_{j+1,x}}{\pi}$.  Based on this notation, $h_x$ can be written as $$ h_x = \frac{1}{2^d}\left[\psi_{d,x}\zeta_{0,x} \|x_*\|\hat{x}_* + \left(\|x\|-\|x_*\|\left(\frac{2\sin \overline{\theta}_{d,x}}{\pi} + \psi_{d,x}\sum_{i=0}^{d-1}\frac{\sin \overline{\theta}_{i,x}}{\pi}\zeta_{i+1,x}\right)\right)\Hat{x}\right].$$

In the remaining proofs, a number of results concerning properties of $\overline{\theta}_{i,x}$ and $\breve{\theta}_{i}$ will be useful. The following lemma records these results: \begin{lem}[Bounds from Lemma 10 in \cite{HV2017}] For $x \neq 0$, let $\overline{\theta}_{0,x} := \angle(x,x_*)$ and  $\overline{\theta}_{i,x} := g(\overline{\theta}_{i-1,x})$ for $i \in [d]$. Let $\breve{\theta}_0 := \pi$ and $\breve{\theta}_i = g(\breve{\theta}_{i-1})$ for $i \in [d]$. Then the following all hold: \begin{align}
\left|\prod_{i=0}^{d-1}\frac{\pi - \overline{\theta}_{i,x}}{\pi}\right| & \leqslant 1, \label{piprod_upperbound}\\
\prod_{i=0}^{d-1}\frac{\pi - \overline{\theta}_{i,x}}{\pi} & \geqslant \frac{\pi - \overline{\theta}_{0,x}}{\pi d^3}, \label{piprod_lowerbound}\\
\left|\sum_{i=0}^{d-1} \frac{\sin \overline{\theta}_{i,x}}{\pi}\left(\prod_{j=i+1}^{d-1} \frac{\pi - \overline{\theta}_{j,x}}{\pi}\right)\right| & \leqslant \frac{d}{\pi}\sin \overline{\theta}_{0,x}, \label{sum_prod_upperbound}\\
\overline{\theta}_{0,x} = \pi + O_1(\delta) & \Longrightarrow \overline{\theta}_{i,x} = \breve{\theta}_i + O_1(i\delta), \label{pi_Oone_bound}\\
\overline{\theta}_{0,x} = \pi + O_1(\delta) & \Longrightarrow \left|\prod_{i=0}^{d-1}\frac{\pi - \overline{\theta}_{i,x}}{\pi}\right| \leqslant \frac{\delta}{\pi}, \label{beta_pi_upperbound} \\
\left|\frac{\pi - 2\overline{\theta}_{i,x}}{\pi}\right| & \leqslant 1\ \forall\ i \geqslant 0, \label{pi_twotheta_upperbound} \\
\overline{\theta}_{d,x}& \leqslant \cos^{-1}\left(\frac{1}{\pi}\right)\ \forall\ d \geqslant 2, \label{cos_inv_bound} \\
\breve{\theta}_i & \leqslant \frac{3\pi}{i + 3}\ \forall\ i \geqslant 0,  \label{breve_theta_bound} \\
\breve{\theta}_i &\geqslant \frac{\pi}{i + 1}\ \forall\ i \geqslant 0. \label{breve_theta_lower_bound}
\end{align}

\end{lem}We first focus on proving Lemma \ref{bound_on_norm_of_v}.

\begin{proof}[Proof of Lemma \ref{bound_on_norm_of_v}]
    
     Suppose $f$ is differentiable at $x$. By \eqref{piprod_upperbound} and \eqref{pi_twotheta_upperbound}, we have that  $\max(|\psi_{d,x}|,|\zeta_{i,x}|) \leqslant 1$ for any $i = 0,\dots,d$. Hence we have the bound \begin{align}
    \|h_x\| %& = \left\|\frac{1}{2^d}\left[\psi_{d}\zeta_{0} x_* + x - \|x_*\|2\al_{d}\hat{x} - \|x_*\|\psi_{d}\sum_{i=0}^{d-1}\al_{i}\zeta_{i+1}\hat{x}\right]\right\| \nonumber \\
    & \leqslant \frac{1}{2^d}\left(|\psi_{d,x}||\zeta_{0,x}|\|x_*\| + \|x\| + \frac{2|\sin \overline{\theta}_{d,x}|}{\pi}\|x_*\| + |\psi_{d,x}| \sum_{i=0}^{d-1}\frac{|\sin \overline{\theta}_{i,x}|}{\pi}|\zeta_{i+1,x}|\|x_*\|\right) \nonumber \\
    & \leqslant \frac{4 + d/\pi}{2^d}\max(\|x\|,\|x_*\|). \label{bound_on_h_norm}
\end{align}  Combining equation \eqref{bound_on_h_norm} and Lemma \ref{conc_v_to_h}, we attain
\begin{align}
\|v_x\| & \leqslant \|h_x \| + \|h_x - v_x\| \nonumber\\
& \leqslant \frac{4+d/\pi}{2^d}\max(\|x\|,\|x_*\|) + \frac{Kd^3\sqrt{\epsilon}}{2^d}\max(\|x\|,\|x_*\|)
+ \frac{2}{2^{d/2}}\|\eta\| \nonumber\\
& \leqslant \frac{Cd}{2^d}\max(\|x\|,\|x_*\|) + \frac{2}{2^{d/2}}\|\eta\| \label{upper_bound_on_v_norm}
\end{align}
 where in the last inequality we used $Kd^3\sqrt{\epsilon} \leqslant 1$ and set $C = 5 + 1/\pi$.
 
 When $f$ is not differentiable at $x$, we have that by equation \eqref{subdifferential_definition}, we can write $v_x = \sum_{\ell=1}^{s}c_{\ell}v_{\ell}$ where each $c_{\ell} \geqslant 0$ and $\sum_{\ell=1}^s c_{\ell} = 1$. Applying \eqref{upper_bound_on_v_norm} for differentiable points, we have that \begin{align*}
     \|v_{x}\| & \leqslant \sum_{\ell=1}^s c_{\ell}\|v_{\ell}\| \leqslant \sum_{\ell=1}^s c_{\ell} \left( \frac{Cd}{2^d}\max(\|x\|,\|x_*\|) + \frac{2}{2^{d/2}}\|\eta\|\right)  = \frac{Cd}{2^d}\max(\|x\|,\|x_*\|) + \frac{2}{2^{d/2}}\|\eta\|.
 \end{align*}
    
\end{proof}

We now show that $h_x$ is Lipschitz for $x$ outside of a ball of the origin. 

\begin{proof}[Proof of Lemma \ref{h_lipschitz_lemma}] Throughout the proof, we will use the following result from Lemma 5.1 in \cite{Huangetal2018}: \begin{align}
    |\overline{\theta}_{0,x} - \overline{\theta}_{0,y}| \leqslant 4 \max\left(\frac{1}{\|x\|},\frac{1}{\|y\|}\right)\|x - y \| \label{theta_0xy_bound}.
\end{align}
For any $x,y \neq 0$, we have that \begin{align*}
\|h_x - h_y\| & \leqslant \underbrace{\frac{1}{2^d}|\psi_{d,x}\zeta_{0,x} - \psi_{d,y}\zeta_{0,y}|\|x_*\|}_{(I)} + \frac{1}{2^d}\|x-y\| + \underbrace{\frac{\|x_*\|}{2^d}\left\|\frac{2\sin \overline{\theta}_{d,x}}{\pi}\hat{x} -  \frac{2\sin \overline{\theta}_{d,y}}{\pi}\hat{y}\right\|}_{(II)} \\
& + \underbrace{\frac{\|x_*\|}{2^d}\left\|\psi_{d,x}\sum_{i=0}^{d-1}\frac{\sin \overline{\theta}_{i,x}}{\pi}\zeta_{i+1,x}\hat{x} - \psi_{d,y}\sum_{i=0}^{d-1}\frac{\sin \overline{\theta}_{i,y}}{\pi}\zeta_{i+1,y}\hat{y}\right\|}_{(III)}.
\end{align*} We will focus on bounding each of the individual quantities.

\paragraph{(I):} The triangle inequality gives $
|\psi_{d,x}\zeta_{d,x} - \psi_{d,y}\zeta_{d,y}| \leqslant |\psi_{d,x}||\zeta_{d,x} - \zeta_{d,y}| + |\zeta_{d,y}||\psi_{d,x} - \psi_{d,y}|.$ By \eqref{piprod_upperbound} and \eqref{pi_twotheta_upperbound}, we have $\max\{|\psi_{d,x}|,|\zeta_{d,x}|\} \leqslant 1$ for all $x \neq 0$. In addition, \begin{align}
|\psi_{d,x} - \psi_{d,y}| = \frac{2}{\pi}|\overline{\theta}_{d,x} - \overline{\theta}_{d,y}|. \label{psi_dxy_bound}
\end{align} Since $g'(\theta) \in [0,1]$ for all $\theta \in [0,\pi]$ and $\overline{\theta}_{i,x} = g(\overline{\theta}_{i-1,x})$, we have that 
$|\overline{\theta}_{i,x} - \overline{\theta}_{i,y}|  \leqslant |\overline{\theta}_{i-1,x} - \overline{\theta}_{i-1,y}|.$
Repeatedly applying this inequality for each $i \in [d]$, we attain
\begin{align}
|\overline{\theta}_{d,x} - \overline{\theta}_{d,y}| & \leqslant |\overline{\theta}_{0,x} - \overline{\theta}_{0,y}| \leqslant 4 \max\left(\frac{1}{\|x\|},\frac{1}{\|y\|}\right)\|x - y \| \label{theta_dxy_bound}
\end{align} where we used \eqref{theta_0xy_bound} in the last inequality. Hence combining \eqref{psi_dxy_bound} and \eqref{theta_dxy_bound}, we get \begin{align*}
|\psi_{d,x} - \psi_{d,y}| = \frac{2}{\pi}|\overline{\theta}_{d,x} - \overline{\theta}_{d,y}| \leqslant \frac{8}{\pi} \max\left(\frac{1}{\|x\|},\frac{1}{\|y\|}\right)\|x - y \|.
\end{align*} Using the definition of $\zeta_{0,x}$, another application of \eqref{theta_0xy_bound} gives \begin{align*}
|\zeta_{0,x} - \zeta_{0,y}| \leqslant \frac{d}{\pi}|\overline{\theta}_{0,x} - \overline{\theta}_{0,y}| \leqslant \frac{4d}{\pi} \max\left(\frac{1}{\|x\|},\frac{1}{\|y\|}\right)\|x - y \|. 
\end{align*} Combining our results, if $K_1 := \frac{8 + 4d}{\pi}$ then
\begin{align}
(I) \leqslant \frac{\|x_*\|}{2^d}K_1\max\left(\frac{1}{\|x\|},\frac{1}{\|y\|}\right)\|x - y\| \label{I_bound}
\end{align} 

\paragraph{(II):} Observe that we have \begin{align*}
    \frac{2\sin \overline{\theta}_{d,x}}{\pi}\hat{x} & = \frac{2\sin \overline{\theta}_{d,x}}{\pi}\hat{y} + O_1\left(\frac{2}{\pi}\|\hat{x}-\hat{y}\|\right) \\
    & = \left(\frac{2\sin \overline{\theta}_{d,y}}{\pi} + O_1\left(\frac{2}{\pi}|\overline{\theta}_{d,x}-\overline{\theta}_{d,y}|\right)\right)\hat{y} + O_1\left(\frac{2}{\pi}\|\hat{x}-\hat{y}\|\right) \\
    & = \frac{2\sin \overline{\theta}_{d,y}}{\pi}\hat{y} + O_1\left(\frac{2}{\pi}\cdot4 + \frac{4}{\pi}\right)\max\left(\frac{1}{\|x\|},\frac{1}{\|y\|}\right)\|x-y\|
\end{align*}where the second line follows from $|\sin \theta_1 - \sin \theta_2| \leqslant|\theta_1-\theta_2|$ and the third from \eqref{theta_dxy_bound} and $\|\hat{x}-\hat{y}\| \leqslant 2\max\left(\frac{1}{\|x\|},\frac{1}{\|y\|}\right)\|x-y\|$. Thus if $K_2 := 12/\pi$, \begin{align}
(II) \leqslant \frac{\|x_*\|}{2^d}K_2\max\left(\frac{1}{\|x\|},\frac{1}{\|y\|}\right)\|x-y\| \label{II_bound}
\end{align}

\paragraph{(III):} The final term follows from \begin{align*}
    \psi_{d,x}\sum_{i=0}^{d-1}\frac{\sin \overline{\theta}_{i,x}}{\pi}\zeta_{i+1,x}\hat{x} & = \psi_{d,y}\sum_{i=0}^{d-1}\frac{\sin \overline{\theta}_{i,x}}{\pi}\zeta_{i+1,x}\hat{x} + O_1\left(\frac{8d}{\pi^2}\max\left(\frac{1}{\|x\|},\frac{1}{\|y\|}\right)\|x-y\|\right) \\
    & = \psi_{d,y}\sum_{i=0}^{d-1}\frac{\sin \overline{\theta}_{i,y}}{\pi}\zeta_{i+1,x}\hat{x} + O_1\left(\frac{4d}{\pi}+\frac{8d}{\pi^2}\right)\max\left(\frac{1}{\|x\|},\frac{1}{\|y\|}\right)\|x-y\| \\
     & = \psi_{d,y}\sum_{i=0}^{d-1}\frac{\sin \overline{\theta}_{i,y}}{\pi}\zeta_{i+1,y}\hat{x}  + \frac{1}{\pi}\sum_{i=0}^{d-1}O_1\left(\frac{d-i-1}{\pi}|\overline{\theta}_{i-1,x}-\overline{\theta}_{i-1,y}|\right)\\
     & + O_1\left(\frac{4d}{\pi}+\frac{8d}{\pi^2}\right)\max\left(\frac{1}{\|x\|},\frac{1}{\|y\|}\right)\|x-y\| \\
    & = \psi_{d,y}\sum_{i=0}^{d-1}\frac{\sin \overline{\theta}_{i,y}}{\pi}\zeta_{i+1,y}\hat{x} + O_1\left(\frac{2d^2}{\pi^2}+\frac{4d}{\pi}+\frac{8d}{\pi^2}\right)\max\left(\frac{1}{\|x\|},\frac{1}{\|y\|}\right)\|x-y\| \\
    & =\psi_{d,y}\sum_{i=0}^{d-1}\frac{\sin \overline{\theta}_{i,y}}{\pi}\zeta_{i+1,y}\hat{y} + O_1\left(\frac{2d}{\pi}+\frac{2d^2}{\pi^2}+\frac{4d}{\pi}+\frac{8d}{\pi^2}\right)\max\left(\frac{1}{\|x\|},\frac{1}{\|y\|}\right)\|x-y\| 
\end{align*} where the first line follows from equations \eqref{psi_dxy_bound} and \eqref{theta_dxy_bound} and using $|\sin \overline{\theta}_{i,x} \zeta_{i+1,x}| \leqslant 1$ for any $i = 0,\dots,d-1$ and $x$; the second line from \eqref{theta_dxy_bound}; the third from $|\zeta_{i+1,x}-\zeta_{i+1,y}|\leqslant \frac{d-i-1}{\pi}|\overline{\theta}_{i-1,x}-\overline{\theta}_{i-1,y}|$; the fourth from $|\overline{\theta}_{i-1,x}-\overline{\theta}_{i-1,y}|\leqslant |\overline{\theta}_{0,x}-\overline{\theta}_{0,y}|$, \eqref{theta_0xy_bound}, and $\sum_{i=0}^{d-1}(d-i-1)=\frac{1}{2}(d-1)d$; and the fifth from $|\sin \overline{\theta}_{i,y}\zeta_{i+1,y}|\leqslant 1$ for all $i = 0,\dots,d-1$ and $\|\hat{x}-\hat{y}\|\leqslant 2\max\left(\frac{1}{\|x\|},\frac{1}{\|y\|}\right)\|x-y\|$.% We can bound the following quantity by isolating each term in the following way: \begin{align*}
  Combining our results, we have that if $
K_3 : = \frac{8d+2d^2}{\pi^2} + \frac{6d}{\pi}$ then\begin{align}
(III) \leqslant \frac{\|x_*\|}{2^d}K_3 \max\left(\frac{1}{\|x\|},\frac{1}{\|y\|}\right)\|x-y\|. \label{III_bound}
\end{align}

Thus for all $x,y\neq 0$, using equations \eqref{I_bound}, \eqref{II_bound}, and \eqref{III_bound}, we conclude that  \begin{align*}
\|h_x - h_y\|  &\leqslant (I) + \frac{1}{2^d}\|x-y\| + (II) + (III) \\
&\leqslant \frac{1}{2^d}\left(\|x_*\|(K_1 + K_2 + K_3)\max\left(\frac{1}{\|x\|},\frac{1}{\|y\|}\right) + 1\right)\|x-y\| \\
& = \left(\frac{(2d^2 + (10\pi+8)d+20\pi)\|x_*\|}{\pi^22^d}\max\left(\frac{1}{\|x\|},\frac{1}{\|y\|}\right) + \frac{1}{2^d}\right)\|x-y\|.
%& =: L_{x,y}\|x-y\|
\end{align*} %where \begin{align*}
%L_{x,y} 
%& = \frac{(2d^2 + (10\pi+8)d+20\pi)\|x_*\|}{\pi^2 2^d}\max\left(\frac{1}{\|x\|},\frac{1}{\|y\|}\right) + \frac{1}{2^d}.
%\end{align*} 
Then if $x,y \notin \mathcal{B}(0,r)$, we can further conclude that \begin{align*}
\|h_{x} - h_y\| & \leqslant \left(\frac{(2d^2 + (10\pi+8)d+20\pi)\|x_*\|}{r\pi^2 2^d} + \frac{1}{2^d}\right)\|x-y\|.
\end{align*}
\end{proof}

We can now show that $h_x$ approximates any $v_x \in \partial f(x)$, which is formalized in Lemma \ref{conc_v_to_h}. Prior to this proof, we define \begin{align}
    w_x := \Lambda_x^{\top}(\Lambda_x x - \Phi_{x_d,x_{*,d}}\Lambda_{x_*}x_*). \label{wx_def}
\end{align} The key idea is that the RRCP and WDC together imply $v_x \approx w_x$ and the WDC further implies $w_x \approx h_x$ which is shown in Lemma \ref{conc_wx_to_hx}.

\begin{proof}[Proof of Lemma \ref{conc_v_to_h}] Suppose $f$ is differentiable at $x$ so that $v_x = \overline{v}_x - q_x$ where
    $\overline{v}_x = \Lambda_x^\top A_{x_d}^\top (A_{x_d}\Lambda_x x - A_{x_{*,d}}\Lambda_{x_*}x_*)$ and $q_x = \Lambda_x^\top A_{x_d}^\top \eta$.  Observe that \begin{align*}
\left\|\overline{v}_x - w_x\right\| & \leqslant \left\|\Lambda_x^\top(A_{x_d}^\top A_{x_d} - I_{n_d})\Lambda_x\right\| \|x\| + \left\|\Lambda_x^\top(A_{x_d}^\top A_{x_{*,d}} - \Phi_{x_d, x_{*,d}})\Lambda_{x_*}\right\|\|x_*\|.
\end{align*} By the local linearity of $\G$, for sufficiently small $z \in \R^k$, we have $\G(x+z)-\G(x) = \Lambda_x z$. Hence by the RRCP, we have for sufficiently small $z,\tilde{z} \in \R^k$, \begin{align*}
    |\langle (A_{x_d}^{\top}A_{x_d}-I_{n_d})\Lambda_x z,\Lambda_xz\rangle| \leqslant L\epsilon \|\Lambda_x\|^2\|z\|^2
\end{align*} and \begin{align*}
    |\langle (A_{x_d}^{\top}A_{x_{*,d}}-\Phi_{x_d,x_{*,d}})\Lambda_x z,\Lambda_{x_*}\tilde{z}\rangle| \leqslant L\epsilon \|\Lambda_x\|\|\Lambda_{x_*}\|\|z\|\|\tilde{z}\|.
\end{align*} Since this holds for any $z,\tilde{z} \in \R^k$, we conclude that \begin{align*}
    \left\|\Lambda_x^\top(A_{x_d}^\top A_{x_d} - I_{n_d})\Lambda_x\right\| \leqslant L\epsilon\|\Lambda_x\|^2\ \text{and}\ \left\|\Lambda_x^\top(A_{x_d}^\top A_{x_{*,d}} - \Phi_{x_d, x_{*,d}})\Lambda_{x_*}\right\|\leqslant L\epsilon\|\Lambda_x\|\|\Lambda_{x_*}\|.
\end{align*} This implies \begin{align*}
\left\|\overline{v}_x - w_x\right\| & \leqslant L\epsilon\left(\|\Lambda_x\|^2 + \|\Lambda_x\|\|\Lambda_{x_*}\|\right) \max(\|x\|,\|x_*\|) \\
& \leqslant 2L \epsilon\left(\frac{1}{2} + \epsilon\right)^d \max(\|x\|,\|x_*\|) \end{align*} where the last inequality follows by the WDC. Furthermore, by Lemma \ref{conc_wx_to_hx}, we have that \begin{align*}
    \|w_x - h_x\| \leqslant \frac{90d^3\sqrt{\epsilon}}{2^d}\max(\|x\|,\|x_*\|).
\end{align*} Combining these two bounds, we have \begin{align}
\|\overline{v}_{x} - h_{x}\| & \leqslant \|\overline{v}_x - w_x\| + \|w_x - h_{x}\| \nonumber\\
& \leqslant \sqrt{\epsilon}\left(2L \frac{(1 + 2\epsilon)^d}{2^d} + 90\frac{d^3}{2^d}\right)\max(\|x\|,\|x_*\|) \nonumber\\
& \leqslant \sqrt{\epsilon}K \frac{d^3}{2^d}\max(\|x\|,\|x_*\|) \label{bound_on_vbar_to_h}
\end{align} for some universal constant $K$ where the third inequality follows since $2\epsilon d \leqslant 1 \Longrightarrow (1 + 2\epsilon)^d \leqslant e^{2\epsilon d} \leqslant 1 + 4\epsilon d$ so choosing $\epsilon < 1/(4d)$ implies $(1+2\epsilon)^d\leqslant2$. Lastly, to bound $\|q_x\|$, observe that \begin{align}
    \|q_x\| \leqslant \|A_{x_d}\Lambda_x\|\|\eta\| \leqslant \sqrt{1 + L\epsilon}\|\Lambda_x\|\|\eta\| \leqslant \sqrt{\frac{13}{12}(1+ L\epsilon)}\frac{1}{2^{d/2}} \|\eta\| \leqslant \frac{2}{2^{d/2}}\|\eta\|. \label{bound_on_q_norm}
\end{align} where in the second inequality we used \eqref{bound_on_AxdLambdax} and in the third inequality we used \eqref{bound_on_lambda_energy}. The last inequality follows by choosing $\epsilon$ such that $\sqrt{\frac{13}{12}(1+L\epsilon)}\leqslant2$. Then we can combine \eqref{bound_on_vbar_to_h} and \eqref{bound_on_q_norm} to obtain \begin{align*}
    \|v_x - h_x\| \leqslant \|\overline{v}_x - h_x\| + \|q_x\| \leqslant K \frac{d^3\sqrt{\epsilon}}{2^d}\max(\|x\|,\|x_*\|) + \frac{2}{2^{d/2}}\|\eta\|.
\end{align*} 

When $f$ is not differentiable at $x$, we can use \eqref{subdifferential_definition} to write $v_x = \sum_{\ell=1}^s c_{\ell}v_{\ell}$ where $c_{\ell} \geqslant 0$ for $\ell \in [s]$. Moreover, note that for each $v_{\ell}$, there exists a direction $w_{\ell}$ such that $v_{\ell} = \lim_{\delta_{\ell} \rightarrow 0^+} \nabla f(x+\delta_{\ell}w_{\ell})$ and $f$ is differentiable at $x+\delta_{\ell}w_{\ell}$ for sufficiently small $\delta_{\ell} > 0$. Appealing to the continuity of $h_x$ for $x \neq 0$, we obtain \begin{align*}
    \|v_x - h_x\| \leqslant \sum_{\ell=1}^s c_{\ell}\|v_{\ell} - h_x \|  & = \sum_{\ell=1}^s c_{\ell} \Big\|\lim_{\delta_{\ell} \rightarrow 0^+} \nabla f(x+\delta_{\ell}w_{\ell}) - h_x\Big\| \\
    & = \sum_{\ell=1}^s c_{\ell} \lim_{\delta_{\ell} \rightarrow 0^+}\| \nabla f(x+\delta_{\ell}w_{\ell}) - h_{x+\delta_{\ell}w_{\ell}}\| \\
    & = \sum_{\ell=1}^s c_{\ell} \lim_{\delta_{\ell} \rightarrow 0^+}\| v_{x+\delta_{\ell}w_{\ell}} - h_{x+\delta_{\ell}w_{\ell}}\| \\
    & \leqslant \sum_{\ell=1}^s c_{\ell} \left(K \frac{d^3\sqrt{\epsilon}}{2^d}\max(\|x\|,\|x_*\|) + \frac{2}{2^{d/2}}\|\eta\|\right)\\
    & = K \frac{d^3\sqrt{\epsilon}}{2^d}\max(\|x\|,\|x_*\|) + \frac{2}{2^{d/2}}\|\eta\|.
\end{align*} 
\end{proof}

We now establish a technical result that shows $w_x$ is approximated by $h_x$. Prior to this proof, we highlight the following result that summarizes some useful bounds from \cite{HV2017}:

\begin{lem}[Results from Lemma 5 in \cite{HV2017}] \label{boundsforvtoh}
Fix $0 < \epsilon < d^{-4}(1/16\pi)^2$ and let $d \geqslant 2$. Let $W_i$ satisfy the WDC with constant $\epsilon$ for $i = 1,\dots d$. Then for any non-zero $x,y \in \R^k$, the following hold: \begin{align}
\left\|\Lambda_x^\top\Lambda_yy - \tilde{h}_{x,y}\right\| & \leqslant 24 \frac{d^3\sqrt{\epsilon}}{2^d}\|y\|, \label{conc_htilde}\\
\left\langle \Lambda_x x, \Lambda_y y \right\rangle & \geqslant \frac{1}{4\pi}\frac{1}{2^d}\|x\|\|y\|, \label{conc_inn_prod}\\
\left| \frac{\|y_d\|}{\|x_d\|} - \frac{\|y\|}{\|x\|} \right| & \leqslant 8d\epsilon \frac{\|y\|}{\|x\|}, \label{conc_norm}\\
|\theta_d - \overline{\theta}_d| & \leqslant 4d\sqrt{\epsilon} \label{conc_angle}
\end{align} where  $\theta_d : = \angle(x_d,y_d)$, $\overline{\theta}_d : = g^{\circ d}(\angle(x,y))$, and the vector $\tilde{h}_{x,y}$ is defined as \begin{align}
\tilde{h}_{x,y} : = \frac{1}{2^d}\left[ \left(\prod_{i=0}^{d-1} \frac{\pi - \overline{\theta}_i}{\pi}\right)y + \sum_{i=0}^{d-1}\frac{\sin\overline{\theta}_i}{\pi} \left(\prod_{j=i+1}^{d-1} \frac{\pi - \overline{\theta}_j}{\pi}\right)\frac{\|y\|}{\|x\|}x\right] \label{htilde_def}
\end{align} with $\overline{\theta}_0 := \angle(x,y)$ and $\overline{\theta}_i := g(\overline{\theta}_{i-1})$ for $i \in [d]$.
\end{lem}

We now establish that $w_x$ is approximated by $h_x$.

\begin{lem} \label{conc_wx_to_hx}
Fix $0 < \epsilon < d^{-4}(1/16\pi)^2$. Let $W_i$ satisfy the WDC with constant $\epsilon$ for $i = 1,\dots d$. For any non-zero $x \in \R^k$, we have \begin{align*}
\|w_{x} - h_{x}\| \leqslant \frac{90d^3}{2^d}\sqrt{\epsilon}\max(\|x\|,\|x_*\|)
\end{align*} where $w_x$ is defined by \eqref{wx_def}.
\end{lem}
\begin{proof} Fix $x \in \R^k\setminus\{0\}$ and set $\theta_d := \angle(x_d,x_{*,d})$. Note that by the definition of $\Phi_{z,w}$ and $M_{\hat{z}\leftrightarrow \hat{w}}$, $w_x$ can be written as \begin{align*}
    w_x
    = \Lambda_x^{\top}\Lambda_x x - \frac{\pi-2\theta_d}{\pi}\Lambda_x^{\top}\Lambda_{x_*}x_* - \frac{2\sin\theta_d}{\pi}\frac{\|\Lambda_{x_*} x_*\|}{\|\Lambda_x x\|}\Lambda_x^{\top}\Lambda_x x
\end{align*} where $\theta_d := \angle (x_d,x_{*,d}).$ Observe that \begin{align*}
\|w_x- h_{x}\| & \leqslant \left\|\Lambda_x^\top\Lambda_x x - \frac{1}{2^d}x\right\| + \left\|\frac{\pi - 2\theta_d}{\pi}\Lambda_x^\top\Lambda_{x_*}x_* - \frac{\pi - 2\overline{\theta}_d}{\pi}\tilde{h}_{x,x_*}\right\|\ \\
& + \left\|\frac{2\sin \theta_d}{\pi}\frac{\|\Lambda_{x_*}x_*\|}{\|\Lambda_x x\|}\Lambda_x^\top\Lambda_x x - \frac{2\sin \overline{\theta}_d}{\pi} \frac{\|x_*\|}{\|x\|}\frac{1}{2^d}x\right\|
\end{align*} where $\overline{\theta}_d:=g^{\circ d}(\angle(x,x_*))$ and $\tilde{h}_{x,x_*}$ is defined in \eqref{htilde_def}.

We focus on bounding each individual quantity separately. For the first term, we have that by (\ref{conc_htilde}) in Lemma \ref{boundsforvtoh}, \begin{align}
\Lambda_x^\top\Lambda_x x  = \frac{1}{2^d}x + O_1 \left(  \frac{24d^3 }{2^d}\right)\sqrt{\epsilon}\max(\|x\|,\|x_*\|). \label{wx_lemma_t1}
\end{align} For the second term, observe that \begin{align}
    \frac{\pi - 2\theta_d}{\pi}\Lambda_x^\top\Lambda_{x_*}x_* & = \frac{\pi - 2\theta_d}{\pi}\tilde{h}_{x,x_*} + O_1\left(  \frac{24d^3 \sqrt{\epsilon}}{2^d}\|x\|\right) \nonumber \\
    & = \left(\frac{\pi - 2\overline{\theta}_d}{\pi} + O_1\left(\frac{2}{\pi}\cdot 4d\sqrt{\epsilon}\right)\right)\tilde{h}_{x,x_*} + O_1\left(  \frac{24d^3 \sqrt{\epsilon}}{2^d}\|x\|\right) \nonumber \\
    & = \frac{\pi - 2\overline{\theta}_d}{\pi}\tilde{h}_{x,x_*} + O_1\left(\frac{8d}{\pi}\cdot\frac{(1 + d/\pi)}{2^d} + \frac{24d^3}{2^d}\right)\sqrt{\epsilon}\max(\|x\|,\|x_*\|) \label{wx_lemma_t2}
\end{align}where in the first equality we used \eqref{conc_htilde} and in the second we used \eqref{conc_angle} and the fact that $\|\tilde{h}_{x,x_*}\|\leqslant 2^{-d}(1+d/\pi)\|x_*\|.$ For the final term, observe that \begin{align}
    \frac{2\sin \theta_d}{\pi}\frac{\|\Lambda_{x_*}x_*\|}{\|\Lambda_x x\|}\Lambda_x^\top\Lambda_x x & = \frac{2\sin \theta_d}{\pi} \frac{\|\Lambda_{x_*}x_*\|}{\|\Lambda_x x\|}\left(\frac{1}{2^d}x + O_1\left(\frac{24d^3\sqrt{\epsilon}}{2^d}\|x\|\right)\right) \nonumber\\
    & = \frac{2\sin \theta_d}{\pi}\frac{\|\Lambda_{x_*}x_*\|}{\|\Lambda_x x\|}\frac{1}{2^d}x + O_1\left(\frac{24d^3\sqrt{\epsilon}}{2^d}\frac{2}{\pi}\frac{\|\Lambda_{x_*}x_*\|}{\|\Lambda_x x\|}\|x\|\right) \nonumber\\
    & = \frac{2\sin \theta_d}{\pi}\left(\frac{\|x_*\|}{\|x\|} + O_1\left(8d\epsilon\frac{\|x_*\|}{\|x\|}\right)\right)\frac{1}{2^d}x + O_1\left(\frac{4\cdot 24d^3\sqrt{\epsilon}}{\pi2^d}\|x_*\|\right)\nonumber \\
    & = \frac{2\sin \theta_d}{\pi}\frac{\|x_*\|}{\|x\|}\frac{1}{2^d}x + O_1\left(\frac{16d\epsilon}{\pi 2^d}\|x\| + \frac{4\cdot 24 d^3\sqrt{\epsilon}}{\pi 2^d}\|x_*\|\right) \nonumber\\
    & = \left(\frac{2\sin \overline{\theta}_d}{\pi} + O_1\left(\frac{8d\sqrt{\epsilon}}{\pi}\right)\right)\frac{\|x_*\|}{\|x\|}\frac{1}{2^d}x + O_1\left(\frac{16d\epsilon}{\pi 2^d}\|x\| + \frac{4\cdot 24 d^3\sqrt{\epsilon}}{\pi 2^d}\|x_*\|\right) \nonumber\\
    & = \frac{2\sin \overline{\theta}_d}{\pi}\frac{\|x_*\|}{\|x\|}\frac{1}{2^d}x + O_1\left(\frac{1}{2^d}\left(\frac{24d + 4\cdot 24d^3}{\pi}\right)\right)\sqrt{\epsilon}\max(\|x\|,\|x_*\|). \label{wx_lemma_t3}
\end{align} where the first line follows from \eqref{conc_htilde}; the second line from $|\sin \theta| \leqslant 1$; the third line from \eqref{conc_norm}; and the fifth line from \eqref{conc_angle}. Combining equations \eqref{wx_lemma_t1}, \eqref{wx_lemma_t2}, and \eqref{wx_lemma_t3} achieves the desired result.%note that it is upper bounded by \begin{align*}
\end{proof}

\subsubsection{Proofs for Section \ref{conv_proof_second_subsection}} \label{proofs_for_second_subsection}

We first establish Lemma \ref{Seps_lemma} which shows that the zeros of $h_x$ occur near $x_*$ and a particular negative multiple $-\rho_d x_*$. Here the lemma is stated more precisely. 

\begin{prop} \label{Seps} Suppose $\beta > 0$ obeys $24\pi d^6\sqrt{\beta}\leqslant 1$ and define $\Seps_{\beta}$ as in \eqref{Sbeta_def}.  If $x \in \Seps_{\beta}$, then either \begin{align*}
|\overline{\theta}_{0,x}| \leqslant 82 \pi d^4\beta\ \text{and}\ |\|x\| - \|x_*\|| \leqslant 838\pi d^5\beta\|x_*\|
\end{align*} or \begin{align*}
|\overline{\theta}_{0,x} - \pi|\leqslant 24\pi^2d^4\sqrt{\beta}\ \text{and}\ |\|x\| - \rho_d\|x_*\|| \leqslant 3517d^8\sqrt{\beta}\|x_*\|.
\end{align*} In particular, we have 
\begin{align*}
\Seps_{\beta} \subset \mathcal{B}(x_*,70000\pi^2 d^9\beta\|x_*\|) \cup \mathcal{B}(-\rho_d x_*, 77422\pi^2d^{12}\sqrt{\beta}\|x_*\|).
\end{align*} Additionally, $\rho_d \rightarrow 1$ as $d \rightarrow \infty$.
\end{prop}

\begin{proof}[Proof of Proposition \ref{Seps}]
Without loss of generality, let $x_* = e_1$ and $\|x_*\| = 1$ where $e_1$ is the first standard basis vector in $\R^{k}$. We also let $x = r\cos \overline{\theta}_0 e_1 + r\sin \overline{\theta}_0 e_2$ where $\overline{\theta}_0 = \angle(x,x_*)$. For simplicity, we use the shorthand notation $\overline{\theta}_{i} = \overline{\theta}_{i,x}$ for $i \in [d]$. Set \begin{align*}
\xi = \left(\frac{\pi - 2\overline{\theta}_d}{\pi}\right)\left(\prod_{i=0}^{d-1}\frac{\pi - \overline{\theta}_i}{\pi}\right)\ \text{and}\ \al = \frac{2\sin \overline{\theta}_d}{\pi} + \left(\frac{\pi - 2\overline{\theta}_d}{\pi}\right) \sum_{i=0}^{d-1} \frac{\sin \overline{\theta}_i}{\pi}\left(\prod_{j=i+1}^{d-1} \frac{\pi - \overline{\theta}_j}{\pi}\right).
\end{align*} Note that we can write \begin{align*}
h_{x} = \frac{1}{2^d}\left(-\xi \hat{x}_* + (r - \al)\hat{x}\right)
\end{align*} Then if $x \in \Seps_{\beta}$, we have that \begin{align}
|-\xi + \cos \overline{\theta}_0 (r - \al)| & \leqslant \beta M \label{cosleqM} \\
|\sin \overline{\theta}_0(r - \al)| & \leqslant \beta M \label{sinleqM}
\end{align} where $M:=\max(r,1)$.

To prove the Proposition, we first show that it is sufficient to only consider the small and large angle case. Then, we show that in the small and large angle case, $x \approx x_*$ and $x \approx -\rho_d x_*$, respectively. We begin by proving that $\max(\|x\|,\|x_*\|) \leqslant 6d$ for any $x \in \Seps_{\beta}$.

\textbf{Bound on maximal norm in $\Seps_{\beta}$:} It suffices to show that $r \leqslant 6d$. Suppose $r > 1$ since if $r \leqslant 1$, the result is immediate. Then either $|\sin \overline{\theta}_0| \geqslant 1/\sqrt{2}$ or $|\cos \overline{\theta}_0| \geqslant 1/\sqrt{2}$. If $|\sin \overline{\theta}_0| \geqslant 1/\sqrt{2}$ then \eqref{sinleqM} gives \begin{align*}
|r - \al| \leqslant \sqrt{2}\beta r \Longrightarrow (1-\sqrt{2}\beta)r \leqslant |\al|.
\end{align*} But \begin{align*}
|\al| & \leqslant \frac{2}{\pi}|\sin \overline{\theta}_d| + \left|\left(\frac{\pi - 2\overline{\theta}_d}{\pi}\right) \sum_{i=0}^{d-1} \frac{\sin \overline{\theta}_i}{\pi}\left(\prod_{j=i+1}^{d-1} \frac{\pi - \overline{\theta}_i}{\pi}\right)\right| \leqslant 1 + \frac{d}{\pi}
\end{align*} where the second inequality used equations \eqref{sum_prod_upperbound} and \eqref{pi_twotheta_upperbound}. Thus \begin{align*}
r \leqslant \frac{1+ \frac{d}{\pi}}{1 - \sqrt{2}\beta} \leqslant 2\left(1 + \frac{d}{\pi}\right) \leqslant 2 + d \leqslant 2d
\end{align*} provided $\beta < 1/4$ and $d \geqslant 2$. If $|\cos \overline{\theta}_0 | \geqslant 1/\sqrt{2}$, then \eqref{cosleqM} gives \begin{align*}
|r - \al| \leqslant \sqrt{2}(\beta r + |\xi|) \Longrightarrow (1 - \sqrt{2}\beta)r \leqslant \sqrt{2}|\xi| + \al.
\end{align*} But by \eqref{piprod_upperbound},  \begin{align*}
|\xi| = \left|\left(\frac{\pi - 2\overline{\theta}_d}{\pi}\right) \left(\prod_{i=0}^{d-1} \frac{\pi - \overline{\theta}_i}{\pi}\right)\right| \leqslant 1\ \text{since}\ \overline{\theta}_i  \in [0,\pi/2]\ \forall\ i\ \geqslant 1.
\end{align*} Hence if $\beta < 1/4$, \begin{align*}
r \leqslant \frac{\sqrt{2} + 2d}{1 - \sqrt{2}\beta} \leqslant 2\sqrt{2} + 4d \leqslant \sqrt{2}d + 4d \leqslant 6d.
\end{align*} Thus in any case, $r \leqslant 6d \Longrightarrow M \leqslant 6d$. 

We now show that it is sufficient to only consider the small angle case $\overline{\theta}_0 \approx 0$ and the large angle case $\overline{\theta}_0 \approx \pi$.

\textbf{Sufficiency:} We have three possible cases:
\begin{itemize}
\item $\sin \overline{\theta}_0 \leqslant 48\pi d^4\beta$: Then we have that $\overline{\theta}_0 = O_1(82\pi d^4\beta)$ or $\overline{\theta}_0 = \pi + O_1(82\pi d^4\beta)$.
\item $\sin \overline{\theta}_0 > 48 \pi d^4 \beta$ and $|r - \al| \geqslant \sqrt{\beta}M$: Observe that due to equation \eqref{sinleqM}, we have that $|r-\al| \leqslant \frac{\beta M}{\sin \overline{\theta}_0}.$ Thus using this inequality in equation \eqref{cosleqM}, we have that \begin{align}
    |\xi| \leqslant \beta M + \frac{\beta M}{\sin \overline{\theta}_0} \leqslant \frac{2\beta M}{\sin \overline{\theta}_0} \leqslant \frac{2\beta M}{48\pi d^4\beta} \leqslant \frac{12d}{48\pi d^4} = \frac{1}{4\pi}d^{-3} \label{upperbound_on_xi_forcaseII}
\end{align} where we used the assumption $\sin \overline{\theta}_0 > 48\pi d^4\beta$ in the second to last inequality and $M \leqslant 6d$ in the last inequality. In addition, \eqref{cos_inv_bound} implies \begin{align}
|\pi - 2\overline{\theta}_d| \geqslant \left|\pi - 2\cos^{-1}\left(\frac{1}{\pi}\right)\right|  \geqslant \frac{1}{2}. \label{pi_2thetad_geq_half}
\end{align} Combining this inequality with \eqref{piprod_lowerbound} and \eqref{upperbound_on_xi_forcaseII}, we obtain \begin{align*}
\frac{1}{2\pi}\left(\frac{\pi - \overline{\theta}_0}{\pi}\right)d^{-3}\leqslant |\xi|\leqslant \frac{1}{4\pi}d^{-3}.
\end{align*} From this, we can conclude that $\overline{\theta}_0 \geqslant \frac{\pi}{2}.$ Moreover, since $|r - \al| \geqslant \sqrt{\beta}M$, then \eqref{sinleqM} implies that $|\sin \overline{\theta}_0| \leqslant \sqrt{\beta}$ so we must have that $\overline{\theta}_0 = \pi + O_1(2\sqrt{\beta})$ since $\overline{\theta}_0 \geqslant \frac{\pi}{2}$ and $\beta < 1$.
\item $|r - \al| \leqslant \sqrt{\beta}M:$ Then \eqref{cosleqM} implies \begin{align*}
|\xi| \leqslant 2 \sqrt{\beta}M.
\end{align*} But note that by \eqref{piprod_lowerbound}, \begin{align*}
\xi = \left(\frac{\pi - 2\overline{\theta}_d}{\pi}\right)\left(\prod_{i=0}^{d-1}\frac{\pi - \overline{\theta}_i}{\pi}\right) \geqslant \frac{(\pi - 2 \overline{\theta}_d)(\pi - \overline{\theta}_0)}{d^3\pi^2}.
\end{align*} In addition, since
$|\pi - 2\overline{\theta}_d| \geqslant \frac{1}{2}$ by \eqref{pi_2thetad_geq_half}, we have  \begin{align*}
|\xi| \geqslant \frac{|(\pi - 2\overline{\theta}_d)(\pi - \overline{\theta}_0)|}{d^3 \pi^2} \geqslant \frac{|\pi - \overline{\theta}_0|}{2d^3\pi^2}
\end{align*} which implies \begin{align*}
|\pi - \overline{\theta}_0| \leqslant 4d^3\pi^2\sqrt{\beta}M \leqslant 24d^4\pi^2 \sqrt{\beta}.
\end{align*} Thus $\overline{\theta}_0 = \pi + O_1(24d^4\pi^2 \sqrt{\beta})$.
\end{itemize}

Since only one of these situations can hold, it suffices to consider either the small angle case $\overline{\theta}_0 = O_1(82\pi d^4 \beta)$ or the large angle case $\overline{\theta}_0 = \pi + O_1(24d^4\pi^2\sqrt{\beta})$. Now, we show that in the small angle case, $x \approx x_*$, while in the large angle case, $x \approx -\rho_d x_*$.

\textbf{Small Angle Case:} Assume $\overline{\theta}_0 = O_1(\delta)$ where we set $\delta := 82\pi d^4\beta$. Note that since $\overline{\theta}_i \leqslant \overline{\theta}_0 \leqslant \delta$ for each $i$, we have that \begin{align*}
\prod_{i=0}^{d-1} \frac{\pi - \overline{\theta}_i}{\pi} 
\geqslant \left(1 - \frac{\delta}{\pi} \right)^d = 1 + O_1 \left(\frac{2d\delta}{\pi}\right)
\end{align*} provided $d\delta/\pi \leqslant 1/2$. Hence \begin{align*}
\xi  & = \left(\frac{\pi - 2\overline{\theta}_d}{\pi}\right) \left(\prod_{i=0}^{d-1} \frac{\pi - \overline{\theta}_i}{\pi}\right) \geqslant \left(1 + O_1\left(\frac{2\delta}{\pi}\right)\right)\left(1 + O_1\left(\frac{2d\delta}{\pi}\right)\right) \\
\end{align*} where we used \eqref{pi_Oone_bound} in the second inequality. In addition, $|\sin \overline{\theta}_d| \leqslant |\overline{\theta}_d| \leqslant \delta$ and \eqref{sum_prod_upperbound} imply that \begin{align*}
\left|\sum_{i=0}^{d-1} \frac{\sin \overline{\theta}_i}{\pi}\left(\prod_{j=i+1}^{d-1} \frac{\pi - \overline{\theta}_j}{\pi}\right)\right| \leqslant \frac{d}{\pi}|\sin \overline{\theta}_d| \leqslant d\delta. 
\end{align*} Hence \begin{align*}
\al  = \frac{2\sin \overline{\theta}_d}{\pi} + \left(\frac{\pi - 2\overline{\theta}_d}{\pi}\right) \sum_{i=0}^{d-1} \frac{\sin \overline{\theta}_i}{\pi}\left(\prod_{j=i+1}^{d-1} \frac{\pi - \overline{\theta}_j}{\pi}\right)
& = O_1\left(\frac{2\delta}{3\pi}\right) + \left(1 + O_1\left(\frac{2\delta}{\pi}\right)\right)O_1(d\delta) \\
%& = O_1\left(\frac{2\delta}{3\pi}\right) + O_1(d\delta) + O_1\left(\frac{2d^2 \delta^2}{\pi}\right) \\
& = O_1 \left(\frac{(4 + 3d\pi + 6d^2)\delta}{3\pi}\right)
\end{align*} where we used $\delta < 1$ in the last equality. Thus since $|-\xi + \cos \overline{\theta}_0(r - \al)| \leqslant \beta M$ and $M \leqslant 6d$, we attain \begin{align*}
- \left(1 + O_1\left(\frac{2\delta}{\pi}\right)\right)\left(1 + O_1\left(\frac{2d\delta}{\pi}\right)\right) & + (1 + O_1(\delta))\left(r + O_1 \left(\frac{(4 + 3d\pi + 6d^2)\delta}{3\pi}\right)\right) = O_1(6d\beta).
\end{align*} Rearranging, this gives \begin{align*}
r-1 & = O_1\left(\frac{2d\delta}{\pi} + \frac{2\delta}{\pi} + \frac{16d\delta^2}{\pi} + (\delta + 1)\frac{(4 + 3d\pi + 6d^2)\delta}{3\pi}\right) + O_1(12d\beta) + O_1(6d\beta)  \\
& = O_1\left(\frac{(12d + 12 + 48d)\delta + (2\epsilon + 1)(4 + 3\pi d + 12d)\delta}{3\pi} + 18d\beta\right)\\
%& = O_1\left(\frac{(12d + 12 + 48d)\sqrt{\epsilon} + (4 + 3\pi d + 12d)\sqrt{\epsilon}}{3\pi}\right)\ \text{as}\ \epsilon < 1/2 \\
%& = O_1\left(\frac{(16 + (3\pi + 12)d + 60d)\sqrt{\epsilon}}{3\pi}\right) + O_1(2\epsilon)r + O_1(\epsilon M) \\
%& = O_1\left(\frac{98d}{3\pi}\sqrt{\epsilon}\right) + O_1(2\epsilon)r + O_1(\epsilon M) \\
%& = O_1\left(\frac{98d}{3\pi}\sqrt{\epsilon}\right) + O_1(12d\epsilon) + O_1(6d\epsilon ) \\
& = O_1(10d\delta + 18d\beta) \\
& = O_1(838\pi d^5\beta)
\end{align*} where we used $\delta < 1/2$ and $d \geqslant 2$ in the second to last equality and the definition of $\delta$ in the final equality.

\textbf{Large Angle Case:} Assume $\overline{\theta}_0 = \pi + O_1(\delta)$ where $\delta := 24d^4\pi^2\sqrt{\beta}.$  We first prove that $\al$ is close to $\rho_d$. Recall that $\overline{\theta}_d = \breve{\theta}_d + O_1(d\delta)$. Then by the mean value theorem: \begin{align*}
|\sin \overline{\theta}_d - \sin \breve{\theta}_d | \leqslant |\overline{\theta}_d - \breve{\theta}_d | \leqslant d \delta
\end{align*} so $\sin \overline{\theta}_d = \sin \breve{\theta}_d + O_1(d\delta)$. Let
$\Gamma_d : = \sum_{i=0}^{d-1} \frac{\sin \breve{\theta}_i}{\pi}\left(\prod_{j=i+1}^{d-1} \frac{\pi - \breve{\theta}_j}{\pi}\right)$ and note that $\rho_d = \frac{2\sin \breve{\theta}_d}{\pi} + \left(\frac{\pi - 2\breve{\theta}_d}{\pi}\right)\Gamma_d.$ In \cite{HV2017}, it was shown that if $d^2\delta /\pi \leqslant 1$, then $|\Gamma_d| \leqslant d$ and 
\begin{align*}
\sum_{i=0}^{d-1} \frac{\sin \overline{\theta}_i}{\pi}\left(\prod_{j=i+1}^{d-1} \frac{\pi - \overline{\theta}_j}{\pi}\right) = \Gamma_d + O_1(3d^3\delta).
\end{align*} By the condition, $d^2\delta /\pi \leqslant 1$, we require $\sqrt{\beta} \leqslant \frac{1}{24\pi d^6}.$
Thus for sufficiently small $\beta$, we have \begin{align*}
\al & = \frac{2\sin \overline{\theta}_d}{\pi} + \left(\frac{\pi - 2\overline{\theta}_d}{\pi}\right) \sum_{i=0}^{d-1} \frac{\sin \overline{\theta}_i}{\pi}\left(\prod_{j=i+1}^{d-1} \frac{\pi - \overline{\theta}_j}{\pi}\right) \\
& = \frac{2\sin \breve{\theta}_d}{\pi} + O_1\left(\frac{2d\delta}{\pi}\right) + \left(\frac{\pi - 2 \breve{\theta}_d}{\pi} + O_1 \left(\frac{2d\delta}{\pi}\right)\right)\left(\Gamma_d + O_1(3d^3 \delta)\right) \\
& = \rho_d + O_1\left(\frac{2d\delta}{\pi}\right) + \Gamma_d O_1\left(\frac{2d\delta}{\pi}\right) + \left(\frac{\pi - 2\breve{\theta}_d}{\pi}\right)O_1\left(3d^3\delta\right) + O_1 \left(\frac{6d^4\delta^2}{\pi}\right) \\
& = \rho_d + O_1\left(\frac{2d\delta}{\pi}\right) + O_1\left(\frac{2d^2\delta}{\pi}\right)+ O_1\left(3d^3\delta\right) + O_1\left(\frac{6d^4\delta^2}{\pi}\right) \\
%& = \rho_d + O_1 \left(\left(\frac{4\delta}{\pi} + 3\delta + \frac{6\delta^2}{\pi}\right)d^4\right) \\
%& = \rho_d + O_1 \left(\delta\left(\frac{4}{\pi} + 3 + \frac{6}{\pi}\right)d^4\right)\ \text{provided}\ \delta = 24d^4\pi^2\sqrt{\beta} \leqslant 1 \\
& = \rho_d + O_1(7d^4\delta).
\end{align*} We now prove $r$ is close to $\rho_d$. Since $x \in \Seps_{\beta}$, \begin{align*}
|-\beta + \cos \overline{\theta}_0(r - \al)|\leqslant \beta M.
\end{align*} Also note that $|\beta| \leqslant \delta/\pi$ by \eqref{beta_pi_upperbound}. Since $\cos \overline{\theta}_0 = 1 + O_1(\overline{\theta}_0^2/2)$, we have that \begin{align*}
O_1(\delta/\pi) + (1 + O_1(\delta^2/2))(r - \rho_d + O_1(7d^4\delta)) = O_1(\beta M).
\end{align*} Using $r \leqslant 6d$, $\rho_d \leqslant 2d$, and $\delta = 24d^4\pi^2\sqrt{\beta} \leqslant 1$, we get \begin{align*}
r - \rho_d & + O_1 \left(\frac{\delta^2}{2}\right)( r - \rho_d) + O_1(7d^4 \delta) + O_1\left(\frac{7d^4\delta^3}{2}\right) = O_1(\beta M) + O_1 \left(\frac{\delta}{\pi}\right)  \\
\Longrightarrow r - \rho_d & = O_1\left(4d\delta^2 + 7d^4\delta + \frac{7d^4\delta^3}{2} + 6d\beta + \frac{\delta}{\pi}\right) \\
& = O_1\left(6d\beta + \delta\left(4d + 7d^4 + \frac{7d^4}{2} + \frac{1}{\pi}\right)\right) \\
& = O_1\left(\left(6d + 24d^4\pi^2\left(4d + \frac{21d^4}{2} + \frac{1}{\pi}\right)\right)\sqrt{\beta}\right) \\
& = O_1(3517d^8\sqrt{\beta}).
\end{align*}

Finally, to complete the proof we use the inequality \begin{align*}
\|x - x_*\| \leqslant |\|x\| - \|x_*\|| + \left(\|x_*\| + |\|x\| - \|x_*\||\right)\overline{\theta}_0.
\end{align*} This inequality states that if a two dimensional point is known to be within $\Delta r$ of magnitude $r$ and an angle $\Delta \theta$ away from $0$, then it is at most a Euclidean distance of $\Delta r + (r + \Delta r)\Delta \theta$ away from the point $(r,0)$ in polar coordinates. Thus for $\overline{\theta}_0 = O_1(82\pi d^4\beta)$, we have $r = 1 + O_1(838\pi d^5\beta)$ so \begin{align*}
    \|x-x_*\| & \leqslant 838\pi d^5\beta + (1 + 838\pi d^5\beta)82\pi d^4\beta \leqslant 70000\pi^2 d^9\beta.
\end{align*} Then if $\overline{\theta}_0 = \pi + O_1(24d^4 \pi^2 \sqrt{\beta})$, $r = \rho_d + O_1(3517d^8\sqrt{\beta})$ so that \begin{align*}
    \|x + \rho_d x_*\| &\leqslant 3517d^8\sqrt{\beta} + (\rho_d + 3517d^8\sqrt{\beta})24d^4\pi^2\sqrt{\beta}    %& \leqslant 3517d^8\sqrt{\beta} + (2d + 3517d^8\sqrt{\beta})24d^4\pi^2\sqrt{\beta} \\
     \leqslant 77422\pi^2d^{12}\sqrt{\beta}.
\end{align*} Hence we attain \begin{align*}
\Seps_{\beta} \subset \mathcal{B}(x_*,70000\pi^2 d^9\beta) \cup \mathcal{B}(-\rho_d x_*, 77422\pi^2d^{12}\sqrt{\beta}).
\end{align*} The result that $\rho_d \rightarrow 1$ as $d \rightarrow \infty$ follows from the following facts: by \eqref{breve_theta_bound}, we have that $\breve{\theta}_d \rightarrow 0\ \text{as}\ d \rightarrow \infty$ which implies $\frac{2\sin \breve{\theta}_d}{\pi} \rightarrow 0\ \text{as}\ d \rightarrow \infty.$ Moreover, in \cite{HV2017}, it was shown that \begin{align*}
\sum_{i=0}^{d-1}\frac{\sin \breve{\theta}_i}{\pi} \left(\prod_{j=i+1}^{d-1} \frac{\pi - \breve{\theta}_j}{\pi}\right) \rightarrow 1\ \text{as}\ d \rightarrow \infty.
\end{align*} Hence \begin{align*}\left(\frac{\pi - 2 \breve{\theta}_d}{\pi}\right) \sum_{i=0}^{d-1}\frac{\sin \breve{\theta}_i}{\pi} \left(\prod_{j=i+1}^{d-1} \frac{\pi - \breve{\theta}_j}{\pi}\right) \rightarrow 1\ \text{as}\ d \rightarrow \infty
\end{align*} so $\rho_d \rightarrow 1$ as $d \rightarrow \infty$.
\end{proof}

We now aim to show that the objective function value for points near the minimizer are lower than near the negative multiple which is formally stated in Lemma \ref{obj_function_bigger_near_negmult}. We first define \begin{align}
    f_0(x) := \frac{1}{2} \Big\| |A\G(x)| - |A\G(x_*)| \Big\|^2 \nonumber
\end{align} which is the objective function without noise and $f_{\eta}(x) = f_0(x) - \langle |A\G(x)| - |A\G(x_*)|, \eta\rangle.$ Then note that $f(x) = f_{\eta}(x) + \frac{1}{2}\|\eta\|^2$. We will first show that the objective function without noise can be closely approximated by a particular function $\mathcal{F}$ which is defined by \begin{align}\begin{split}
    \F(x) & := \frac{1}{2^{d+1}}(\|x\|^2 + \|x_*\|^2) - \frac{1}{2^d}\left(\frac{\pi - 2 \overline{\theta}_d}{\pi}\right)\left(\prod_{j=0}^{d-1}\frac{\pi - \overline{\theta}_{j}}{\pi}\right)\langle x,x_* \rangle \label{expected_loss_full} \\
    & - \frac{1}{2^d}\left(\frac{2\sin \overline{\theta}_d}{\pi} + \left(\frac{\pi - 2\overline{\theta}_d}{\pi}\right) \sum_{i=0}^{d-1} \frac{\sin \overline{\theta}_i}{\pi}\left(\prod_{j=i+1}^{d-1} \frac{\pi - \overline{\theta}_i}{\pi}\right)\right)\|x\|\|x_*\| %\label{expected_loss_2}
\end{split}\end{align} where $\overline{\theta}_0 := \angle(x,y)$ and $\overline{\theta}_i := g(\overline{\theta}_{i-1})$ for $i \in [d]$. This result is formalized in the following lemma: \begin{lem} \label{obj_close_to_expectation}
Fix $0 < \epsilon < 1/(16\pi d^2)^2$. Suppose that $A \in \R^{m \times n_d}$ satisfies the RRCP with respect to $\G$ with constant $\epsilon$ and $\G$ is such that each $W_i \in \R^{n_i \times n_{i-1}}$ satisfies the WDC with constant $\epsilon$ for $i \in [d]$. Then we have that for all non-zero $x,x_* \in \R^k$: \begin{align*}
    |f_0(x) - \mathcal{F}(x)| \leqslant & \frac{(L + 12)d^3\sqrt{\epsilon}}{2^d}\|x\|^2 + \frac{(L+12)d^3\sqrt{\epsilon}}{2^d}\|x_*\|^2 + \frac{2L\epsilon}{2^d}\|x\|\|x_*\|  \\
    & + \frac{1}{2^d}\left[24d^3 + 8d\left(1 + \frac{d}{\pi}\right) + \frac{48d + 48d^3}{\pi}\right]\sqrt{\epsilon}\|x\|\|x_*\|.
\end{align*}
\end{lem}

\begin{proof}[Proof of Lemma \ref{obj_close_to_expectation}] Fix $x,x_* \in \R^{k} \setminus \{0\}$. For notational simplicity, define \begin{align*}
    \xi_{x,x_*} := \frac{\pi - 2\overline{\theta}_d}{\pi}\left(\prod_{i=0}^{d-1} \frac{\pi - \overline{\theta}_i}{\pi}\right)\langle x,x_*\rangle  + \left(\frac{\pi - 2\overline{\theta}_d}{\pi}\sum_{i=0}^{d-1}\frac{\sin\overline{\theta}_i}{\pi} \left(\prod_{j=i+1}^{d-1} \frac{\pi - \overline{\theta}_i}{\pi}\right) + \frac{2\sin\overline{\theta}_d}{\pi} \right)\|x_*\|\|x\|.
\end{align*} Then observe that $\mathcal{F}$ can be written more compactly as 
    $\F(x) = \frac{1}{2^{d+1}}(\|x\|^2 + \|x_*\|^2) - \frac{1}{2^d} \xi_{x,x_*}.$
Then the following bound shows we need to approximate three particular terms: \begin{align*}
|f_0(x) - \mathcal{F}(x)| & \leqslant \frac{1}{2}\left|\|A_{x_d}x_d\|^2 - \frac{1}{2^d}\|x\|^2\right| + \frac{1}{2}\left|\|A_{x_{*,d}}x_{*,d}\|^2 - \frac{1}{2^d}\|x_*\|^2\right|  \\
& + \left| \langle A_{x_d}x_d, A_{x_{*,d}}x_{*,d}\rangle - \frac{1}{2^d}\xi_{x,x_*}\right|.
\end{align*} 

Bounds on the first two terms follow directly by the RRCP and WDC in the following way. Note that \begin{align}
    \left|\|A_{x_d}x_d\|^2 - \frac{1}{2^d}\|x\|^2\right| & \leqslant \left|\|A_{x_d}x_d\|^2 - \|x_d\|^2\right| + \left| \|x_d\|^2 - \frac{1}{2^d}\|x\|^2\right|. \nonumber
\end{align} Since $A$ satisfies the RRCP with respect to $\G$, we have that \begin{align*}
    \left|\|A_{x_d}x_d\|^2 - \|x_d\|^2\right| \leqslant L\epsilon\|\Lambda_x\|^2\|x\|^2 \leqslant  L\epsilon\left(\frac{1}{2} + \epsilon\right)^d\|x\|^2
\end{align*} where the last inequality follows by the WDC. Then by \eqref{conc_htilde}, we have  \begin{align*}
     \left| \|x_d\|^2 - \frac{1}{2^d}\|x\|^2\right|\leqslant  24 \frac{d^3\sqrt{\epsilon}}{2^d}\|x\|^2
\end{align*}  Using these two bounds, we have \begin{align}
\left|\|A_{x_d}x_d\|^2 - \frac{1}{2^d}\|x\|^2\right| %& \leqslant \left|x_d^\top A_{x_d}^\top A_{x_d} x_d - x_d^\top x_d\right| + \left| x_d^\top x_d - \frac{1}{2^d}\|x\|^2\right| \nonumber\\
& \leqslant L\epsilon\left(\frac{1}{2} + \epsilon\right)^d\|x\|^2 + 24 \frac{d^3\sqrt{\epsilon}}{2^d}\|x\|^2 \nonumber\\
%& \leqslant \frac{L\epsilon(1+2\epsilon)^d}{2^d}\|x\|^2 + 24 \frac{d^3 \sqrt{\epsilon}}{2^d}\|x\|^2 \nonumber\\
%& \leqslant \frac{L\epsilon(1 + 4\epsilon d)}{2^d}\|x\|^2 + 24\frac{d^3 \sqrt{\epsilon}}{2^d}\|x\|^2 \nonumber\\
& \leqslant \frac{(2L + 24)d^3\sqrt{\epsilon}}{2^d}\|x\|^2 \label{conc_to_F_first_bound}
\end{align} since $(1+2\epsilon)^d \leqslant e^{2\epsilon d} \leqslant 1+4\epsilon d \leqslant 2$ for $\epsilon \leqslant 1/(4d)$. By the same logic, we have that  
\begin{align}
\left|\|A_{x_{*,d}}x_{*,d}\|^2 - \frac{1}{2^d}\|x_*\|^2\right| \leqslant \frac{(2L + 24)d^3\sqrt{\epsilon}}{2^d}\|x_*\|^2. \label{conc_to_F_second_bound}
\end{align} 

For the last term, note that \begin{align*}
\left| \langle A_{x_d}x_d, A_{x_{*,d}}x_{*,d}\rangle - \frac{1}{2^d}\xi_{x,x_*}\right| & \leqslant \left| \langle A_{x_d}x_d, A_{x_{*,d}}x_{*,d}\rangle - \langle \Phi_{x_d,x_{*,d}} x_d, x_{*,d}\rangle\right| + \left|\langle \Phi_{x_d,x_{*,d}} x_d, x_{*,d}\rangle - \frac{1}{2^d}\xi_{x,x_*}\right|.
\end{align*} For the first term, the RRCP and WDC imply
\begin{align}
& \left| \langle A_{x_d}x_d, A_{x_{*,d}}x_{*,d}\rangle - \langle \Phi_{x_d,x_{*,d}} x_d, x_{*,d}\rangle\right| \leqslant L\epsilon \left(\frac{1}{2}+\epsilon\right)^d\|x\|\|x_*\| \leqslant \frac{2L\epsilon}{2^d}\|x\|\|x_*\| \label{A_to_Phi_bound}
\end{align} for $\epsilon \leqslant 1/(4d)$. For the second term, by the definition of $\Phi_{z,w}$ and $M_{\hat{z} \leftrightarrow \hat{w}}$, we have 
\begin{align*}
\left|\langle \Phi_{x_d,x_{*,d}} x_d, x_{*,d}\rangle - \frac{1}{2^d}\xi_{x,x_*}\right| & \leqslant \|x\|\underbrace{\left\|\frac{\pi - 2\theta_d}{\pi}\Lambda_x^\top\Lambda_{x_*}x_* - \frac{\pi - 2\overline{\theta}_d}{\pi}\tilde{h}_{x,x_*}\right\|}_{(I)}\\
& +\|x\|\underbrace{\left\|\frac{2\sin \theta_d}{\pi}\frac{\|x_{*,d}\|}{\|x_d\|}\Lambda_x^\top \Lambda_x x - \frac{2\sin \overline{\theta}_d}{\pi} \frac{\|x_*\|}{\|x\|}\frac{1}{2^d}x\right\|}_{(II)}
\end{align*} where $\tilde{h}_{x,x_*}$ is defined in \eqref{htilde_def}. It was shown in the proof of Lemma \ref{conc_wx_to_hx} that \begin{align*}
(I) \leqslant \frac{1}{2^d} \left(24 d^3 + \frac{8d}{\pi}\left( 1+ \frac{d}{\pi}\right)\right)\sqrt{\epsilon}\|x_*\|
\end{align*} and \begin{align*}
(II) & %\leqslant Q_{3,1} + Q_{3,2} + Q_{3,3} \\ 
%& \leqslant \frac{8d(1 + 2\epsilon)^d}{\pi 2^d}\sqrt{\epsilon}\|x_*\|+ \frac{16d(1+2\epsilon)^d}{\pi 2^d}\sqrt{\epsilon}\|x_*\|+ \frac{48d^3\sqrt{\epsilon}}{\pi2^d}\|x_*\| \\
\leqslant \frac{1}{2^d}\left(\frac{24d(1 + 2\epsilon)^d + 48d^3}{\pi}\right)\sqrt{\epsilon}\|x_*\|.
\end{align*} Combining the results for \eqref{A_to_Phi_bound}, (I), and (II) we have \begin{align}
\begin{split}
\left| \langle A_{x_d}x_d, A_{x_{*,d}}x_{*,d}\rangle - \frac{1}{2^d}\xi_{x,x_*}\right| %& \leqslant \left| x_d^\top A_{x_d}^\top A_{x_{*,d}}x_{*,d} - x_d^\top \Phi_{x_d,x_{*,d}} x_{*,d}\right| + \left|x_d^\top \Phi_{x_d,x_{*,d}} x_{*,d} - \frac{1}{2^d}\xi_{x,x_*}\right| \nonumber \\
& \leqslant \frac{2L\epsilon}{2^d}\|x\|\|x_*\| + \frac{1}{2^d} \left(24 d^3 + \frac{8d}{\pi}\left( 1+ \frac{d}{\pi}\right)\right)\sqrt{\epsilon}\|x_*\| \label{conc_to_F_third_bound}\\
& + \frac{1}{2^d}\left(\frac{24d(1 + 2\epsilon)^d + 48d^3}{\pi}\right)\sqrt{\epsilon}\|x_*\|. 
\end{split}
\end{align} Combining equations \eqref{conc_to_F_first_bound}, \eqref{conc_to_F_second_bound}, and \eqref{conc_to_F_third_bound} achieves the desired result.

\end{proof}

Now that we have established that the objective function without noise can be approximated by $\mathcal{F}$, we now show that $\F$ satsifies particular quadratic upper and lower bounds to establish the desired properties of the true objective function $f$:
\begin{lem} \label{poly_bounds_on_conc_objective}
Fix $0 < r < \frac{1}{4d^2\pi}$ and let $\kappa := \min_{d \geqslant 2} \rho_d > 0$. Then for any $\phi_d \in [\rho_d, 1]$, we have that \begin{align}
    \mathcal{F}(x) & \leqslant \frac{\|x_*\|^2}{2^{d+1}}\left(\phi_d^2 - 2\phi_d + \frac{45d}{\kappa^3}r\right) + \frac{\|x_*\|^2}{2^{d+1}}\ \forall\ x \in \mathcal{B}(\phi_d x_*, r \|x_*\|), \label{conc_obj_upper_bound}\\
    \mathcal{F}(y) & \geqslant \frac{\|x_*\|^2}{2^{d+1}}(\phi_d^2 - 2\rho_d \phi_d - 139d^4r) + \frac{\|x_*\|^2}{2^{d+1}}\ \forall\ y \in \mathcal{B}(-\phi_d x_*, r \|x_*\|). \label{conc_obj_lower_bound}
\end{align}
\end{lem}
\begin{proof}[Proof of Lemma \ref{poly_bounds_on_conc_objective}]
Define 
    $\psi_d : = \frac{\pi - 2\overline{\theta}_d}{\pi}, \zeta_{i+1} : = \prod_{j=i+1}^{d-1}\frac{\pi - \overline{\theta}_{j}}{\pi},$ and $\al_i := \frac{\sin \overline{\theta}_i}{\pi}.$ Then note that we can write $\mathcal{F}$ as  \begin{align}
    \mathcal{F}(x) : = \frac{1}{2^{d+1}}(\|x\|^2 + \|x_*\|^2) - \frac{1}{2^d}\left(\psi_d \zeta_0 \langle x, x_*\rangle + \left(2\al_d + \psi_d\sum_{i=0}^{d-1}\al_i\zeta_{i+1}\right)\|x\|\|x_*\|\right). \nonumber
\end{align}Fix $x \in \mathcal{B}(\phi_d x_*, r\|x_*\|)$. Then observe that we have $\theta_0 \leqslant \frac{\pi r}{2\phi_d}$ and $(\phi_d - r)\|x_*\| \leqslant \|x\| \leqslant (\phi_d + r)\|x_*\|$. Furthermore, $\cos \theta_0 \geqslant 1 - \frac{\theta_0^2}{2}.$ Thus, we have the following bounds: \begin{align*}
    \psi_d \geqslant 1 - \frac{r}{\phi_d},\ \zeta_0 \geqslant \prod_{i=0}^{d-1}\left(1 - \frac{r}{2\phi_d}\right),\ \text{and}\ \cos \theta_0 \geqslant 1 - \frac{\pi^2r^2}{8\phi_d}.
\end{align*} Hence we see that \begin{align*}
    \mathcal{F}(x) - \frac{\|x_*\|^2}{2^{d+1}} & \leqslant \frac{\|x\|^2}{2^{d+1}} - \frac{1}{2^d}\psi_d \zeta_0 \cos \theta_0 \|x\|\|x_*\| \\
    & \leqslant \frac{1}{2^{d+1}}(\phi_d + r)^2\|x_*\|^2 - \frac{}{2^d}\psi_d \zeta_0 \cos \theta_0 (\phi_d -r)\|x_*\|^2 \\
    & \leqslant \frac{\|x_*\|^2}{2^{d+1}}\left(\phi_d^2 + 2r\phi_d + r^2 - 2 \left(1 - \frac{r}{\phi_d}\right)\left(1 - \frac{dr}{\phi_d}\right)(\phi_d - r)\left(1 - \frac{\pi^2r^2}{8\phi_d^2}\right)\right).
\end{align*} %We can expand this term further to obtain \begin{align*}
    %& \phi_d^2 + 2r\phi_d + r^2 - 2 \left(1 - \frac{r}{\phi_d}\right)\left(1 - \frac{dr}{\phi_d}\right)(\phi_d - r)\left(1 - \frac{\pi^2r^2}{8\phi_d^2}\right) \\
    %& \leqslant \phi_d^2 + 2r\phi_d + r^2 - 2\phi_d + \frac{\pi^2r^2}{4\phi_d} + 2(d+2)r - \frac{(d+2)\pi^2r^3}{4\phi_d} - \frac{2(2d+1)r^2}{\phi_d} + \frac{(2d+1)\pi^2r^4}{4\phi_d^3} + \frac{2dr^3}{\phi_d^2} - \frac{2d\pi^2r^5}{\phi_d^4}.
%\end{align*} 
 where in the first inequality we used $2\al_d + \psi_d\sum_{i=0}^{d-1}\al_i\zeta_{i+1}\geqslant 0.$ Noting that $\phi_d \in [\rho_d,1]$ and $r < 1$, with some algebra we attain \begin{align*}
    & \phi_d^2 + 2r\phi_d + r^2 - 2 \left(1 - \frac{r}{\phi_d}\right)\left(1 - \frac{dr}{\phi_d}\right)(\phi_d - r)\left(1 - \frac{\pi^2r^2}{8\phi_d^2}\right) \\
    & \leqslant \phi_d^2 - 2\phi_d + \frac{(8d + (2d+1)\pi^2 + 7)r}{\phi_d^3} \\
    & \leqslant \phi_d^2 - 2\phi_d + \frac{45d}{\kappa^3}r
\end{align*} so we may conclude that for $x \in \mathcal{B}(\phi_d x_*, r\|x_*\|)$, \begin{align*}
    \mathcal{F}(x) - \frac{\|x_*\|^2}{2^{d+1}} & \leqslant \frac{\|x_*\|^2}{2^{d+1}}\left(\phi_d^2 - 2\phi_d + \frac{45d}{\kappa^3}r\right). 
\end{align*}

Fix $x \in \mathcal{B}(-\phi_d x_*, r\|x_*\|)$. Then note that we have $\pi - \theta_0 \leqslant \frac{\pi^2}{2}r$ and $(\phi_d - r)\|x_*\| \leqslant \|x\| \leqslant (\phi_d + r)\|x_*\|$. Furthermore, for sufficiently small $r > 0$, we have that $\langle x,x_*\rangle \leqslant 0$ so that $-\psi_d \zeta_0 \langle x,x_*\rangle \geqslant 0$ (note that $\overline{\theta}_d \leqslant \pi/2$). Thus \begin{align*}
    \mathcal{F}(x) - \frac{1}{2^{d+1}}\|x_*\|^2 & = \frac{1}{2^{d+1}}\|x\|^2 - \frac{1}{2^d}\left(\psi_d \zeta_0 \langle x,x_*\rangle + (\psi_d \sum_{i=0}^{d-1} \al_i \zeta_{i+1} + 2\al_d)\|x\|\|x_*\|\right) \\
    & \geqslant \frac{1}{2^{d+1}}(\phi_d - r)^2\|x_*\|^2 - \frac{1}{2^d}(\psi_d \sum_{i=0}^{d-1} \al_i \zeta_{i+1} + 2\al_d)(\phi_d + r)\|x_*\|.
\end{align*} Note that we have $\overline{\theta}_0 = \pi + O_1(r\pi^2/2)$. As shown in Proposition \ref{Seps}, if $d^2(r\pi^2/2)/\pi \leqslant 1$, then we have that  $$\psi_d \sum_{i=0}^{d-1} \al_i \zeta_{i+1} + 2\al_d = \rho_d + O_1(7d^4r\pi^2/2).$$ Hence we have\begin{align*}
    \mathcal{F}(x) - \frac{\|x_*\|^2}{2^{d+1}} & \geqslant \frac{1}{2^{d+1}}(\phi_d - r)^2\|x_*\|^2 - \frac{1}{2^d}(\rho_d + 7\pi^2d^4r/2)(\phi_d + r)\|x_*\|^2 \\
    & = \frac{\|x_*\|^2}{2^{d+1}}(\phi_d^2 - 2r\phi_d + r^2 - 2(\rho_d\phi_d + r\rho_d + 7\pi^2d^4r\phi_d/2 + 7\pi^2d^4r^2/2))\|x_*\|^2 \\
    & \geqslant \frac{\|x_*\|^2}{2^{d+1}}(\phi_d^2 - 2r - 2\rho_d\phi_d -2r -7\pi^2d^4r -7\pi^2d^4r^2)\\
    %& \geqslant \frac{\|x_*\|^2}{2^{d+1}}(\phi_d^2 - 2\rho_d \phi_d - (4 + 4Kd^4)r)\\
    & \geqslant \frac{\|x_*\|^2}{2^{d+1}}(\phi_d^2 - 2\rho_d \phi_d - 139d^4r)
\end{align*} where we used the fact that $\phi_d \in [\rho_d , 1]$ and $0 < r < 1.$ This completes the proof.
\end{proof}

With this result, we are equipped to prove Lemma \ref{obj_function_bigger_near_negmult}.

\begin{proof}[Proof of Lemma \ref{obj_function_bigger_near_negmult}] By the same argument for \eqref{bound_on_q_norm}, we have that $|\langle A_{x_d}\Lambda_x x, \eta\rangle| \leqslant \frac{2}{2^{d/2}}\|x\|\|\eta\|$ for any $x \in \R^k$. Thus for $x \in \mathcal{B}(\phi_d x_*, \varphi \|x_*\|)$, \begin{align}
    |\langle |A\G(x)| - |A\G(x_*)|,\eta\rangle| & \leqslant |\langle A_{x_d}\Lambda_x x, \eta \rangle| + |\langle A_{x_{*,d}} \Lambda_{x_*} x_*,\eta\rangle| \nonumber \\
    & \leqslant (\|x\| + \|x_*\|) \frac{2}{2^{d/2}}\|\eta\| \nonumber \\
    & \leqslant (\varphi \|x_*\| + 2\|x_*\|)\frac{2}{2^{d/2}}\|\eta\| \nonumber
\end{align} where we used the fact that $\|x\| \leqslant (\phi_d + \varphi)\|x_*\| \leqslant (1 + \varphi)\|x_*\|$ in the last inequality.

Let $\kappa:=\min_{d \geqslant 2} \rho_d$. If $x \in \mathcal{B}(\phi_d x_*, \varphi \|x_*\|)$ and $K_d := 24d^3 + 8d\left(1 + \frac{d}{\pi}\right) + \frac{48d + 48d^3}{\pi}$, then Lemma \ref{obj_close_to_expectation} and Lemma \ref{poly_bounds_on_conc_objective} give \begin{align*}
    f_{\eta}(x) & \leqslant \mathcal{F}(x) + |f_0(x) - \mathcal{F}(x)| + |\langle |A\G(x)| - |A\G(x_*)|,\eta\rangle|\\
    & \leqslant \frac{\|x_*\|^2}{2^{d+1}}\left(\phi_d^2 - 2\phi_d + \frac{45d}{\kappa^3}\varphi\right) + \frac{\|x_*\|^2}{2^{d+1}} + \frac{(L+12)d^3\sqrt{\epsilon}}{2^d}\|x\|^2 \\
    & + \frac{(L+12)d^3\sqrt{\epsilon}}{2^d}\|x_*\|^2 + \frac{2L\epsilon}{2^d}\|x\|\|x_*\| + \frac{1}{2^d}K_d\sqrt{\epsilon}\|x\|\|x_*\|+  (\varphi \|x_*\| + 2\|x_*\|)\frac{2}{2^{d/2}}\|\eta\| \\
    & \leqslant \frac{\|x_*\|^2}{2^{d+1}}(\phi_d^2 - 2\phi_d + \frac{45d}{\kappa^3}\varphi + 1 + 2(L+12)d^3\sqrt{\epsilon}\left((\phi_d + \varphi)^2 + 1\right) + 4L\sqrt{\epsilon}(\phi_d + \varphi) + 2K_d(\phi_d + \varphi)) \\
    & + (\varphi \|x_*\| + 2\|x_*\|)\frac{2}{2^{d/2}}\|\eta\| \\
    & \leqslant \frac{\|x_*\|^2}{2^{d+1}}(1 + \phi_d^2 - 2\phi_d + \frac{45d}{\kappa^3}\sqrt{\epsilon} + \tilde{K}_d\sqrt{\epsilon}) + (\varphi \|x_*\| + 2\|x_*\|)\frac{2}{2^{d/2}}\|\eta\|
\end{align*} where $\tilde{K}_d := 6(L+12)d^3 + 8L + 4K_d$, in the second inequality we used $\|x\|\leqslant(\phi_d + \varphi)\|x_*\|$, and in the last inequality we used $\epsilon < \sqrt{\epsilon}$, $\rho_d \leqslant 1$ and $\varphi < 1$. 

Similarly, if $y \in \mathcal{B}(-\phi_d x_*,\varphi \|x_*\|)$, then \begin{align*}
    f_{\eta}(y) & \geqslant \mathcal{F}(y) - |f_0(y) - \mathcal{F}(y)| - |\langle |A\G(y)| - |A\G(x_*)|,\eta\rangle|\\ 
    & \geqslant \frac{\|x_*\|^2}{2^{d+1}}(\phi_d^2 - 2\rho_d \phi_d - 139d^4\varphi) + \frac{\|x_*\|^2}{2^{d+1}} - \frac{(L+12)d^3\sqrt{\epsilon}}{2^d}\|x\|^2 \\
    & - \frac{(L+12)d^3\sqrt{\epsilon}}{2^d}\|x_*\|^2 - \frac{2L\epsilon}{2^d}\|x\|\|x_*\| - \frac{1}{2^d}K_d\sqrt{\epsilon}\|x\|\|x_*\|-  (\varphi \|x_*\| + 2\|x_*\|)\frac{2}{2^{d/2}}\|\eta\|\\
    & \geqslant \frac{\|x_*\|^2}{2^{d+1}}(1 + \phi_d^2 - 2\rho_d \phi_d - 139d^4\sqrt{\epsilon} - \tilde{K}_d\sqrt{\epsilon}) - (\varphi \|x_*\| + 2\|x_*\|)\frac{2}{2^{d/2}}\|\eta\|.
\end{align*} In sum, we have for $x \in \mathcal{B}(\phi_d x_*,\varphi\|x_*\|)$, \begin{align}
    f_{\eta}(x) \leqslant \frac{\|x_*\|^2}{2^{d+1}} \left(1 + \phi_d^2 - 2\phi_d + \frac{45d\sqrt{\epsilon}}{\kappa^3} + \tilde{K}_d\sqrt{\epsilon}\right) + (\varphi \|x_*\| + 2\|x_*\|)\frac{2}{2^{d/2}}\|\eta\| \label{lemma5upperbound_quantity}
\end{align} while for $y \in \mathcal{B}(-\phi_d x_*,\varphi\|x_*\|)$, \begin{align}
    f_{\eta}(y) \geqslant \frac{\|x_*\|^2}{2^{d+1}}\left(1 + \phi_d^2 - 2\rho_d \phi_d - 139d^4 \sqrt{\epsilon} - \tilde{K}_d\sqrt{\epsilon}\right)- (\varphi \|x_*\| + 2\|x_*\|)\frac{2}{2^{d/2}}\|\eta\| \label{lemma5lowerbound_quantity}.
\end{align} 

Note that we require the lower bound in \eqref{lemma5lowerbound_quantity} to be larger than the upper bound in \eqref{lemma5upperbound_quantity}. Setting $\varphi = \epsilon$ and using both $\epsilon < \sqrt{\epsilon}$ and $\|\eta\| \leqslant c_2\frac{\|x_*\|}{2^{d/2}d^{48}}$, we see that we require %\begin{align}
%    \left(\frac{45d}{\kappa^3} + 2\tilde{K}_d  + 139d^4\right)\sqrt{\epsilon} < 2\phi_d - 2\rho_d\phi_d - 8\|\eta\| \label{boundonepsilontogetfunctionbound}
%\end{align} Thus if 
\begin{align}
    \epsilon \leqslant \left(\frac{2\phi_d(1-\rho_d) - 8c_2/d^{48}}{\frac{45d}{\kappa^3} + 2\tilde{K}_d + 139d^4}\right)^2. \label{boundonepsilontogetfunctionbound}
\end{align} By Lemma \ref{bounds_on_rho_lemma}, we have  $1-\rho_d \geqslant 1/(C(d+2)^2)$ for some numerical constant $C$ and $\phi_d \geqslant \kappa$. Hence it suffices to have $\frac{c_2}{d^{48}} \leqslant \frac{\kappa}{8C(d+2)^2}$  and $\varphi = \epsilon \leqslant r_1/d^{12}$ for some numerical constants $r_1$ and $c_2$. %Thus, as a concrete bound that only depends on $d$, we can take \begin{align*}
    %\epsilon \leqslant \left(\frac{\frac{2\kappa}{C(d+2)^2} - 8\|\eta\|}{\frac{45d}{\kappa^3} + 139d^4 + 2\tilde{K}_d}\right)^2.
%\end{align*} 
\end{proof}

\begin{lem} \label{bounds_on_rho_lemma}
We have that $\rho_d$ satisfies $\min_{d \geqslant 2} \rho_d > 0$ and for some numerical constant $C$, \begin{align*}
    \frac{1}{C(d+2)^2} \leqslant 1-\rho_d\ \forall\ d \geqslant 2.
\end{align*}
\end{lem}

\begin{proof}[Proof of Lemma \ref{bounds_on_rho_lemma}]
%Recall that we defined \begin{align*}
%\rho_d : = \frac{2 \sin \breve{\theta}_d}{\pi} + \left(\frac{\pi - 2 \breve{\theta}_d}{\pi}\right) \sum_{i=0}^{d-1}\frac{\sin \breve{\theta}_i}{\pi} \left(\prod_{j=i+1}^{d-1} \frac{\pi - \breve{\theta}_j}{\pi}\right)
%\end{align*} where $\breve{\theta}_0 = \pi$ and $\breve{\theta}_i = g(\breve{\theta}_{i-1})$. 
Let $\Gamma_d : = \sum_{i=0}^{d-1}\frac{\sin \breve{\theta}_i}{\pi} \left(\prod_{j=i+1}^{d-1} \frac{\pi - \breve{\theta}_j}{\pi}\right)$. In Lemma A.4 of \cite{Huangetal2018}, it has been established that $\Gamma_d \in [0,1]$ and $\min_{d \geqslant 2} \Gamma_d > 0$. By \eqref{breve_theta_bound} and \eqref{breve_theta_lower_bound}, we have that $\breve{\theta}_d \leqslant 3\pi/(d+3)$ and $\breve{\theta}_d \geqslant \pi/(d+1)$ for all $d \geqslant 2$. Since $\sin(2x) - 3x/4 \geqslant 0$ for all $x \in [0,\pi/3]$, observe that \begin{align*}
    \frac{2\sin \breve{\theta}_d}{\pi} \geqslant \frac{2\sin \left(\frac{\pi}{d+1}\right)}{\pi} \geqslant  \frac{2}{\pi}\cdot\frac{3}{4}\left(\frac{\pi}{2(d+1)}\right) = \frac{3}{4(d+1)}\ \forall\ d \geqslant 2.
\end{align*} Thus for any $d \geqslant 2$,\begin{align*}
    \rho_d = \frac{2 \sin \breve{\theta}_d}{\pi} + \left(\frac{\pi - 2 \breve{\theta}_d}{\pi}\right) \Gamma_d 
    & \geqslant \frac{3}{4(d+1)} + \frac{\pi - \frac{6\pi}{d+3}}{\pi} \Gamma_d \\
    & \geqslant \frac{3}{4(d+1)} \Gamma_d + \frac{\pi - \frac{6\pi}{d+3}}{\pi} \Gamma_d \\
    %& = \left(\frac{3}{4(d+1)} + 1 - \frac{6}{d+3}\right)\Gamma_d \\
    & = \left(\frac{3}{4(d+1)} + \frac{d-3}{d+3}\right)\Gamma_d  \geqslant \frac{1}{20}\Gamma_d
\end{align*}  where the second inequality is due to $\Gamma_d \in [0,1]$. We conclude that $\min_{d \geqslant 2} \rho_d \geqslant 1/20 \min_{d \geqslant 2} \Gamma_d > 0.$

We now establish the lower bound on $1-\rho_d$ for all $d \geqslant 2$. It was shown in Lemma A.4 of \cite{Huangetal2018} that $1-\Gamma_d \geqslant \frac{1}{a_7(d+2)^2}$ for some numerical constant $a_7$. Observe that \begin{align*}
    \rho_d = \left(1 - \frac{2}{\pi}\breve{\theta}_d\right)\Gamma_d + \frac{2}{\pi}\sin \breve{\theta}_d = \Gamma_d + \frac{2}{\pi}\left(\sin \breve{\theta}_d - \breve{\theta}_d\Gamma_d\right) \leqslant \Gamma_d + \frac{2}{\pi}\breve{\theta}_d (1-\Gamma_d). 
\end{align*} Furthermore, note that for all $d \geqslant 2$, $\breve{\theta}_d \leqslant \breve{\theta}_2 = g(g(\pi)) = g(\pi/2) = \cos^{-1}(1/\pi).$ Hence for all $d \geqslant 2$, \begin{align*}
    1 - \rho_d 
    \geqslant 1- \Gamma_d - \frac{2}{\pi}\breve{\theta}_d(1-\Gamma_d) 
    & = (1-\Gamma_d) \left(1 - \frac{2}{\pi}\breve{\theta}_d\right) \\
    & \geqslant \frac{1}{a_7(d+2)^2}\left(1 - \frac{2}{\pi}\cos^{-1}\left(\frac{1}{\pi}\right)\right) \\
    & \geqslant \frac{0.2}{a_7(d+2)^2}.
\end{align*}
\end{proof}

\subsubsection{Proofs for Section \ref{conv_proof_third_subsection}} \label{proof_of_convexity_lemma_subsection}

Here we prove the convexity-like property of $f$ around the minimizer $x_*$.

\begin{proof}[Proof of Lemma \ref{convexity_lemma}]

Suppose our objective function $f$ is differentiable at $x$. Recall that the gradient of $f$ is given by $v_x = \overline{v}_x - q_x$ where
    $\overline{v}_x = \Lambda_x^\top A_{x_d}^\top (A_{x_d}\Lambda_x x - A_{x_{*,d}}\Lambda_{x_*}x_*)$ and $q_x = \Lambda_x^\top A_{x_d}^\top \eta$. We will first show that $\overline{v}_x$ satisfies \begin{align*}
    \left\| \overline{v}_x - \Lambda_x^\top (\Lambda_xx - \Lambda_xx_*) \right\| \leqslant \frac{1}{16}\frac{1}{2^d}\|x-x_*\|.
\end{align*} Note that by the triangle inequality, we have that \begin{align*}
    \left\|\overline{v}_x - \Lambda_x^\top (\Lambda_xx - \Lambda_xx_*)\right\| & \leqslant \underbrace{\left\|\Lambda_x^\top A_{x_d}^\top (A_{x_d}\Lambda_x x - A_{x_{d}}\Lambda_{x_*}x_*) - \Lambda_x^\top(\Lambda_x x - \Lambda_{x_*} x_*)\right\|}_{T_1} \\
    & + \underbrace{\left\|\Lambda_x^\top A_{x_d}^\top(A_{x_d} - A_{x_{*,d}})\Lambda_{x_*}x_*\right\|}_{T_2}.
\end{align*} We will establish control of each of these terms separately.

\paragraph{Controlling $T_1$:} Since $f$ is differentiable at $x$, note that by the local linearity of $\G$ we have that for sufficiently small $z \in \R^k$, $\G(x+z) - \G(x) = \Lambda_x z$. Hence for all $z$, the RRCP implies that \begin{align*}
    |\langle A_{x_d}\Lambda_x z,A_{x_d}(\Lambda_x x - \Lambda_{x_*}x_*)\rangle - \langle \Lambda_x z,\Lambda_x x - \Lambda_{x_*}x_*\rangle| \leqslant L\epsilon \|\Lambda_x\|\|\Lambda_x x - \Lambda_{x_*}x_*\|\|z\|.
\end{align*} Since this holds for all $z$, we have that \begin{align}
    \|\Lambda_x^\top A_{x_d}^\top (A_{x_d}\Lambda_x x - A_{x_{d}}\Lambda_{x_*}x_*) - \Lambda_x^\top(\Lambda_x x - \Lambda_{x_*} x_*)\| \leqslant L\epsilon \|\Lambda_x\|\|\Lambda_x x - \Lambda_{x_*}x_*\|. \label{AxdAxd_conc_identity}
\end{align}In addition, we have that by Lemma \ref{G_lipschitz}, if $\epsilon < 1/(200^4d^6)$ and $x \in \mathcal{B}(x_*, d\sqrt{\epsilon}\|x_*\|)$ then \begin{align}
    \|\Lambda_x x - \Lambda_{x_*}x_*\| \leqslant \frac{1.2}{2^{d/2}}\|x-x_*\|. \label{G_lipschitz_d}
\end{align}  Combining \eqref{AxdAxd_conc_identity}, \eqref{G_lipschitz_d}, and \eqref{bound_on_lambda_energy} in Lemma \ref{upper_bounds_on_lambda_Az} we see that \begin{align}
      \left\|\Lambda_x^\top A_{x_d}^\top (A_{x_d}\Lambda_x x - A_{x_{d}}\Lambda_{x_*}x_*) - \Lambda_x^\top(\Lambda_x x - \Lambda_{x_*} x_*)\right\| \leqslant \frac{1.2\sqrt{\frac{13}{12}}L\epsilon}{2^d}\|x - x_*\|. \label{AxdAxd_conc_identity_final}
\end{align} Thus choosing $\epsilon$ so that $\epsilon < 1/(32\cdot1.2\sqrt{13/12}L)$ in \eqref{AxdAxd_conc_identity_final} shows that \begin{align}
    T_1 = O_1\left(\frac{1}{32}\right)\frac{1}{2^d}\|x-x_*\|. \label{T1_bound}
\end{align}%Thus by \eqref{AxdAxd_conc_identity_final}), we have that 
%\begin{align}
%    \left\|\Lambda_x^\top A_{x_d}^\top A_{x_d}\Lambda_x(x-x_*) - \Lambda_x^\top\Lambda_x(x-x_*)\right\| \leqslant \frac{13}{12}\frac{L\epsilon}{2^d}\|x-x_*\| \label{T1_first_inequality}.
%\end{align} We can then control $\Lambda_x^\top \Lambda_x(x-x_*)$ by noting that the WDC implies \begin{align}
%    \left\|\Lambda_x^\top \Lambda_x - \frac{1}{2^d} I_{n_d}\right\| \leqslant \frac{4d\epsilon}{2^d}. \label{Lambdax_lambdax_conc_identity_2d}
%\end{align} This inequaltiy is proven in Lemma 8 of \cite{HV2016}. Thus by the triangle inequality, we can use \eqref{T1_first_inequality}) and \eqref{Lambdax_lambdax_conc_identity_2d}) to attain \begin{align*}
%    \left\|\Lambda_x^\top A_{x_d}^\top A_{x_d}\Lambda_x(x-x_*) - \frac{1}{2^d}(x-x_*)\right\| & \leqslant \left\|\Lambda_x^\top A_{x_d}^\top A_{x_d}\Lambda_x(x-x_*) - \frac{1}{2^d}\Lambda_x^\top \Lambda_x(x-x_*)\right\| \\
%    & + \left\|\Lambda_x^\top \Lambda_x(x-x_*) - \frac{1}{2^d}(x-x_*)\right\| \\
%    & \leqslant \frac{1}{2^d}\left(\frac{13}{12}L\epsilon + 4d\epsilon\right)\|x-x_*\|.
%\end{align*} If $\epsilon < \frac{1}{32(\frac{13}{12}L + 4d)}$ then we can conclude that \begin{align*}
%    \left\|\Lambda_x^\top A_{x_d}^\top A_{x_d}\Lambda_x(x-x_*) - \frac{1}{2^d}(x-x_*)\right\| \leqslant \frac{1}{2^d} \frac{1}{32}\|x-x_*\|
%\end{align*} i.e., $T_1$ satisfies \begin{align}
%    T_1 = \frac{1}{2^d}(x-x_*) + O_1\left(\frac{1}{32}\right)\frac{1}{2^d}\|x-x_*\|. \label{T1_bound}
%\end{align}

\paragraph{Controlling $T_2$:} We will first show that for sufficiently small $\epsilon$,  $$\|(A_{x_d} - A_{x_{*,d}})\Lambda_{x_*}x_*\|^2 \leqslant \frac{1.44(4L + \frac{48d}{\pi})\sqrt{\epsilon}}{2^d}\|x - x_*\|^2.$$ Letting $\{a_i\}_{i=1}^m$ denote the rows of $A$, observe that we can write \begin{align*}
    \left\|(A_{x_d} - A_{x_{*,d}})\Lambda_{x_*}x_*\right\|^2 & = \left\|(A_{x_d} - A_{x_{*,d}})x_{*,d}\right\|^2 \\
    & = \sum_{i=1}^m\left(\sgn(\langle a_i, x_d\rangle) - \sgn(\langle a_i , x_{*,d}\rangle )\right)^2\langle a_i, x_{*,d}\rangle^2 \\
    & \leqslant \sum_{i=1}^m\left(\sgn(\langle a_i, x_d\rangle) - \sgn(\langle a_i , x_{*,d}\rangle)\right)^2\langle a_i, (x_d -x_{*,d})\rangle^2 \\
    & = \sum_{i=1}^m\left(\one(\langle a_i, x_d\rangle \neq 0) + \one(\langle a_i , x_{*,d}\rangle \neq 0) - 2\sgn(\langle a_i, x_d\rangle\langle a_i , x_{*,d}\rangle)\right)\langle a_i, (x_d -x_{*,d})\rangle^2 \\
    & = \|A_{x_d}(x_d - x_{*,d})\|^2 + \|A_{x_{*,d}}(x_d - x_{*,d})\|^2 - 2\langle x_d - x_{*,d},A_{x_d}^\top A_{x_{*,d}}(x_d - x_{*,d})\rangle. 
\end{align*} We first establish concentration of $A_{x_d}(x_d - x_{*,d})$. Since $A$ satisfies the RRCP with respect to $\G$, we have that $$|\langle (A_{x_d}^\top A_{x_d} - I_{n_d})(x_d - x_{*,d}), x_d - x_{*,d}\rangle| \leqslant L\epsilon \|x_d - x_{*,d}\|^2$$ which ultimately gives \begin{align}
    \|A_{x_d}(x_d - x_{*,d})\|^2 \leqslant (1 + L\epsilon)\|x_d - x_{*,d}\|^2. \label{A_xd_bound_w_xdx0d}
\end{align} Likewise the same upper bound holds for $A_{x_{*,d}}(x_d - x_{*,d})$: \begin{align}
    \|A_{x_{*,d}}(x_d - x_{*,d})\|^2 \leqslant (1 + L\epsilon)\|x_d - x_{*,d}\|^2. \label{A_xdstar_bound_w_xdx0d}
\end{align} We now aim to upper bound the inner product $-2\langle x_d - x_{*,d},A_{x_d}^\top A_{x_{*,d}}(x_d - x_{*,d})\rangle$. We first note that since $A$ satisfies the RRCP, we have \begin{align*}
     |\langle x_d - x_{*,d},(A_{x_d}^\top A_{x_{*,d}} - \Phi_{x_d,x_{*,d}})(x_d - x_{*,d})\rangle| \leqslant L\epsilon \|x_d - x_{*,d}\|^2. 
\end{align*} Hence we have that \begin{align}
   \langle x_d - x_{*,d},A_{x_d}^\top A_{x_{*,d}}(x_d - x_{*,d})\rangle = \langle x_d - x_{*,d},\Phi_{x_d,x_{*,d}}(x_d - x_{*,d})\rangle + O_1(L\epsilon) \|x_d - x_{*,d}\|^2. \label{A_xd_innerprod_w_phi}
\end{align} But recall that $x \in \mathcal{B}(x_*,d\sqrt{\epsilon}\|x_*\|)$ which implies $ |\overline{\theta}_{0,x}| \leqslant 2d\sqrt{\epsilon}$. Since  $|\overline{\theta}_{d,x}| \leqslant |\overline{\theta}_{0,x}|$ we have $|\overline{\theta}_{d,x}| \leqslant 2d\sqrt{\epsilon}$. Also equation \eqref{conc_angle} gives \begin{align*}
    |\theta_{d,x} - \overline{\theta}_{d,x}| \leqslant 4d\sqrt{\epsilon}.
\end{align*} Hence we have that $|\theta_{d,x}| \leqslant 6d\sqrt{\epsilon}.$ Thus $\Phi_{x_d,x_{*,d}}$ is approximately an isometry since \begin{align*}
    \left\|\Phi_{x_d,x_{*,d}} - I\right\| & \leqslant \frac{2|\theta_{d,x}|}{\pi}\|I\| + \frac{2|\sin \theta_{d,x}|}{\pi}\|M_{\hat{x}_d \leftrightarrow \hat{x}_{*,d}}\| \leqslant \frac{24d\sqrt{\epsilon}}{\pi}.
\end{align*} This implies that \begin{align}
    \langle x_d - x_{*,d},\Phi_{x_d,x_{*,d}}(x_d - x_{*,d})\rangle & = \|x_d - x_{*,d}\|^2 + O_1\left(\frac{24d\sqrt{\epsilon}}{\pi}\right)\|x_d - x_{*,d}\|^2. \label{Phi_inner_prod_bound}
\end{align} Combining \eqref{A_xd_innerprod_w_phi} and \eqref{Phi_inner_prod_bound} we attain \begin{align*}
    \langle x_d - x_{*,d},A_{x_d}^\top A_{x_{*,d}}(x_d - x_{*,d})\rangle & = \|x_d - x_{*,d}\|^2 + O_1\left(\frac{24d\sqrt{\epsilon}}{\pi}+ L\epsilon\right)\|x_d - x_{*,d}\|^2.
\end{align*} Note that this implies that \begin{align}
    -2\langle x_d - x_{*,d},A_{x_d}^\top A_{x_{*,d}}(x_d - x_{*,d})\rangle \leqslant \left(-2 +\frac{48d\sqrt{\epsilon}}{\pi}+ 2L\epsilon\right)\|x_d - x_{*,d}\|^2. \label{A_inner_prod_neg_lowerbound}
\end{align} Returning to establishing concentration of $(A_{x_d} - A_{x_{*,d}})\Lambda_{x_*}x_*$, we can use \eqref{A_xd_bound_w_xdx0d}, \eqref{A_xdstar_bound_w_xdx0d} and \eqref{A_inner_prod_neg_lowerbound} to obtain \begin{align*}
    \left\|(A_{x_d} - A_{x_{*,d}})\Lambda_{x_*}x_*\right\|^2 & \leqslant \|A_{x_d}(x_d - x_{*,d})\|^2 + \|A_{x_{*,d}}(x_d - x_{*,d})\|^2 - 2\langle x_d - x_{*,d},A_{x_d}^\top A_{x_{*,d}}(x_d - x_{*,d})\rangle \\
    & \leqslant \left(2 + 2L\epsilon -2 +\frac{48d\sqrt{\epsilon}}{\pi}+ 2L\epsilon\right)\|x_d - x_{*,d}\|^2 \\
    & = \left(4L\epsilon + \frac{48d\sqrt{\epsilon}}{\pi}\right)\|x_d - x_{*,d}\|^2.
\end{align*} Using this inequality, equation \eqref{G_lipschitz_d}, and the fact that $\epsilon < \sqrt{\epsilon}$, we attain \begin{align}
    \left\|(A_{x_d} - A_{x_{*,d}})\Lambda_{x_*}x_*\right\|^2 & \leqslant \left(4L\epsilon + \frac{48d\sqrt{\epsilon}}{\pi}\right)\|x_d - x_{*,d}\|^2 \leqslant \frac{1.44(4L + \frac{48d}{\pi})\sqrt{\epsilon}}{2^d}\|x - x_*\|^2. \label{Axd_minus_Axdstar_bound}
\end{align} Then by equations \eqref{bound_on_lambda_energy} and \eqref{bound_on_AxdLambdax} from Lemma \ref{upper_bounds_on_lambda_Az}, we have that \begin{align}
    \|A_{x_d}\Lambda_x\| \leqslant \sqrt{1 + L\epsilon}\|\Lambda_x\| \leqslant \sqrt{\frac{13}{12}(1+ L\epsilon)}\frac{1}{2^{d/2}}. \label{bound_on_AxdLambdax2}
\end{align} Combining \eqref{bound_on_AxdLambdax2} and \eqref{Axd_minus_Axdstar_bound} and choosing $\epsilon$ so that $\sqrt{\frac{13}{12}(1+ L\epsilon)} \leqslant 2$, we attain \begin{align*}
    \|A_{x_d}\Lambda_x\|\left\|(A_{x_d} - A_{x_{*,d}})\Lambda_{x_*}x_*\right\| \leqslant 2 \sqrt{1.44\left(4L + \frac{48d}{\pi}\right)\sqrt{\epsilon}}\frac{1}{2^d}\|x-x_*\|.
\end{align*} Thus if \begin{align*}
    \epsilon^{1/4} < \frac{1}{64\sqrt{1.44\left(4L + \frac{48d}{\pi}\right)}}
\end{align*} we attain \begin{align*}
    \|A_{x_d}\Lambda_x\|\left\|(A_{x_d} - A_{x_{*,d}})\Lambda_{x_*}x_*\right\| \leqslant \frac{1}{32}\frac{1}{2^d}\|x-x_*\| \end{align*} i.e., $T_2$ satisfies
    \begin{align}
     T_2 = O_1\left(\frac{1}{32}\right)\frac{1}{2^d}\|x-x_*\|. \label{T2_bound}
\end{align} 

Combining our results for $T_1$ and $T_2$ in equations \eqref{T1_bound} and \eqref{T2_bound} we ultimately get \begin{align}
     \left\|\overline{v}_x - \Lambda_x^\top (\Lambda_xx - \Lambda_xx_*)\right\| & \leqslant \frac{1}{16}\frac{1}{2^d}\|x-x_*\|. \label{vbar_conc_lambda_x}
\end{align} To finish establishing concentration of $\overline{v}_x$, we appeal to Lemma A.9 of \cite{Huangetal2018} which showed that if $\epsilon < 1/(200^4d^6)$ and $x \in \mathcal{B}(x_*, d\sqrt{\epsilon}\|x_*\|)$ then \begin{align}
    \left\|\Lambda_x^\top (\Lambda_x x - \Lambda_{x_*} x_*) - \frac{1}{2^d}(x-x_*) \right\| \leqslant \frac{1}{16}\frac{1}{2^d}\|x-x_*\|. \label{wen_bound_on_lambda_x}
\end{align} Thus by combining equations  \eqref{vbar_conc_lambda_x} and \eqref{wen_bound_on_lambda_x}, we finally attain \begin{align*}
    \left\|\overline{v}_x - \frac{1}{2^d}(x-x_*)\right\| & \leqslant \left\|\overline{v}_x - \Lambda_x^\top (\Lambda_xx - \Lambda_xx_*)\right\| + \left\|\Lambda_x^\top (\Lambda_x x - \Lambda_{x_*} x_*) - \frac{1}{2^d}(x-x_*) \right\|  \leqslant \frac{1}{8}\frac{1}{2^d}\|x-x_*\|.
\end{align*} as desired. Including the bound on $\|q_x\|$ from \eqref{bound_on_q_norm}, we achieve the final desired result: \begin{align*}
    \left\|v_x - \frac{1}{2^d}(x-x_*)\right\| \leqslant \left\|\overline{v}_x - \frac{1}{2^d}(x-x_*)\right\| + \|q_x\| \leqslant \frac{1}{8}\frac{1}{2^d}\|x-x_*\| + \frac{2}{2^{d/2}}\|\eta\|.
\end{align*}

Finally, for non-differentiable $x \neq 0$ and $v \in \partial f(x)$, by \eqref{subdifferential_definition} we have that there exists $c_{\ell} \geqslant 0$ for $\ell \in [s]$ such that $\sum_{\ell=1}^s c_{\ell} = 1$ and $v = \sum_{\ell=1}^s c_{\ell}v_{\ell}$. Hence \begin{align*}
    \left\|v - \frac{1}{2^d}(x-x_*)\right\| \leqslant \sum_{\ell=1}^s c_{\ell}\left\|v_{\ell} - \frac{1}{2^d}(x-x_*)\right\| \leqslant \frac{1}{8}\frac{1}{2^d}\|x-x_*\| + \frac{2}{2^{d/2}}\|\eta\|.
\end{align*}
\end{proof}

\section{Gaussian Matrices Satisfy the RRCP} \label{RRCP_section}

We set out to prove that Gaussian $A$ satisfies the RRCP with respect to $\G$ with high probability. The particular result is stated as follows: \begin{prop}[Range Restricted Concentration Property (RRCP)] \label{realRRCP}
Fix $0 < \epsilon < 1$. Let $W_{i} \in \R^{n_i \times n_{i-1}}$ have i.i.d. $\mathcal{N}(0,1/n_i)$ entries for $i = 1,\dots,d$. Let $A \in \R^{m \times n_d}$ have i.i.d. $\mathcal{N}(0,1/m)$ entries independent from $\{W_i\}$. Then if $m > \tilde{C}_{\epsilon}dk \log (n_1n_2\dots n_d)$, then with probability at least $1 - \tilde{\gamma} m^{4k} \exp\left(- \frac{\tilde{c}_{\epsilon}}{2} m\right)$, we have that for all $x,y,x_1,x_2,x_3,x_4 \in \R^k$, \begin{align}
    |\langle (A_{\G(x)}^\top A_{\G(y)} & - \Phi_{\G(x),\G(y)})(\G(x_1) - \G(x_2)), \G(x_3)-\G(x_4)\rangle | \nonumber \\
    & \leqslant L\epsilon\|\G(x_1) - \G(x_2)\|\|\G(x_3)-\G(x_4)\| \nonumber
\end{align} Here $\tilde{\gamma}$ and $L$ are positive universal constants, $\tilde{c}_{\epsilon}$ depends polynomially on $\epsilon$, and $\tilde{C}_{\epsilon}$ depends polynomially on $\epsilon^{-1}$.
\end{prop}

We will prove Proposition \ref{realRRCP} via the following steps: \begin{enumerate}
    \item We first establish that for any \textit{fixed} non-zero $z,w \in \R^n$, the inner product $\langle A_z^\top A_wx,y\rangle$ concentrates around its expectation $\langle \Phi_{z,w}x,y\rangle$ for all $x$ and $y$ in a fixed $k$-dimensional subspace of $\R^n$.
    \item Then we show that this concentration holds uniformly for all $z,w,x,y$ that live in the union of a finite number of $k$-dimensional subspaces of $\R^n$.
    \item To complete the proof, we apply the result from Step 2 for all $z,w,x,y$ in the range of the generative model which precisely lives in the union of $k$-dimensional subspaces.
\end{enumerate}

\subsection{Concentration Over a Fixed Subspace}

We first show that the matrix $A_z^\top A_w$ concentrates around $\Phi_{z,w}$ for any fixed $z \neq w$ while acting on a fixed $k$-dimensional subspace $T$ of $\R^n$. We will refer to this result as the \textit{Restricted Concentration Property} (RCP).

\begin{prop}[Variant of Lemma 5.1 in \cite{Baraniuk_lemma}; RCP] \label{RCPproposition} Fix $0 < \epsilon < 1$ and $k < m$. Let $A \in \R^{m \times n}$ have i.i.d. $\mathcal{N}(0,1/m)$ entries and fix $z,w \in \R^n \setminus \{0\}$. Let $T \subset \R^n$ be a $k$-dimensional subspace. Then if $m \geqslant C k$, we have that with probability exceeding $1 - 2\exp(-c_1 m)$, \begin{align}
|\langle A_z^\top A_w x, x \rangle - \langle \Phi_{z,w} x, x \rangle| \leqslant \epsilon \|x\|^2\ \forall\ x \in T 
\end{align} and \begin{align}
|\langle A_z^\top A_w x, y \rangle - \langle \Phi_{z,w} x, y \rangle| \leqslant 3\epsilon \|x\|\|y\|\ \forall\ x,y \in T. \label{RCP_3bound}
\end{align} Furthermore, let $U = \bigcup_{i=1}^M U_i$ and $V = \bigcup_{j=1}^N V_j$ where $U_i$ and $V_j$ are subspaces of $\R^n$ of dimension at most $k$ for all $i \in [M]$ and $j \in [N]$. Then if $m \geqslant 2Ck$
\begin{align}
\left|\langle  A_{z}^\top A_w u, v \rangle - \langle \Phi_{z,w} u,v\rangle\right| \leqslant 3\epsilon \|u\|\|v\|\ \forall\ u \in U,\ v \in V, 
\end{align} with probability exceeding $1 - 2MN\exp(-c_1 m)$. Here $c_1$ depends polynomially on $\epsilon$ and $C = \Omega(\epsilon^{-1}\log \epsilon^{-1})$.
\end{prop}
 For the proof, we require the following large deviation inequality for subexponential random variables: \begin{lem}[Corollary 5.17 in \cite{Vershynin_notes}] \label{subexpbound}
Let $Y_1,\dots,Y_m$ be independent, centered, subexponential random variables. Let $K = \max_{i \in[m]} \|Y_i\|_{\psi_1}$. Then for all $\epsilon > 0$, \begin{align*}
\Pro\left(\frac{1}{m} \left|\sum_{i=1}^m Y_i \right| \geqslant \epsilon \right) \leqslant 2 \exp\left[-c \min\left(\frac{\epsilon^2}{K^2}, \frac{\epsilon}{K}\right)m\right]
\end{align*} where $c > 0$ is an absolute constant. Here $\|\cdot\|_{\psi_1}$ is the subexponential norm: $\|X\|_{\psi_1} := \sup_{p \geqslant 1} p^{-1}\left(\E|X|^p\right)^{1/p}$.
\end{lem} \noindent We also require the following simple technical result.

\begin{prop}\label{tech_lemma} Fix $z,w \in \R^n \setminus \{0\}$ and $0 < \epsilon < 1$. Let $T$ be a subspace of $\R^n$. If 
\begin{align}\label{tech_hypothesis}
\left|\langle  A_z^\top A_w x,  x\rangle - \langle \Phi_{z,w} x, x\rangle\right| \leqslant \epsilon \|x\|^2\ \forall\ x \in T 
\end{align} then 
\begin{align*}
\left|\langle A_z^\top A_w x, y\rangle - \langle \Phi_{z,w} x, y\rangle\right| \leqslant 3\epsilon\|x\|\|y\|\ \forall\ x,y \in T.
\end{align*}
\end{prop}

\noindent With these two results, we are now equipped to prove Proposition \ref{RCPproposition}.

\begin{proof}[Proof of Proposition \ref{RCPproposition}]

 Without loss of generality, it suffices to show concentration over $T \cap \mathcal{S}^{n-1}$. For notational simplicity, set $\Sigma_{z,w} : = A_z^\top A_w - \Phi_{z,w}$.

\paragraph{Step 1: Approximation.} We first show that if concentration over an $\epsilon$-net of $T \cap \mathcal{S}^{n-1}$ holds, then a continuity argument establishes concentration over all points in $T \cap \mathcal{S}^{n-1}$.  Choose an $\frac{\epsilon}{14}$-net $Q_T \subset T \cap \mathcal{S}^{n-1}$ such that $|Q_T| \leqslant (42/\epsilon)^k$ and for any $x \in T \cap \mathcal{S}^{n-1}$, \begin{align}
\min_{q \in Q_T}\|x - q\| \leqslant \frac{\epsilon}{14}. \label{min_bound}
\end{align}  We will prove that \begin{align}
    |\langle \Sigma_{z,w}q,q\rangle| \leqslant \frac{\epsilon}{8}\ \forall\ q \in Q_T \Longrightarrow |\langle \Sigma_{z,w}x,x\rangle| \leqslant \epsilon\ \forall\ x \in T. \label{conc_over_Q_implies_T}
\end{align} 

Now, define \begin{align}
\al^* : = \inf\left\{\al > 0 : |\langle \Sigma_{z,w} x,x\rangle| \leqslant \al\|x\|^2\ \forall\ x \in T\right\}. \label{alpha_star_def}
\end{align} We want to show that $\al^* \leqslant \epsilon$. Fix $x \in T \cap \mathcal{S}^{n-1}$. Then there exists a $q \in Q_T$ such that $\|x - q \| \leqslant \epsilon/14.$ In addition, observe that $x - q \in T$ since $q \in Q_T \subset T$ so by \eqref{alpha_star_def}, \begin{align} \label{xminqbound}
|\langle \Sigma_{z,w}(x-q),x-q\rangle| \leqslant \al^*\|x-q\|^2 \leqslant \al^*\frac{\epsilon^2}{196}.
\end{align} Now, note that by the definition of $\al^*$, \begin{align*}
|\langle \Sigma_{z,w} x, x \rangle| \leqslant \al^*\|x\|^2\ \forall\ x \in T.
\end{align*} Thus Proposition \ref{tech_lemma} gives \begin{align*}
|\langle \Sigma_{z,w} x, y \rangle| \leqslant 3\al^*\|x\|\|y\|\ \forall\ x,y \in T.
\end{align*} Applying this result to $x-q$ and $q$ gives \begin{align} \label{xminqandqbound}
|\langle \Sigma_{z,w} (x-q), q \rangle| \leqslant 3\al^*\|x- q\| \leqslant \al^*\frac{3\epsilon}{14}.
\end{align} Let $E$ be the event that $|\langle \Sigma_{z,w} q,q \rangle| \leqslant \frac{\epsilon}{8}$ for any $q \in Q_T$. Using $\langle \Sigma_{z,w} x,x \rangle = \langle \Sigma_{z,w}(x-q),x-q \rangle + 2\langle \Sigma_{z,w} x,q \rangle - \langle \Sigma_{z,w} q, q \rangle$ and $\langle \Sigma_{z,w} x,q \rangle = \langle \Sigma_{z,w} (x-q),q\rangle + \langle \Sigma_{z,w} q,q \rangle$, we have that on $E$, \begin{align*}
|\langle \Sigma_{z,w} x,x \rangle| & \leqslant |\langle \Sigma_{z,w}(x-q),x-q\rangle| + 2|\langle \Sigma_{z,w} x,q\rangle| + |\langle \Sigma_{z,w} q,q\rangle| \\
& \leqslant |\langle \Sigma_{z,w}(x-q),x-q\rangle| + 2|\langle \Sigma_{z,w} (x-q),q\rangle| + 3|\langle \Sigma_{z,w}q,q\rangle| \\
& \leqslant \al^* \frac{\epsilon^2}{196} + \al^*\frac{3\epsilon}{7} + \frac{3\epsilon}{8} \\
& = \al^*\left(\frac{\epsilon^2}{196} + \frac{3\epsilon}{7}\right) + \frac{3\epsilon}{8}
\end{align*} where we used \eqref{xminqbound}, \eqref{xminqandqbound}, and the event $E$ in the third inequality. Thus \begin{align}
|\langle \Sigma_{z,w} x,x \rangle| \leqslant \al^*\left(\frac{\epsilon^2}{196} + \frac{3\epsilon}{7}\right) + \frac{3\epsilon}{8}\ \forall\ x \in T \cap \mathcal{S}^{n-1}. \label{alpha_star_bound}
\end{align} However, recall that $\al^*$ was defined to be the smallest number such that \begin{align*}
|\langle \Sigma_{z,w} x,x \rangle| \leqslant \al^*\ \forall\ x \in T \cap \mathcal{S}^{n-1}.
\end{align*} Hence $\al^*$ must be smaller than the right hand side of \eqref{alpha_star_bound}, i.e. \begin{align*}
\al^* \leqslant \al^*\left(\frac{\epsilon^2}{196} + \frac{3\epsilon}{7}\right) + \frac{3\epsilon}{8} \Longrightarrow \al^* \leqslant \frac{3\epsilon}{8}\left(\frac{1}{1 - \frac{\epsilon^2}{196} - \frac{3\epsilon}{7}}\right) \leqslant \epsilon
\end{align*} since $0 < \epsilon < 1$. This establishes \eqref{conc_over_Q_implies_T}.

\paragraph{Step 2: Concentration.} We now establish concentration for a fixed point $x \in \Seps^{n-1}$. Then observe that \begin{align*}
\left|\langle \Sigma_{z,w} x,x \rangle \right| = \frac{1}{m} \left|\sum_{i=1}^m Y_i \right|
\end{align*} where $Y_i = X_i - \E[X_i]$, $X_i = \sign(\langle \tilde{a}_i,z \rangle \langle \tilde{a}_i, w\rangle) \langle \tilde{a}_i , x\rangle^2$, and  each $\tilde{a}_i \sim \mathcal{N}(0,I_n)$. Hence $Y_i$ are independent, centered, subexponential random variables. We now estimate their subexponential norm prior to invoking Lemma \ref{subexpbound}.

By Remark 5.18 in \cite{Vershynin_notes}, the subexponential norm satisfies \begin{align}
    \|Y_i\|_{\psi_1} = \|X_i - \E[X_i]\|_{\psi_1} \leqslant 2\|X_i\|_{\psi_{1}}. \label{sub_exp_mean_bound}
\end{align} Let $Z_i := \langle \tilde{a}_i,x\rangle \sim \mathcal{N}(0,1)$.  Then $\|Z_i\|_{\psi_2} \leqslant K_1$ for some absolute constant $K_1$ where $\|\cdot\|_{\psi_2}$ is the sub-gaussian norm. Observe that $\E|X_i|^p \leqslant \E|Z_i^2|^p$ which implies $\|X_i\|_{\psi_1} \leqslant \|Z_i^2\|_{\psi_1}$. Thus we have \begin{align*}
    \|Y_i\|_{\psi_1} \leqslant 2\|X_i\|_{\psi_1} \leqslant 2\|Z_i^2\|_{\psi_1} \leqslant 4\|Z_i\|^2_{\psi_2} \leqslant 4K_1^2
\end{align*} where we used equation \eqref{sub_exp_mean_bound} in the first inequality and Lemma 5.14 in \cite{Vershynin_notes} in the second to last inequality. Thus $K =\max_{i \in [m]} \|Y_i\|_{\psi_1} \leqslant 4K_1^2$ for an absolute constant $K_1$. Defining $K_2 := 4K_1^2$, Lemma \ref{subexpbound} guarantees that for any fixed $z,w,x \in \R^n \setminus \{0\}$ and $\epsilon > 0$, \begin{align}
\Pro\left(|\langle \Sigma_{z,w} x,x \rangle| \geqslant \epsilon\right) \leqslant 2\exp(-c_0(\epsilon) m) \label{conc_bound}
\end{align} where $c_0(\epsilon) = c \min(\epsilon^2/K_2^2,\epsilon/K_2)$ and $c > 0$ is an absolute constant.

\paragraph{Step 3: Union Bound.} We now show concentration over $Q_T$ holds. Recall that $|Q_T| \leqslant (42/\epsilon)^{k}$ so we can apply a union bound to \eqref{conc_bound} to attain \begin{align} \label{qbound}
\Pro\left(|\langle \Sigma_{z,w} q,q \rangle| \geqslant \frac{\epsilon}{8}\ \forall\ q \in Q_T\right) \leqslant 2 \left(\frac{42}{\epsilon}\right)^k\exp\left(-c_0\left(\frac{\epsilon}{8}\right)m\right).
\end{align} 

\paragraph{} By equations \eqref{conc_over_Q_implies_T} and \eqref{qbound}, we conclude that \begin{align*}
\Pro\left(|\langle \Sigma_{z,w} x,x \rangle| \geqslant \epsilon \|x\|^2\ \forall\ x \in T\right) \leqslant 2 \left(\frac{42}{\epsilon}\right)^k\exp\left(-c_0\left(\frac{\epsilon}{8}\right)m\right).
\end{align*}  The probability bound in the proposition can be shown by noting that \begin{align*}
    1 - 2(42/\epsilon)^k\exp(-c_0(\epsilon/8)m) & = 1 - 2\exp\left(-c_0(\epsilon/8) m + k\log\left(\frac{42}{\epsilon}\right)\right).
\end{align*} Thus if \begin{align*}
    \frac{2}{c_0(\epsilon/8)}\log\left(\frac{42}{\epsilon}\right) k \leqslant Ck\leqslant m
\end{align*} where $C = \Omega(\epsilon^{-1}\log\epsilon^{-1})$, we have that the result holds with probability exceeding \begin{align*}
    1-2\exp\left(-c_0(\epsilon/8) m + k\log\left(\frac{42}{\epsilon}\right)\right) \geqslant 1 - 2\exp(-c_1 m)
\end{align*} where $c_1 = c_0(\epsilon/8)/2$. Applying Proposition \ref{tech_lemma} to our result gives \eqref{RCP_3bound} with the same probability. The extension to the union of subspaces follows by applying \eqref{RCP_3bound} to all subspaces of the form $\text{span}(U_i,V_j)$ and using a union bound. Note that these subspaces have dimension at most $2k$, accounting for the extra factor of $2$ in the bound on $m$.

\end{proof}

\subsection{Uniform Concentration Over a Union of Subspaces}

We will now set out to prove a stronger version of Proposition \ref{RCPproposition} that holds uniformly for all $z$ and $w$ in (possibly) different $k$-dimensional subspaces:

\begin{prop}[Uniform RCP]\label{uniformrcp} 
Fix $0 < \epsilon < 1$ and $k < m$. Let $A \in \R^{m \times n}$ have i.i.d. $\mathcal{N}(0,1/m)$ entries. Let $T$, $W$, and $Z$ be fixed $k$-dimensional subspaces of $\R^n$. Then if $m \geqslant  C_{\epsilon}k$, then with probability at least $1 - \hat{\gamma} m^{4k}\exp(-\tilde{c}_{\epsilon} m),$ we have \begin{align}
\left|\langle A_z^\top A_w x,y \rangle - \langle \Phi_{z,w}x,y \rangle\right| \leqslant L\epsilon \|x\|\|y\|\ \forall\ x,y \in T,\ w \in W,\ z \in Z \label{Uniform_RCP_firstbound}
\end{align} where $\hat{\gamma}$ is a positive universal constant, $\tilde{c}_{\epsilon}$ depends on $\epsilon$ and $C_{\epsilon}$ depends polynomially on $\epsilon^{-1}$. Furthermore, let $U = \bigcup_{i=1}^{N_1} U_i$, $V = \bigcup_{j=1}^{N_2} V_j$, $W = \bigcup_{k=1}^{N_3} W_k$, and $Z = \bigcup_{\ell=1}^{N_4} Z_{\ell}$ be the union of at most $k$-dimensional subspaces of $\R^n$. Then if $m \geqslant 2C_{\epsilon}k$,
\begin{align}
\left|\langle A_z^\top A_w u,v \rangle - \langle \Phi_{z,w}u,v \rangle\right| \leqslant L\epsilon \|u\|\|v\|\ \forall\ u \in U,\ v \in V,\ w \in W,\ z \in Z \label{Uniform_RCP_secondbound}
\end{align} with probability exceeding $1 - N_1N_2N_3N_4\hat{\gamma} m^{4k}\exp(-\tilde{c}_{\epsilon} m)$. Here $L$ is a positive universal constant.
\end{prop}

Note that Proposition \ref{RCPproposition} established concentration of $\langle A_z^\top A_w x,y \rangle$ around $\langle \Phi_{z,w}x,y \rangle$ for $x$ and $y$ in a fixed $k$-dimensional subspace for \textit{fixed} $z,w \in \R^n \setminus\{0\}$. We are interested in showing that this concentration holds uniformly for all $z$ and $w$ in the range of our generative model. The proof of this result uses an interesting fact from $1$-bit compressed sensing which establishes that if a sufficient number of random hyperplanes cut the unit sphere, the diameter of each tesselation piece is small with high probability \citep{Vershynin_hyperplane_thm}. We state the theorem here for convenience:

\begin{thm}[Theorem 2.1 in \cite{Vershynin_hyperplane_thm}]\label{randhyperplanethm}
Let $n,m,s > 0$ and set $\delta = C_1 \left(\frac{s}{m} \log(2n/s)\right)^{1/5}$. Let $a_i \in \R^n$ have i.i.d. $\mathcal{N}(0,1)$ entries for $i \in [m]$. Then with probability at least $1 - C_2 \exp(- c \delta m)$, the following holds uniformly for all $x,\tilde{x} \in \R^n$ that satisfy $\|x\|_2 = \|\tilde{x}\|_2 = 1$, $\|x\|_1 \leqslant \sqrt{s}$, and $\|\tilde{x}\|_1 \leqslant \sqrt{s}$ for $s \leqslant n$: \begin{align}
\langle a_i, \tilde{x} \rangle \langle a_i, x \rangle \geqslant 0,\ \forall\ i \in [m] \Longrightarrow \|\tilde{x} - x\|_2 \leqslant \delta.
\end{align} Here $C_1, C_2, c$ are positive universal constants.
\end{thm}

We will use this result to prove the following: given a sufficient number of random hyperplanes and a $k$-dimensional subspace $Z$, there exists a finite set of points $Z_0 \subset Z$ that live in the interior of the tesselation pieces generated by the random hyperplanes such that any point in $Z$ can be closely approximated by a point in $Z_0$ with high probability.

\begin{lem} \label{Z_0lemma} Fix $0 < \epsilon < 1$. Let $A \in \R^{m \times n}$ have i.i.d. $\mathcal{N}(0,1/m)$ entries with rows $\{a_{\ell}\}_{\ell=1}^m$. Let $Z$ be a $k$-dimensional subspace of $\R^n$. Define $E_{Z,A}$ to be the event that there exists a set $Z_0 \subset Z$ with the following properties: \begin{itemize}
    \item each $z_0 \in Z_0$ satisfies $\langle a_{\ell}, z_0 \rangle  \neq 0$ for all $\ell \in [m]$,
    \item $|Z_0| \leqslant 10m^{2k}$, and 
    \item for all $z \in Z$ such that $\|z\|=1$, there exists a $z_0 \in Z_0$ such that $\|z-z_0\|\leqslant \epsilon.$
\end{itemize} If $m \geqslant \hat{C} k$, then $\Pro(E_{Z,A}) \geqslant 1 - C_2\exp(-c\epsilon m)$. %Let $Z \subset \R^n$ be a $k$-dimensional subspace. Define the following subset of $Z$: \begin{align} \label{Z_0def}
   % Z_0 := \left\{z \in Z : \|z\| = 1\ \text{and}\ a_{\ell}^\top z \neq 0\ \forall\ \ell \in [m]\right\}.
%\end{align} Define the random variable $I_{Z,A} := |Z_0|$. Then if $m \geqslant c_{\epsilon} k$, the following event holds with probability exceeding $1 - C_2 \exp(-c \epsilon m)$: \begin{align}
    %E_{Z,A} : = \left\{I_{Z,A}\leqslant 10m^{2k}\ \text{and}\ \forall\ z \in Z\ \text{s.t.}\ \|z\| = 1,\ \exists\ z_0 \in Z_0\ 
%\text{s.t.}\ \|z - z_0\| \leqslant \epsilon\right\}. \label{E_Zdef}
%\end{align} 
  Here $C_2$ and $c$ are positive absolute constants and $\hat{C}$ depends polynomially on $\epsilon^{-1}.$
\end{lem}

\begin{proof}[Proof of Lemma \ref{Z_0lemma}]
By the rotational invariance of the Gaussian distribution, we may take $Z$ to be in the span of the first $k$ standard basis vectors. We may further without loss of generality assume $A \in \R^{m \times k}$. We will invoke the following lemma which establishes that the unit sphere of $Z$ is partitioned into at most $10m^{2k}$ regions by the rows $\{a_{\ell}\}_{\ell=1}^m$ of $A$ with probability $1$:

\begin{lem}\label{boundoncardAz}
Let $V$ be a subspace of $\R^n$. Let $A \in \R^{m \times n}$ have i.i.d. $\mathcal{N}(0,1/m)$ entries. With probability 1, \begin{align*}
|\{\diag(\sign(Av))A : v \in V\}| \leqslant 10m^{2\dimension V}.
\end{align*}
\end{lem}

In each tesselation piece defined by the rows of $A$, choose a single point $z_0$ from $Z$ with unit norm such that $a_{\ell}^{\top}z_0 \neq 0$ for all $\ell \in [m]$ (if such a point exists in the tesselation piece). Let $Z_0$ denote this collection of points and set $I_{Z,A} := |Z_0|$. By Lemma \ref{boundoncardAz} with $V = Z$, the cardinality of $Z_0$ is bounded with probability $1$: $I_{Z,A} \leqslant 10m^{2k}$. Then observe that we can set the parameters $n$ and $s$ in Theorem \ref{randhyperplanethm} equal to $k$ since $A \in \R^{m \times k}$ and $Z$ is in the span of the first $k$ standard basis vectors. Then if $m \geqslant \left(C_1^5\log(2)/\epsilon^5\right)k =: \hat{C} k$, we have that the quantity $\delta$ in the theorem is bounded by $\epsilon$: \begin{align*}
    \delta := C_1 \left(\frac{k}{m}\log(2)\right)^{1/5} \leqslant  \epsilon
\end{align*} so $\Pro(E_{Z,A}) \geqslant 1 - C_2 \exp(-c\epsilon m)$ for some positive universal constants $c$, $C_1$, and $C_2$ and $\hat{C}$ depends polynomially on $\epsilon^{-1}$.
\end{proof}

We now proceed with the proof of the Uniform RCP.

\begin{proof}[Proof of Proposition \ref{uniformrcp}]
Let $E_{Z,A}$ be the event defined in Lemma \ref{Z_0lemma}. By Lemma \ref{Z_0lemma}, we have that if $m \geqslant \hat{C} k$, there exists an event $E_{Z,A}$ with $\Pro(E_{Z,A}) \geqslant 1 - C_2 \exp(-c \epsilon m)$ on which there exists a finite subset $Z_0$ of $Z$ with cardinality $I_{Z,A} \leqslant 10m^{2k}$ such that for any $z \in Z$ with $\|z\|=1$, there exists a $z_0 \in Z_0$ such that $\|z-z_0\|\leqslant \epsilon$. The analogous finite set $W_0$ with cardinality $I_{W,A} \leqslant 10m^{2k}$ also exists on the event $E_{W,A}$ with probability at least $1 - C_2\exp(-c\epsilon m)$. Thus if $m \geqslant \hat{C} k$, the event $E_{Z,W} := E_{Z,A} \cap E_{W,A}$ satisfies \begin{align*}
    \Pro(E_{Z,W}) \geqslant 1 - 2 C_2 \exp(-c\epsilon m).
\end{align*} 

We now establish concentration over $Z_0$ and $W_0$. Let $E_0$ be the event that \begin{align*}
\left|\langle A_{z_0}^\top A_{w_0}  x,y \rangle - \langle \Phi_{z_0,w_0}x,y\rangle\right| & \leqslant 3\epsilon \|x\|\|y\|\ \forall\ x,y \in T,\ z_0 \in Z_0,\ w_0 \in W_0. 
\end{align*} By Proposition \ref{RCPproposition}, if $m \geqslant C k$, we have that the following holds for fixed $z_0 \in Z_0$ and $w_0 \in W_0$ with probability exceeding $1 - 2\exp(-c_1 m)$: \begin{align*}
    \left|\langle A_{z_0}^\top A_{w_0}x,y \rangle - \langle \Phi_{z_0,w_0}x,y\rangle\right| \leqslant 3\epsilon\|x\|\|y\|\ \forall\ x,y \in T.
\end{align*} Furthermore, on $E_{Z,W}$, a union bound over all $z_0 \in Z_0$ and $w_0 \in W_0$ shows that $$\Pro(E_0) \geqslant 1 - 2I_{Z,A}I_{W,A}\exp\left(-\frac{c_1 m}{2}\right) \geqslant 1 - \gamma m^{4k}\exp\left(-\frac{c_1m}{2}\right)$$ where $\gamma$ is a positive absolute constant and $c_1$ depends on $\epsilon$.

For the remainder of this proof, we work on the event $E_0 \cap E_{Z,W}$. Fix non-zero $z \in Z$ and $w \in W$. Define the following set: \begin{align*}
    \Omega_{z,w} : = \left\{\ell \in [m] : \langle a_{\ell}, z\rangle = 0\ \text{or}\ \langle a_{\ell}, w\rangle = 0\right\}.
\end{align*} Note that since $Z$ and $W$ are $k$-dimensional and any subset of $k$ rows of $A$ are linearly independent with probability $1$, at most $k$ entries of $Az$ are zero and similarly for $Aw$. Hence $|\Omega_{z,w}| \leqslant 2k$. Furthermore, observe that \begin{align*}
 A_{z}^\top A_w & = \sum_{\ell =1}^m\sgn(\langle a_{\ell},z\rangle \langle a_{\ell}, w\rangle) a_{\ell}a_{\ell}^\top \\
 & = \sum_{\ell \in \Omega_{z,w}}\sgn(\langle a_{\ell},z\rangle \langle a_{\ell}, w\rangle) a_{\ell}a_{\ell}^\top + \sum_{\ell \in \Omega_{z,w}^c}\sgn(\langle a_{\ell},z\rangle \langle a_{\ell}, w\rangle) a_{\ell}a_{\ell}^\top \\
 & = \sum_{\ell \in \Omega_{z,w}^c}\sgn(\langle a_{\ell},z\rangle \langle a_{\ell}, w\rangle) a_{\ell}a_{\ell}^\top
\end{align*} by the definition of $\Omega_{z,w}$. However, on the event $E_{Z,W}$, there exists a $z_0 \in Z_0$ and $w_0 \in W_0$ such that for all $\ell \in \Omega_{z,w}^c$, \begin{align*}
    \sgn(\langle a_{\ell},z\rangle) = \sgn(\langle a_{\ell},z_0\rangle)\ \text{and}\ \sgn(\langle a_{\ell},w\rangle) =  \sgn(\langle a_{\ell},w_0\rangle)
\end{align*} i.e. $z$ and $z_0$ (likewise $w$ and $w_0$) lie on the same side of each hyperplane defined by $\{a_{\ell}\}_{\ell=1}^m$. Hence we have \begin{align}
    A_{z}^\top A_w = \sum_{\ell \in \Omega_{z,w}^c}\sgn(\langle a_{\ell},z\rangle \langle a_{\ell}, w\rangle) a_{\ell}a_{\ell}^\top & = \sum_{\ell \in \Omega_{z,w}^c}\sgn(\langle a_{\ell},z_0\rangle \langle a_{\ell}, w_0\rangle) a_{\ell}a_{\ell}^\top \nonumber \\
    & = A_{z_0}^{\top}A_{w_0} - \sum_{\ell \in \Omega_{z,w}}\sgn(\langle a_{\ell},z_0\rangle \langle a_{\ell}, w_0\rangle) a_{\ell}a_{\ell}^\top \nonumber \\
    & =: A_{z_0}^{\top}A_{w_0} - \tilde{A}_{z_0}^{\top}\tilde{A}_{w_0}. \label{AzAw_as_Az0Aw0}
\end{align} We now use the following lemma which says that  $\tilde{A}_{z_0}^{\top}\tilde{A}_{w_0}$ is small when acting on $T$.

\begin{lem}\label{Atilde_small_lemma}
Fix $0 < \epsilon < 1$ and $k < m$. Suppose that $A \in \R^{m \times n}$ has i.i.d. $\mathcal{N}(0,1/m)$ entries. Let $T \subset \R^n$ be a $k$-dimensional subspace and $W_0$ and $Z_0$ be subsets of $\R^n$. Let $E$ be the event the following inequality holds for all $\Omega \subset [m]$ satisfying $|\Omega| \leqslant 2k$: \begin{align}
    |\langle\tilde{A}_{z_0}^\top \tilde{A}_{w_0} x,y \rangle | \leqslant \epsilon\|x\|\|y\|\ \forall\ x,y \in T,\ w_0 \in W_0,\ z_0 \in Z_0
    \label{Atilde_small_bound}
\end{align} where \begin{align*}
    \tilde{A}_{z_0}^\top \tilde{A}_{w_0} : = \sum_{\ell \in \Omega}\sgn(\langle a_{\ell},z_0\rangle \langle a_{\ell}, w_0\rangle) a_{\ell}a_{\ell}^\top.
\end{align*} Then there exists a $\delta_{\epsilon} > 0$ such that if $m \geqslant 9\epsilon^{-1}k$ and $2k \leqslant \delta_{\epsilon}m$, $\Pro(E) \geqslant 1 - 2m\exp(-\epsilon m/36)$.
\end{lem}

 Let $E$ be the event defined in Lemma \ref{Atilde_small_lemma}. On the event $E \cap E_0 \cap E_{Z,W}$, we have that for all $z \in Z \cap \Seps^{n-1}$ and $w \in W \cap \Seps^{n-1}$, there exists a $z_0 \in Z_0$ and $w_0 \in W_0$ such that for any $x, y \in T$, \begin{align*}
\left|\langle A_z^\top A_w x,y \rangle - \langle \Phi_{z,w}x,y \rangle\right| & = \left|\langle A_{z_0}^{\top}A_{w_0} x,y \rangle - \langle\tilde{A}_{z_0}^{\top}\tilde{A}_{w_0}x,y\rangle - \langle \Phi_{z,w}x,y \rangle\right| \\
& \leqslant \left|\langle A_{z_0}^{\top}A_{w_0} x,y \rangle - \langle \Phi_{z,w}x,y \rangle\right| + | \langle\tilde{A}_{z_0}^{\top}\tilde{A}_{w_0}x,y\rangle| \\
& \leqslant \left|\langle A_{z_0}^{\top}A_{w_0} x,y \rangle - \langle \Phi_{z_0,w_0}x,y \rangle\right| + \left|\langle \Phi_{z_0,w_0}x,y \rangle - \langle \Phi_{z,w}x,y \rangle\right| + | \langle\tilde{A}_{z_0}^{\top}\tilde{A}_{w_0}x,y\rangle| \\
& \leqslant 3\epsilon \|x\|\|y\| + \frac{88}{\pi}\epsilon \|x\|\|y\| + \epsilon\|x\|\|y\|\\
& =: L \epsilon\|x\|\|y\|
\end{align*} where we define $L:= 3 + \frac{88}{\pi} + 1 < 33$. In the first equality, we used the event $E_{Z,W}$ and \eqref{AzAw_as_Az0Aw0}. In the last inequality, we used the continuity of $\Phi_{z,w}$ from Lemma \ref{continuityphi} along with the event $E_0 \cap E$. Letting $C_{\epsilon}:= 9 \epsilon^{-1}\hat{C}$ where $\hat{C}$ is given by Lemma \ref{Z_0lemma}, we have that if $m \geqslant C_{\epsilon}k$, the event $E \cap E_0 \cap E_{Z,W}$ holds with probability exceeding \begin{align*}
    \Pro\left(E \cap E_0 \cap E_{Z,W}\right) & \geqslant 1 - 2m\exp(-\epsilon m/36) - \gamma m^{4k}\exp\left(-\frac{c_1m}{2}\right) - 2C_2 \exp\left(-c\epsilon m\right) \\
    & \geqslant 1 - \hat{\gamma} m^{4k}\exp\left(-\tilde{c}_{\epsilon}m\right) 
\end{align*} where $\hat{\gamma}$ is a positive absolute constant and $\tilde{c}_{\epsilon}$ depends polynomially on $\epsilon$. The extension to the union of subspaces follows by applying \eqref{Uniform_RCP_firstbound} to all combinations of subspaces $T_{ij} = \text{span}(U_i,V_j)$, $W_k$, and $Z_{\ell}$ where each $T_{ij}$ have dimension at most $2k$ and using a union bound.
\end{proof}

\subsection{Application to Range of Generative Model}

We now apply Proposition \ref{uniformrcp} to prove Proposition \ref{realRRCP}:

\begin{proof}[Proof of Proposition \ref{realRRCP}]

For pedagogical purposes, we first establish the lemma in the $d=2$ case. In order to apply Proposition \ref{uniformrcp}, we will show that $\{\G(x) - \G(y) : x,y \in \R^k\}$ is a subset of the union of at most $10^6(n_1^2n_2)^{2k}$ subspaces of dimensionality at most $2k$.

For fixed $W_1,W_2$, let $\mathcal{A}_{+,1} = \{W_{1,+,x} : x \neq 0\}$ and $\mathcal{B}_{+,2} = \{W_{2,+,x} : x \neq 0\}$.  
By Lemma 15 in \cite{HV2017}, there exists a probability 1 event, $E$, over $(W_1,W_2)$ on which $|\mathcal{A}_{+,1}| \leqslant 10 n_1^k$ and $|\mathcal{B}_{+,2}| \leqslant 10^2 n_1^k n_2^k$. 
On $E$, 
\[
|\{ W_{2,+,x}W_{1,+,x} : x \neq 0\}| \leqslant 10^3 (n_1^2 n_2)^k.
\]
 Note that $\dim \text{range}(W_{2,+,x}W_{1,+,x}) \leqslant k$ for all $x \neq 0$.  Hence $$\{\G(x) : x \in \R^k\} \subset \{ W_{2,+,x}W_{1,+,x} w : x, w \in \Seps^{k-1} \} \subset V$$ where $V$ the union of at most $10^3 (n_1^2 n_2)^k$ subspaces of dimensionality at most $k$. This implies that \begin{align*}
    \{\G(x)&-\G(y) : x,y \in \R^k\}
    \subset V'
\end{align*} where $V'$ is the union of at most $10^6(n_1^2n_2)^{2k}$ subspaces of dimensionality at most $2k$.

By applying the second half of Proposition \ref{uniformrcp} to the sets $V'$, $V'$, $V$, and $V$, we get that for fixed $W_1$, $W_2$, \begin{align}
    |\langle (A_{\G(x)}^\top A_{\G(y)} & - \Phi_{\G(x),\G(y)})(\G(x_1) - \G(x_2)), \G(x_3)-\G(x_4)\rangle | \nonumber \\
    & \leqslant L\epsilon\|\G(x_1) - \G(x_2)\|\|\G(x_3)-\G(x_4)\| \label{2_layer_concentration_inequality}
\end{align}
with probability at least $$1 - 10^{3(2) + 6(2)}(n_1^2n_2)^{2k + 4k}\hat{\gamma}m^{4k} e^{-\tilde{c}_{\epsilon} m} \geqslant 1 - \tilde{\gamma}m^{4k}e^{-\tilde{c}_{\epsilon}m/2 },$$ provided $m \geqslant \hat{K}C_{\epsilon}\tilde{c}_{\epsilon}^{-1}k\log(n_1n_2) =: \tilde{C}_{\epsilon}k \log(n_1n_2)$, where $\tilde{\gamma}$ and $\hat{K}$ are positive universal constants, $\tilde{c}_{\epsilon}$ depends polynomially on $\epsilon$, and $C_{\epsilon}$ depends polynomially on $\epsilon^{-1}$.

Integrating over the probability space of $(W_1,W_2)$, independence of $A$ and $(W_1,W_2)$ implies that  \eqref{2_layer_concentration_inequality} holds for random $(W_1,W_2)$ with the same probability bound.
Continuing from \eqref{2_layer_concentration_inequality}, we have \begin{align}
    |\langle (A_{\G(x)}^\top A_{\G(y)} & - \Phi_{\G(x),\G(y)})(\G(x_1) - \G(x_2)), \G(x_3)-\G(x_4)\rangle | \nonumber \\
    & \leqslant L\epsilon\|\G(x_1) - \G(x_2)\|\|\G(x_3)-\G(x_4)\| \nonumber
\end{align} $\forall x,y ,x_1,x_2,x_3,x_4 \in \R^k$ with probability at least $ 1 - \tilde{\gamma}m^{4k}e^{-\tilde{c}_{\epsilon}m/2 }$ for some positive absolute constant $\tilde{\gamma}$ and $\tilde{c}_{\epsilon}$ depends polynomially on $\epsilon$. 

The case for $d \geqslant 2$ follows similarly. We have 
\[
|\{ \PiWdix : x \neq 0\}| \leqslant 10^{(d^2)} (n_1^d n_2^{d-1} \cdots n_{d-1}^2 n_d)^k
\] on the probability 1 event. This implies that $\{\G(x) : x \in \R^k\} \subset \{\Pi_{i=d}^1W_{i,+,x} w : x,w \in \Seps^{k-1}\}$ is a subset of the union of at most $10^{(d^2)} (n_1^d n_2^{d-1} \cdots n_{d-1}^2 n_d)^k$ subspaces of dimensionality at most $k$. Moreover, $\{\G(x)-\G(y) : x,y \in \R^k\}$ is a subset of the union of at most \[10^{(2d^2)} (n_1^d n_2^{d-1} \cdots n_{d-1}^2 n_d)^{2k}\] subspaces of dimensionality at most $2k$. Hence the analogous bound \eqref{2_layer_concentration_inequality} holds for all $x,y,x_1,x_2,x_3,x_4 \in \R^k$ with probability at least $$1 - 10^{(2d^2 + 4d^2)} (n_1^d n_2^{d-1} \cdots n_{d-1}^2 n_d)^{2k+4k}\hat{\gamma}m^{4k} e^{-\tilde{c}_{\epsilon} m} \geqslant 1 - \tilde{\gamma}m^{4k}e^{-\tilde{c}_{\epsilon} m/2},$$ provided $m \geqslant \tilde{C}_{\epsilon} d k \log(n_1 n_2 \cdots n_d)$, where $\tilde{\gamma}$ is a positive absolute constant, $\tilde{c}_{\epsilon}$ depends polynomially on $\epsilon$, and $\tilde{C}_{\epsilon}$ depends polynomially on $\epsilon^{-1}$. 
    
\end{proof}

\subsection{RRCP Supplementary Results}

\begin{proof}[Proof of Proposition \ref{tech_lemma}]
Fix $0 < \epsilon < 1$. Suppose \eqref{tech_hypothesis} holds and fix $x,y \in T$. Without loss of generality, assume $x$ and $y$ are unit normed. We will use the shorthand notation $\Phi = \Phi_{z,w}$. Since $T$ is a subspace, $x - y \in T$ so by \eqref{tech_hypothesis}, \begin{align*}
\left|\langle  A_z^\top A_w (x - y), x-y\rangle - \langle \Phi (x-y), x-y\rangle\right| \leqslant \epsilon\|x-y\|^2
\end{align*} or equivalently \begin{align}
\langle \Phi (x-y), x-y\rangle - \epsilon \|x-y\|^2 \leqslant \langle  A_z^\top A_w(x - y), x-y\rangle  \leqslant \langle \Phi (x-y), x-y\rangle + \epsilon \|x-y\|^2. \label{expand_out}
\end{align} Note that \begin{align*}
\|x-y\|^2 = 2 - 2\langle x,y\rangle,
\end{align*} \begin{align*}
\langle \Phi (x-y), x-y\rangle = \langle \Phi x,x\rangle + \langle \Phi y,y\rangle - 2\langle \Phi x,y\rangle,
\end{align*} and \begin{align*}
\langle  A_z^\top A_w (x - y), x-y\rangle = \langle A_z^\top A_w x,  x\rangle + \langle A_z^\top A_w y, y\rangle - 2 \langle A_z^\top A_w x, y\rangle
\end{align*} where we used the fact that $\Phi$ and $A_z^\top A_w$ are symmetric. Rearranging \eqref{expand_out} yields \begin{align*}
2 \left(\langle \Phi x,y \rangle - \langle A_z^\top A_w x, y \rangle\right) \leqslant \left(\langle \Phi x, x\rangle - \langle  A_z^\top A_w x ,  x\rangle\right) + \left(\langle \Phi  y, y\rangle - \langle  A_z^\top A_w y,  y\rangle\right) + (2 - 2\langle x,y \rangle)\epsilon.
\end{align*} By assumption, the first two terms are bounded from above by $\epsilon$. Thus \begin{align*}
2 \left(\langle \Phi x,y \rangle - \langle A_z^\top A_w x,  y \rangle\right) & \leqslant 2 \epsilon + (2 - 2\langle x,y\rangle)\epsilon  = 2(2 - \langle x,y \rangle)\epsilon \leqslant 6 \epsilon
\end{align*} so \begin{align*}
\langle \Phi x,y \rangle - \langle A_z^\top A_w x, y \rangle \leqslant 3 \epsilon.
\end{align*} The lower bound is identical and establishes the desired result.
\end{proof}

\begin{proof}[Proof of Lemma \ref{boundoncardAz}]
It suffices to prove the same upperbound for $|\{\sgn(Av) : v \in V\}|.$ Let $\ell = \dim V$. By rotational invariance of Gaussians, we may take $V = \text{span}(e_1,\dots,e_{\ell})$ without loss of generality. Without loss of generality, we may let $A$ have dimensions $m \times 
\ell$ and take $V = \R^{\ell}$.

We will appeal to a classical result from sphere covering \cite{sphere_covering}. If $m$ hyperplanes in $\R^{\ell}$ contain the origin and are such that the normal vectors to any subset of $\ell$ of those hyperplanes are independent, then the complement of the union of these hyperplanes is partitioned into at most \begin{align*}
2 \sum_{i = 0}^{\ell - 1} \binom{m-1}{i}
\end{align*} disjoint regions. Each region uniquely corresponds to a constant value of $\sgn(Av)$ that has all non-zero entries. With probability $1$, any subset of $\ell$ rows of $A$ are linearly independent, and thus, \begin{align*}
|\{\sgn(Av) : v \in \R^{\ell},\ (Av)_i \neq 0\ \forall\ i\}| \leqslant 2 \sum_{i=0}^{\ell - 1} \binom{m-1}{i} \leqslant 2\ell \left(\frac{e m}{\ell}\right)^{\ell} \leqslant 10m^{\ell}
\end{align*} where the first inequality uses the fact that $\binom{m}{\ell} \leqslant (em/\ell)^{\ell}$ and the second inequality uses that $2\ell (e/\ell)^{\ell} \leqslant 10$ for all $\ell \geqslant 1$.

For arbitrary $v$, at most $\ell$ entries of $Av$ can be zero by linear independence of the rows of $A$. At each $v$, there exists a direction $\tilde{v}$ such that $(A(v + \delta \tilde{v}))_{i} \neq 0$ for all $i$ and for all $\delta$ sufficiently small. Hence, $\sgn(Av)$ differs from one of $\{\sgn(Av) : v \in \R^{\ell},\ (Av)_{i} \neq 0\ \forall\ i\}$ by at most $\ell$ entries. Thus, \begin{align*}
|\{\sgn(Av) : v \in \R^{\ell}\}| \leqslant \binom{m}{\ell} |\{\sgn(Av) : v \in \R^{\ell},\ (Av)_i \neq 0\ \forall\ i\}| \leqslant m^{\ell} 10 m^{\ell} = 10m^{2\ell}.
\end{align*}
\end{proof}

\begin{proof}[Proof of Lemma \ref{Atilde_small_lemma}] For any $\Omega \subset [m]$, let $A_{\Omega}$ denote the submatrix of $A$ with rows $a_{\ell}^{\top}$ where $\ell \in \Omega$. We claim that it suffices to show \begin{align}
    \|A_{\Omega}x\| \leqslant \sqrt{\epsilon} \|x\|\ \forall\ x \in T\ \forall\ \Omega \subset [m]\ \text{satisfying}\ |\Omega| \leqslant 2k \leqslant \delta_{\epsilon} m. \label{A_omega_bound}
\end{align} To see this, observe that for any $w_0 \in W_0$, $z_0 \in Z_0$, and $x,y \in T$ and $\Omega \subset [m]$, we have that \begin{align*}
    |\langle\tilde{A}_{z_0}^\top \tilde{A}_{w_0} x,y \rangle | 
    & = \left|\left\langle \diag(\sgn(A_{\Omega}z_0) \odot \sgn(A_{\Omega}w_0))A_{\Omega}x,A_{\Omega}y\right\rangle\right| \\
    & \leqslant \|\diag(\sgn(A_{\Omega}z_0) \odot \sgn(A_{\Omega}w_0))\|\|A_{\Omega}x\|\|A_{\Omega}y\| \\
    & \leqslant \|A_{\Omega}x\|\|A_{\Omega}y\|
\end{align*} where we used the Cauchy-Schwarz inequality in the first inequality. Hence establishing \eqref{A_omega_bound} will imply the desired conclusion.

By the rotational invariance of the Gaussian distribution, we may take $T$ to be in the span of the first $k$ standard basis vectors. We may further without loss of generality assume $A \in \R^{m \times k}$ so it suffices to establish  $\|A_{\Omega}\| \leqslant \sqrt{\epsilon}$. Fix $\Omega \subset[m]$ satisfying $|\Omega| \leqslant 2k$. By Corollary 5.35 in \cite{Vershynin_notes}, we have that for any $t \geqslant 0$, it holds with probability $1 - 2\exp(-t^2/2)$ that $$\sqrt{m}\|A_{\Omega}\| \leqslant \sqrt{|\Omega|} + \sqrt{k} + t.$$ Taking $t = \sqrt{\epsilon m}/3$, we conclude that if $|\Omega| \leqslant \epsilon m/9$ and $m \geqslant 9k/\epsilon$, then $\|A_{\Omega}\| \leqslant \sqrt{\epsilon}$ with probability $1 - 2\exp(-\epsilon m/18).$

We now establish that $\|A_{\Omega}\| \leqslant \sqrt{\epsilon}$ holds simultaneously over all subsets $\Omega \subset [m]$ of a sufficiently small size with a union bound. Observe that since $\lim_{\delta \rightarrow 0} \left(\frac{e}{\delta}\right)^{\delta} = 1$, there exists a $\delta_* > 0$ such that $\left(\frac{e}{\delta_*}\right)^{\delta_{*}} \leqslant \exp(\epsilon /36)$. Put $\delta_{\epsilon} := \min\{\epsilon, \delta_{*}\}$. Let $E$ be the event that $\|A_{\Omega}\| \leqslant \sqrt{\epsilon}$ for all subsets $\Omega \subset [m]$ satisfying $|\Omega| \leqslant 2k \leqslant \delta_{\epsilon} m$. If $m \geqslant 9\epsilon^{-1}k$, a union bound shows that this event holds with probability at least \begin{align*}
    1 - 2 \sum_{\ell = 1}^{\lfloor\delta_{\epsilon} m\rfloor}\binom{m}{\ell}\exp(-\epsilon m/18) & \geqslant 1 - 2\lfloor\delta_{\epsilon} m\rfloor \binom{m}{\lfloor\delta_{\epsilon} m\rfloor}\exp(-\epsilon m/18) \\
    & \geqslant 1 - 2\lfloor\delta_{\epsilon}m\rfloor \left(\frac{em}{\delta_{\epsilon} m}\right)^{\delta_{\epsilon}m}\exp(-\epsilon m/18) \\
    &=  1 - 2  \lfloor\delta_{\epsilon}m\rfloor\left[\left(\frac{e}{\delta_{\epsilon}}\right)^{\delta_{\epsilon}}\right]^m\exp(-\epsilon m/18) \\
    & \geqslant 1 - 2\lfloor\delta_{\epsilon} m\rfloor \exp(-\epsilon m/36) \\
    & \geqslant 1 - 2m\exp(-\epsilon m/36)
\end{align*} where we used the fact that $\left(\frac{e}{\delta_{\epsilon}}\right)^{\delta_{\epsilon}} \leqslant \exp(\epsilon /36)$ in the second to last inequality.

\end{proof}

We now prove the continuity of $\Phi_{z,w}$ for non-zero $z,w \in \R^n$. Recall that \begin{align*}
    \Phi_{z,w} : = \frac{\pi - 2\theta_{z,w}}{\pi} I_{n} + \frac{2\sin \theta_{z,w}}{\pi}M_{\hat{z} \leftrightarrow \hat{w}}
\end{align*} where $\theta_{z,w} := \angle(z,w)$ and $M_{z \leftrightarrow w}$ is the matrix that sends $\hat{z} \mapsto e_1$, $\hat{w} \mapsto \cos \theta_{z,w} e_1 + \sin \theta_{z,w} e_2$, and $h \mapsto 0$ for all $h \in \text{span}(\{z,w\}^{\perp}).$

\begin{lem}[Continuity of $\Phi_{z,w}$]\label{continuityphi}
Fix $0 < \epsilon < 1$ and $z,w \in \mathcal{S}^{n-1}$. If $\|\tilde{z} - z\| \leqslant \epsilon$ and $\|\tilde{w} - w\| \leqslant \epsilon$ for some $\tilde{z},\tilde{w} \in \mathcal{S}^{n-1}$, then \begin{align*}
\|\Phi_{\tilde{z},\tilde{w}} - \Phi_{z,w}\| \leqslant \frac{88}{\pi}\epsilon.
\end{align*}
\end{lem}
\begin{proof}[Proof of Lemma \ref{continuityphi}] In this proof, we will utilize the following three inequalities: \begin{align}
    |\theta_{x_1,y} - \theta_{x_2,y}| & \leqslant |\theta_{x_1,x_2}|,\ \forall\ x_1,x_2,y \in \mathcal{S}^{n-1} \label{relativeanglebound} \\
    2\sin(\theta_{x,y}/2) & \leqslant \|x - y\|,\ \forall\ x,y \in \mathcal{S}^{n-1} \label{sinupperbound} \\
    \theta/4 & \leqslant \sin(\theta/2),\ \forall\ \theta \in [0,\pi]. \label{sinlowerbound}
\end{align} Observe that \begin{align*}
    \|\Phi_{\tilde{z},\tilde{w}} - \Phi_{z,w}\| & \leqslant \frac{2|\theta_{\tilde{z},\tilde{w}} - \theta_{z,w}|}{\pi}\|I_n\| + \left\| \frac{2\sin \theta_{\tilde{z},\tilde{w}}}{\pi}M_{\tilde{z} \leftrightarrow \tilde{w}} - \frac{2\sin \theta_{z,w}}{\pi}M_{z \leftrightarrow w}\right\|.
\end{align*} First, observe that by \eqref{relativeanglebound}, we have that  \begin{align*}
    |\theta_{\tilde{z},\tilde{w}} - \theta_{z,w}| & \leqslant |\theta_{\tilde{z},\tilde{w}} - \theta_{z,\tilde{w}}| + |\theta_{z,\tilde{w}} - \theta_{z,w}| \leqslant |\theta_{\tilde{z},z}| + |\theta_{\tilde{w},w}|.
\end{align*} Then, by \eqref{sinupperbound} and \eqref{sinlowerbound}, we have that \begin{align*}
    |\theta_{\tilde{z},z}| \leqslant 4\sin(\theta_{\tilde{z},z}/2) \leqslant 2\|\tilde{z} - z\| \leqslant 2\epsilon.
\end{align*} The same upper bound holds for $|\theta_{\tilde{w},w}|$. Thus we attain \begin{align}
    |\theta_{\tilde{z},\tilde{w}} - \theta_{z,w}| \leqslant |\theta_{\tilde{z},z}| + |\theta_{\tilde{w},w}| \leqslant 4\epsilon. \label{thetabound}
\end{align} 

Let $R$ be a rotation matrix that maps $z \mapsto e_1$ and $w \mapsto \cos \theta_{z,w} e_1 + \sin \theta_{z,w} e_2$ where $e_1$ and $e_2$ are the first and second standard basis vectors, respectively. Let $\tilde{R}$ denote the matrix that applies the same rotatation to the system $\tilde{z}$ and $\tilde{w}$. Recall that $M_{z \leftrightarrow w} := R^\top D R\ \text{and}\  M_{\tilde{z} \leftrightarrow \tilde{w}} : = \tilde{R}^\top \tilde{D} \tilde{R}$ where \begin{align*}
    D : = \left[\begin{array}{ccc}
    \cos\theta_{z,w} & \sin\theta_{z,w} & 0  \\
     \sin \theta_{z,w} & -\cos \theta_{z,w} & 0 \\
     0 & 0 & 0_{k-2}
\end{array}\right]\ \text{and}\ \tilde{D} : = \left[\begin{array}{ccc}
    \cos\theta_{\tilde{z},\tilde{w}} & \sin\theta_{\tilde{z},\tilde{w}} & 0  \\
     \sin\theta_{\tilde{z},\tilde{w}} & -\cos \theta_{\tilde{z},\tilde{w}} & 0 \\
     0 & 0 & 0_{k-2}
\end{array}\right].
\end{align*} An elementary calculation shows that $D$ has $2$ pairs of non-zero eigenvalues and eigenvectors $(\lambda_1, d_1)$ and $(\lambda_2,d_2)$ where \begin{align*}
    \lambda_1 = -1\ \text{and}\ d_1 = (\cos \theta_{z,w} - 1)e_1 + \sin \theta_{z,w}e_2
\end{align*} while \begin{align*}
    \lambda_2 = 1\ \text{and}\ d_2 = (\cos \theta_{z,w} + 1)e_1 + \sin \theta_{z,w} e_2.
\end{align*} Let $D = - d_1d_1^\top +  d_2d_2^\top$ be the eigenvalue decomposition for $D$. Then by the definition of $M_{z \leftrightarrow w}$, \begin{align*}
    M_{z \leftrightarrow w} = R^\top D R
    & = -R^\top d_1 d_1^\top R + R^\top d_2 d_2^\top R 
    =: - v_1 v_1^\top +  v_2 v_2^\top
\end{align*} so $v_1 := R^\top d_1$ and $v_2 := R^\top d_2$ are the eigenvectors of $M_{z \leftrightarrow w}$ with corresponding eigenvalues $-1$ and $1$, respectively. Then, recall that $Rz = e_1$ while $Rw = \cos \theta_{z,w} e_1 + \sin\theta_{z,w} e_2$. Thus the eigenvectors $d_1$ and $d_2$ can be written as \begin{align*}
    d_1 = Rw - Rz\ \text{and}\ d_2 = Rw + Rz.
\end{align*} Thus the eigenvectors of $M_{z\leftrightarrow w}$ are precisely \begin{align*}
    v_1 = w - z\ \text{and}\ v_2 = w + z.
\end{align*} By the same argument, the eigenvectors of $M_{\Tilde{z} \leftrightarrow \Tilde{w}}$ are \begin{align*}
    \Tilde{v}_1 = \Tilde{w} - \Tilde{z}\ \text{and}\ \tilde{v}_2 = \Tilde{w} + \Tilde{z}
\end{align*} with corresponding eigenvalues $-1$ and $1$, respectively. Hence, we have that \begin{align*}
    \frac{2\sin \theta_{z, w}}{\pi}M_{z \leftrightarrow w} & = \frac{2\sin \theta_{z, w}}{\pi}\left(-v_1v_1^\top + v_2v_2^\top\right) \\
    & = \frac{2\sin \theta_{z, w}}{\pi}\left(-(w-z)(w-z)^\top + (w+z)(w+z)^\top\right)
\end{align*} and likewise \begin{align*}
    \frac{2\sin \theta_{\Tilde{z},\Tilde{w}}}{\pi}M_{\Tilde{z} \leftrightarrow \Tilde{w}} & = \frac{2\sin \theta_{\Tilde{z},\Tilde{w}}}{\pi}\left(-(\Tilde{w} - \Tilde{z})(\Tilde{w}-\Tilde{z})^\top + (\Tilde{w} + \Tilde{z})(\Tilde{w} + \Tilde{z})^\top\right).
\end{align*} For simplicity of notation, let $h = w - z$, $\Tilde{h} = \tilde{w} - \Tilde{z}$, $g = w + z$, and $\Tilde{g} = \tilde{w} + \Tilde{z}$. Then \begin{align*}
    \left\|\frac{2\sin \theta_{z, w}}{\pi}M_{z \leftrightarrow w} -  \frac{2\sin \theta_{\Tilde{z}, \Tilde{w}}}{\pi}M_{\Tilde{z} \leftrightarrow \Tilde{w}}\right\| & = \frac{2}{\pi}\left\|\sin \theta_{z, w}\left(-hh^\top + gg^\top\right) + \sin \theta_{\Tilde{z}, \Tilde{w}}\left(\Tilde{h}\Tilde{h}^\top - \Tilde{g}\Tilde{g}^\top\right) \right\| \\
    & \leqslant \frac{2}{\pi}\left(\| \sin \theta_{z, w}hh^\top - \sin \theta_{\Tilde{z}, \Tilde{w}}\Tilde{h}\Tilde{h}^\top \| + \|\sin \theta_{z, w}gg^\top - \sin \theta_{\Tilde{z}, \Tilde{w}}\tilde{g}\Tilde{g}^\top \|\right).
\end{align*} Note that since $z,w,\Tilde{z},\Tilde{w} \in \mathcal{S}^{n-1}$, $\|h\|,\|\Tilde{h}\|,\|g\|,\|\Tilde{g}\| \leqslant 2$. In addition, \begin{align*}
    \|h-\Tilde{h}\| & \leqslant \|z - \tilde{z}\| + \|w - \Tilde{w}\| \leqslant 2\epsilon
\end{align*} and \eqref{thetabound} implies \begin{align*}
    |\sin\theta_{z,w} - \sin\theta_{\Tilde{z},\Tilde{w}}| \leqslant |\theta_{z,w} - \theta_{\Tilde{z},\Tilde{w}}| \leqslant 4\epsilon.
\end{align*} Hence \begin{align*}
    \|\sin \theta_{z, w}hh^\top - \sin \theta_{\tilde{z}, \tilde{w}}\Tilde{h}\Tilde{h}^\top\| & \leqslant \|\sin \theta_{z, w}hh^\top - \sin \theta_{z, w}h\Tilde{h}^\top\| + \|\sin \theta_{z, w}h\tilde{h}^\top - \sin \theta_{z, w}\tilde{h}\Tilde{h}^\top\| \\
    & + \|\sin \theta_{z, w}\tilde{h}\Tilde{h}^\top - \sin \theta_{\tilde{z}, \tilde{w}}\Tilde{h}\Tilde{h}^\top\| \\
    & \leqslant |\sin\theta_{z,w}|\|h\|\|h-\Tilde{h}\| + |\sin\theta_{z,w}|\|\Tilde{h}\|\|h-\Tilde{h}\| + \|\Tilde{h}\Tilde{h}^\top\||\sin\theta_{z,w}-\sin\theta_{\Tilde{z},\Tilde{w}}| \\
    & \leqslant 20\epsilon.
\end{align*} The same bound holds for $\|\sin\theta_{z,w}gg^\top - \sin\theta_{\Tilde{z},\Tilde{w}} \Tilde{g}\Tilde{g}^\top\|$. Hence we attain \begin{align}
    \left\|\frac{2\sin \theta_{z, w}}{\pi}M_{z \leftrightarrow w} -  \frac{2\sin \theta_{\Tilde{z}, \Tilde{w}}}{\pi}M_{\Tilde{z} \leftrightarrow \Tilde{w}}\right\| & \leqslant \frac{80}{\pi} \epsilon \label{Mbound}.
\end{align} 

Combining \eqref{thetabound} and \eqref{Mbound}, we see that \begin{align*}
    \|\Phi_{\tilde{z},\tilde{w}} - \Phi_{z,w}\| & \leqslant \frac{2|\theta_{\tilde{z},\tilde{w}} - \theta_{z,w}|}{\pi}\|I_n\| + \left\| \frac{2\sin \theta_{\tilde{z},\tilde{w}}}{\pi}M_{\tilde{z} \leftrightarrow \tilde{w}} - \frac{2\sin \theta_{z,w}}{\pi}M_{z \leftrightarrow w}\right\| \leqslant \frac{88}{\pi}\epsilon.
\end{align*}                                                                                        

\end{proof}

We now prove the inequalities used in the proof of Lemma \ref{continuityphi}. \begin{proof}[Proof of equations \eqref{relativeanglebound}, \eqref{sinupperbound}, and \eqref{sinlowerbound}] For \eqref{relativeanglebound}, we proceed similarly to the proof on page $12$ of \cite{Drury2001}. Observe that we can write \begin{align*}
    x_1 & = \cos \theta_{x_1,y} y + \sin \theta_{x_1,y} y_1^{\perp}
\end{align*} and \begin{align*}
    x_2 & = \cos \theta_{x_2,y} y + \sin \theta_{x_2,y} y_2^{\perp}
\end{align*} where $y_1^{\perp}$ and  $y_2^{\perp}$ are unit vectors that are orthogonal to $y$. Then observe that \begin{align*}
    \langle x_1,x_2 \rangle & = \langle \cos \theta_{x_1,y} y + \sin \theta_{x_1,y} y_1^{\perp},\cos \theta_{x_2,y} y + \sin \theta_{x_2,y} y_2^{\perp} \rangle\\
    %& = \cos \theta_{x_1,y} \cos \theta_{x_2,y} + \sin \theta_{x_1,y}\cos \theta_{x_1,y}\langle y,y_1^{\perp}\rangle + \cos\theta_{x_2,y}\sin\theta_{x_2,y}\langle y,y_2^{\perp}\rangle + \sin \theta_{x_1,y}\sin \theta_{x_2,y} \langle y_1^{\perp},y_2^{\perp} \rangle \\
    & =  \cos \theta_{x_1,y} \cos\theta_{x_2,y} + \sin \theta_{x_1,y}\sin \theta_{x_2,y} \langle y_1^{\perp},y_2^{\perp} \rangle.
\end{align*} Since $\theta_{x_1,y},\theta_{x_2,y} \in [0,\pi]$, we have that $\sin \theta_{x_1,y}\sin\theta_{x_2,y} \geqslant 0$. In addition, $\langle y_1^{\perp},y_2^{\perp} \rangle \leqslant \|y_1^{\perp}\|\|y_2^{\perp}\| = 1$ so we attain \begin{align*}
    \langle x_1,x_2 \rangle \leqslant \cos \theta_{x_1,y} \cos\theta_{x_2,y} + \sin \theta_{x_1,y}\sin \theta_{x_2,y} = \cos(\theta_{x_1,y} - \theta_{x_2,y})
\end{align*} by the trigonometric identity $\cos(\al \mp \beta) = \cos\al\cos\beta \pm \sin\al\sin\beta.$ Since the function $\cos^{-1} $ is decreasing on $[-1,1]$, we see that \begin{align*}
    \theta_{x_1,y} - \theta_{x_2,y} \leqslant \cos^{-1}(\langle x_1,x_2 \rangle) = \theta_{x_1,x_2}.
\end{align*} Similarly, $\theta_{x_2,y} - \theta_{x_1,y} \leqslant \theta_{x_1,x_2}$ so we attain $|\theta_{x_1,y} - \theta_{x_2,y}| \leqslant |\theta_{x_1,x_2}|.$

For \eqref{sinupperbound}, observe that for $x,y \in \Seps^{n-1}$, \begin{align*}
    \|x-y\|^2 & = \|x\|^2 + \|y\|^2 - 2 \langle x,y \rangle \\
    & = \|x\|^2 + \|y\|^2 - 2\|x\|\|y\|\cos\theta_{x,y} \\
    & = 2(1 - \cos \theta_{x,y}).
\end{align*} Thus, using the half angle formula \begin{align*}
    \sin \frac{\theta}{2} = \sgn\left(2\pi - \theta + 4\pi \left\lfloor \frac{\theta}{4\pi}\right\rfloor\right)\sqrt{\frac{1 - \cos \theta}{2}}
\end{align*} we see that \begin{align*}
    \|x-y\| = \sqrt{2(1- \cos\theta_{x,y})} = 2\sqrt{\frac{1 - \cos \theta_{x,y}}{2}} \geqslant 2\sin\frac{\theta_{x,y}}{2}.
\end{align*} 

For \eqref{sinlowerbound}, one can note that the function $\psi(\theta) : = 4\sin\frac{\theta}{2} - \theta$ is positive for all $\theta \in [0,\pi].$ 
\end{proof}

 \section*{Acknowledgements}
  PH is partially supported by NSF CAREER Grant DMS-1848087. OL acknowledges support by the NSF Graduate Research Fellowship under Grant No. DGE-1450681.

\bibliographystyle{plain}

\bibliography{dpr.bib}

\end{document}